\documentclass[12pt]{article}
\pdfoutput=1
\usepackage[utf8]{inputenc}

\usepackage{xcolor}
\usepackage{amsmath,amssymb,amsthm,mathrsfs,bm,setspace,xspace}
\usepackage{upgreek}
\usepackage{graphicx}
\usepackage{float}
\usepackage{booktabs}
\usepackage{multirow}
\usepackage{enumitem}
\usepackage{geometry}

\usepackage{placeins}
\definecolor{darkblue}{rgb}{0.1,0.1,.7}
\usepackage[colorlinks, linkcolor=darkblue, citecolor=darkblue, urlcolor=darkblue, linktocpage]{hyperref}

\definecolor{martenorange}{rgb}{0.91, 0.41, 0.17}

\newcommand{\Ds}{\Delta_\sigma}
\newcommand{\De}{\Delta_\epsilon}
\newcommand{\Dep}{\Delta_{\epsilon'}}
\newcommand{\Zt}{\mathbb{Z}_2}

\newcommand{\Q}{\mathbb{Q}}
\newcommand{\Z}{\mathbb{Z}}
\newcommand{\half}{\tfrac12}
\newcommand{\rstar}{r_\ast}
\newcommand{\code}[1]{\texttt{#1}}
\newcommand{\Vvec}{\vec V}
\newcommand{\What}{\widehat W}
\newcommand{\slo}{\sigma_-}
\newcommand{\shi}{\sigma_+}
\newcommand{\elo}{\epsilon_-}
\newcommand{\ehi}{\epsilon_+}
\newcommand{\Isig}{I_\sigma}
\newcommand{\Ieps}{I_\epsilon}
\newcommand{\tausplit}{\tau_{\rm split}}

\newcommand{\levstar}{\mathrm{lev}_\star}
\newcommand{\mgeo}{m_{\rm geo}}

\theoremstyle{plain}
\newtheorem{theorem}{Theorem}
\newtheorem{lemma}{Lemma}
\newtheorem{proposition}{Proposition}
\theoremstyle{definition}

\numberwithin{equation}{section}
\usepackage[most]{tcolorbox}

\newtcolorbox{flagshipresult}{
	enhanced,
	breakable,
	colback=green!8,
	colframe=green!40!black,
	boxrule=0.5pt,
	arc=1.5pt,
	left=7pt,
	right=7pt,
	top=5pt,
	bottom=5pt,
	before skip=7pt,
	after skip=7pt
}

\newtcolorbox{flagshipresultwithtitle}{
	enhanced,
	breakable,
	colback=green!8,
	colframe=green!40!black,
	boxrule=0.5pt,
	arc=1.5pt,
	left=7pt,
	right=7pt,
	top=5pt,
	bottom=5pt,
	before skip=7pt,
	after skip=7pt,
	fonttitle=\bfseries,
	title=Certified result
}

\definecolor{proofochre}{RGB}{194,157,0}
\tcolorboxenvironment{theorem}{
	 enhanced, breakable, frame hidden, colback=white,
	 borderline west={1.5pt}{0pt}{proofochre},
	 borderline east={1.5pt}{0pt}{proofochre},
	 boxsep=0pt, left=10pt, right=10pt, top=0pt, bottom=0pt,
	 before skip=8pt, after skip=0pt
}
\tcolorboxenvironment{lemma}{
	 enhanced, breakable, frame hidden, colback=white,
	 borderline west={1.5pt}{0pt}{proofochre},
	 borderline east={1.5pt}{0pt}{proofochre},
	 boxsep=0pt, left=10pt, right=10pt, top=0pt, bottom=0pt,
	 before skip=8pt, after skip=0pt
}
\tcolorboxenvironment{proposition}{
	 enhanced, breakable, frame hidden, colback=white,
	 borderline west={1.5pt}{0pt}{proofochre},
	 borderline east={1.5pt}{0pt}{proofochre},
	 boxsep=0pt, left=10pt, right=10pt, top=0pt, bottom=0pt,
	 before skip=8pt, after skip=0pt
}
\tcolorboxenvironment{proof}{
	 enhanced, breakable, frame hidden, colback=white,
	 borderline west={1.5pt}{0pt}{proofochre},
	 borderline east={1.5pt}{0pt}{proofochre},
	 boxsep=0pt, left=10pt, right=10pt, top=0pt, bottom=0pt,
	 before skip=0pt, after skip=8pt
}

\hypersetup{pdftitle={Region bounds, error margins and certified positivity in the numerical conformal bootstrap; or the first mathematically rigorous bounds on the 3d Ising CFT critical exponents},pdfauthor={Marten Reehorst}}
\begin{document}
\vspace*{-1.6in} \thispagestyle{empty}
\vspace{1cm}
{\Large
\begin{center}
{Region bounds, error margins and certified positivity in the numerical conformal bootstrap; \\
	Or the first mathematically rigorous bounds on the 3d Ising CFT critical exponents.}
\end{center}
}
\vspace{0.1cm}
\begin{center}
{\bf Marten Reehorst}\\[0.1cm]
{\it Department of Mathematics, King's College London, Strand, London WC2R 2LS, United Kingdom}
\vspace{0.1cm}
\end{center}

\begin{abstract}
\noindent Numerical conformal bootstrap bounds are commonly described as
``rigorous'', yet as published they are not theorems: they hold at sampled
points rather than regions of parameter space, for a truncated rational
approximation of the conformal blocks rather than the blocks themselves, for a
finite list of spins rather than all of them, and up to floating-point rounding.
We address all four gaps. 

We show: how to efficiently search for functionals positive on \emph{regions} including external-dimension space; how positivity on the large-spin crossing vectors can be imposed as such a region constraint; how new analytic bounds about conformal block recursion relations allow us to efficiently find functionals that are positive on the true untruncated underlying crossing vectors; and how to efficiently mathematically certify positivity of the found functionals using exact arithmetic. The analytic inputs needed to run the same program on the mixed $\sigma$--$\epsilon$ system, which would remove the spin-2 gap assumption, are proved as well, and the finite list of certificates that remain to be run for it is stated.

As a flagship application we certify an exclusion region in the
$(\Ds,\De)$ plane: for $\Ds\in[0.505,0.5405]$, the entire range
$\De\in(\tfrac12,3]$ is excluded except a small rectangle containing
the 3d Ising island. This rigorously proves two-sided bounds on these quantities --- or equivalently the critical exponents $\eta$ and $\nu$ --- under explicitly spelled-out spectrum assumptions believed to identify this CFT.

The exact certification procedure is based on: Bernstein-basis positivity
certificates in up to three variables; exact $\Q(\sqrt2)$ arithmetic; a
compactified continuum treatment of the infinite spin tail; and an explicit
nonnegative upper bound on the difference between the true conformal blocks and
the finite truncations used in the numerics. Every finite certificate check is exact integer
arithmetic, and the infinitely many positivity conditions (all spins, all
twists, a continuum of external dimensions, the true blocks) reduce to finitely
many machine-checked sign conditions. The essential code to reproduce the
search and the certification, and all machine-checkable certificates, are
made available.

Finally, we emphasize the innovations put forward in this paper are not only useful to formalize bounds. They can prevent practical failure modes that are routinely encountered in the numerical conformal bootstrap, and they can (perhaps surprisingly) offer \textbf{faster} numerics. A prime example of this is how region bounds can turn an expensive mapping of an island to obtain bounds into a (handful of) 1-dimensional bisections. This changes the optimality class of the search to a much more favorable one.

\end{abstract}

\newpage
{\small \parskip=-0.1em \tableofcontents}
\newpage

\section{Introduction and summary}
\label{sec:intro}

The numerical conformal bootstrap \cite{Rattazzi:2008pe} has produced the most
precise determinations of the 3d Ising critical exponents
\cite{ElShowk:2012ht,El-Showk:2014dwa,Kos:2014bka,Kos:2016ysd}, and its results
are routinely --- and rightly --- described as \emph{bounds} rather than
estimates: a functional that acts positive on all operators in an assumed spectrum rules out that CFT. 

However, a gap between the methods used in the numerical conformal bootstrap literature and a mathematical theorem remains. The gap consists of four specific loopholes, present in essentially every published bootstrap bound\footnote{One further caveat exists, with an easy and purely analytic resolution. Functionals	cannot establish an exclusion exactly \emph{at} a unitarity bound when there is a pole there. The standard denominator clearing fails there. Positivity of the ``cleared'' polynomial $P(\Delta)$ does not imply positivity of the true underlying block structure $V(\Delta)$. So exclusion windows reaching such a point are open at that endpoint. The resolution is that unitarity forces an operator sitting exactly at the bound to decouple from the relevant OPE. Lemma~\ref{lem:corner} formally proves the exclusion of this endpoint given a functional that is positive above it.}:
\begin{enumerate}[nosep]
\item \textbf{Points versus regions.} A semidefinite-program (SDP) run excludes a single point of
  parameter space (e.g.\ one pair $(\Ds,\De)$). Two-sided bounds are then
  inferred by sampling many points and interpolating the shape of the allowed
  region (by eye or some other heuristic method). Strictly, nothing forbids a thin protrusion of the allowed
  region between the sampled points, or an undetected second isolated island. Nor is it clear where outside the sampled allowed point the boundary should be drawn. These potential failure modes are illustrated for clarity in figure \ref{fig:pointsamplingfailurethreemodes}.
\item \textbf{Finitely many versus all spins.} Positivity is imposed for spins
  $\ell\le\ell_{\max}$ (and often checked at a few 	larger ones). The crossing
  equation contains operators of unbounded spin; for any cutoff $\ell_{\max}$ there
  are physical operators with $\ell>\ell_{\max}$. A functional that fails to be positive at any spin does not provide any bound at all.
\item \textbf{Approximate versus true blocks.} The solver never sees conformal
  blocks; it sees a rational approximation obtained by truncating a recursion at
  some order $\kappa$. The functional is positive on the approximation; the
  physical claim needs positivity on the blocks themselves.
\item \textbf{Floating point versus exact arithmetic.} The functional is a
  vector of floating-point numbers produced by an interior-point solver; every
  inequality it is supposed to satisfy has been checked only up to rounding. Moreover, the functional search is usually set up to terminate when positivity is just barely achieved (without any margin). This means that tiny rounding mistakes can flip the sign of such conditions invalidating the result.
\end{enumerate}

\begin{figure}[h]
	\centering
	\includegraphics[width=0.7\linewidth]{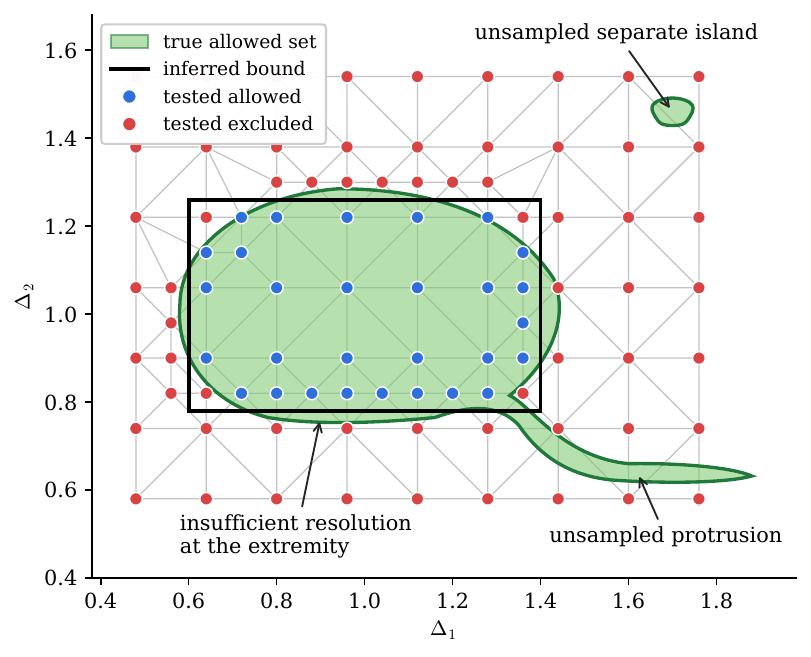}
	\caption{A hypothetical bootstrap bound illustrating failure modes of point-based bounds. Allowed/disallowed points are shown in blue/red. The true underlying allowed shape is shown in green. The figure shows how sampling at a finite number of points can miss isolated allowed islands as well as thin protrusions that stick out in between the sampled points. This could result in wildly inaccurate bounds if the true CFT is located inside the unsampled island or an unsampled protrusion! Finally, the tested points must be interpreted as a bound and there is no clear way to do this; it is unclear how far outside the found allowed points the boundary should be drawn (see the green region spilling outside the hypothetical assumed bound given by the black rectangle). This results in a finite generally un-estimated error on the reported bounds.}
	\label{fig:pointsamplingfailurethreemodes}
\end{figure}

Generally, a careful study consisting of: a thorough mapping of the island's shape; stability tests under changes of the block truncation parameters and spin-sets; and using sufficient safety margins in the truncation parameters, is believed to ensure that the obtained bounds represent the true bound up to the number of digits presented.

On the other hand it is fairly common in (lower-compute) early-stage exploratory work for researchers to be confounded by unexpected results that turn out to be due to one of the failure modes described above.\footnote{The author can attest to such things occurring in their own research and has heard much anecdotal evidence from fellow numerical conformal bootstrap practitioners.} Sometimes these mistakes even survive in published work. Famously, \cite{Kos:2016ysd} reported too small an island, and hence bounds that were later seen to be unjustified, due to an insufficiently resolved charting of the island's shape. This is illustrated in Figure \ref{fig:annotatedsigmaepsilonboundsonly}. The point here is not to highlight this specific mistake but to show that such mistakes are easy to make when using point bounds. Coefficient bounds are inherently region-bound questions and thus should be addressed by region-bound methods, which we introduce in this paper. 

\begin{figure}
	\centering
	\begin{minipage}[t]{0.48\linewidth}
		\centering
		\includegraphics[width=\linewidth]{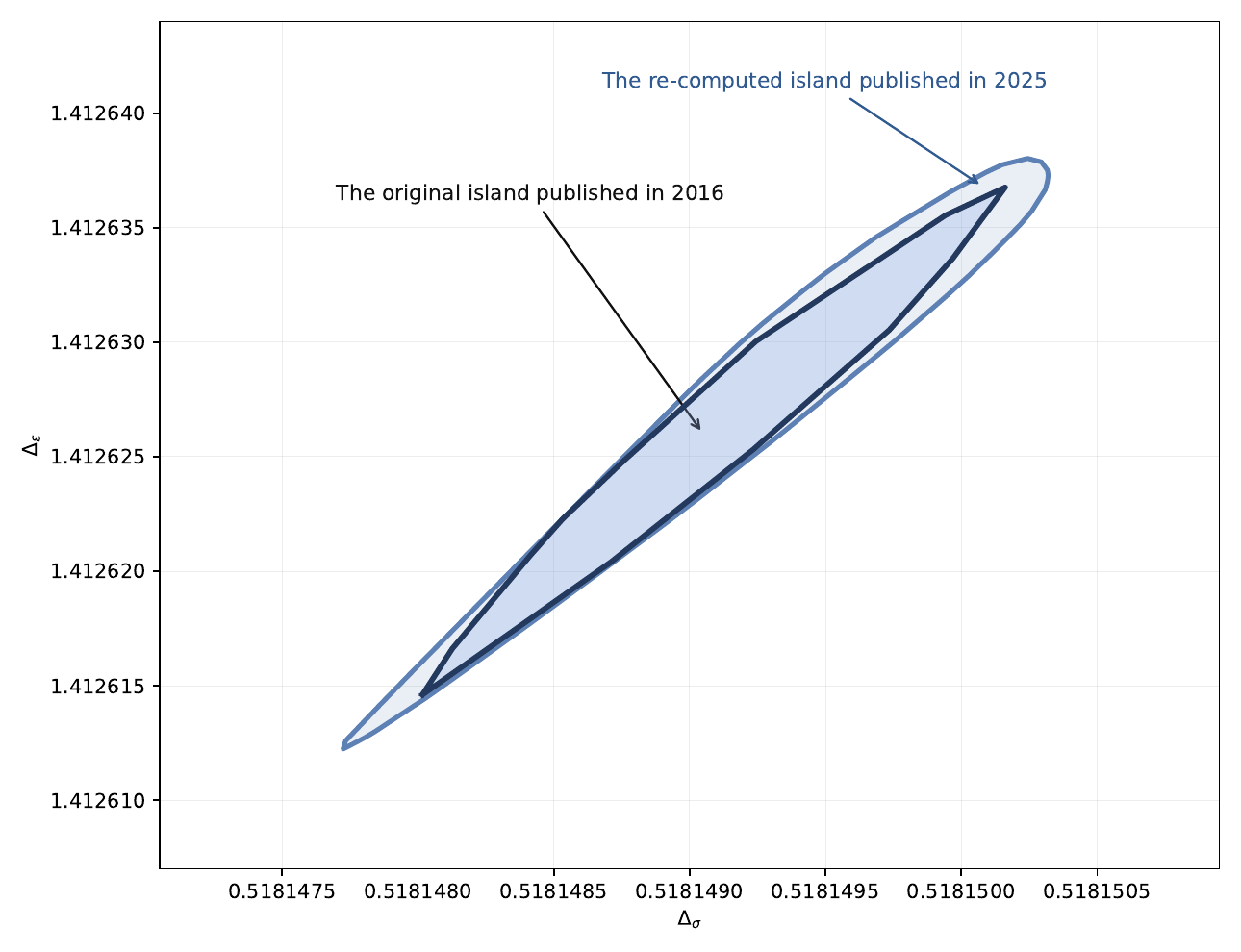}
	\end{minipage}\hfill
	\begin{minipage}[t]{0.48\linewidth}
		\centering
		\includegraphics[width=\linewidth]{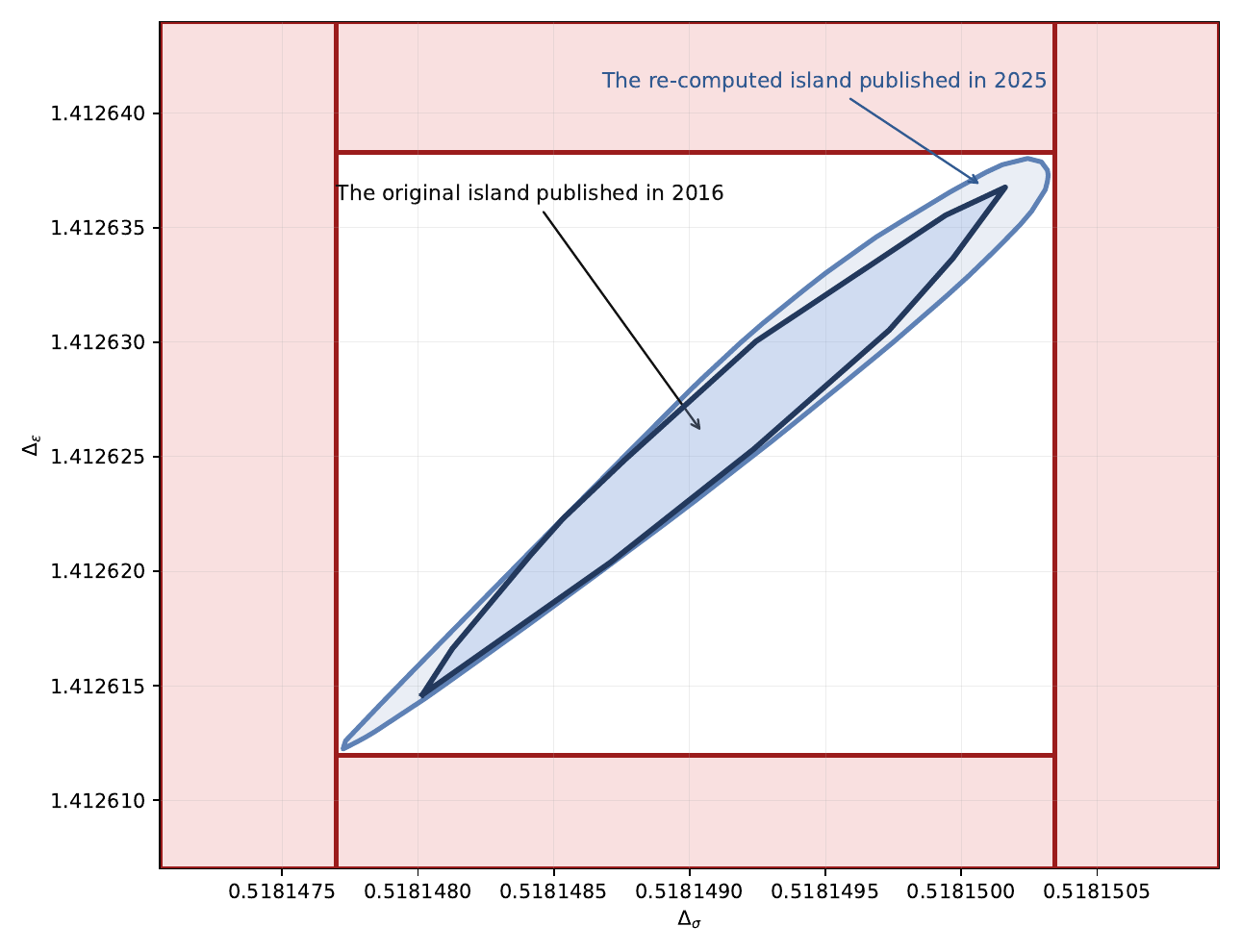}
	\end{minipage}
		\caption{Numerical conformal bootstrap bounds obtained in \cite{Kos:2016ysd} (light-blue) compared against a version re-computed in \cite{Chang:2024whx} using more compute and newer methods (darker blue). This illustrates an explicit example where the conventional point-based `charting' approach led to bounds that were later found to be too strong. The figure on the right illustrates how region bounds can be used to prevent such mistakes in the future: A fully closing ring on all sides would certify bounds regardless of the shape of the island (or islands!).}
	\label{fig:annotatedsigmaepsilonboundsonly}
\end{figure}

To turn bounds into a mathematical theorem we divide the process into two parts. First we search for a candidate functional that we think has the right properties and then we certify rigorously that it indeed has these properties (or possibly certify that it does not prompting a new search).

We introduce multiple new methods and improvements in both the functional search and the certification steps that were essential in formulating our proof.

Starting with the search, the following two methods were essential to finding functionals that allowed the required certifications:
\begin{itemize}
	\item We maximize the margin by which positivity is achieved instead of demanding only that positivity is (barely) achieved.\footnote{This margin can also be used as a navigator function \cite{Reehorst:2021ykw}.}
	\item The addition of constraints that impose the required positivity in the non-compact directions such as at infinite $\ell$ (as well as at very large but finite spin: $\ell\sim 10^6$). 
\end{itemize}
Without these additions, every margin-free functional we searched for failed certification: at some large but finite spin, or on the $\ell=\infty$ face, positivity was violated by a small amount. Our experience shows that a functional found without a margin and without explicit control of the non-compact directions has no reason to be certifiable as it stands. We speculate, although we have not checked, that most functionals used in the existing numerical bootstrap literature therefore do not allow certification. \footnote{We do expect a certifiable functional to lie very close to such a functional, so that the bounds themselves would practically not change at the precision at which they were reported.}

These improvements assured that all candidates found in our search could be successfully certified by our certification pipeline. This pipeline produces a machine-checked \emph{certificate} certifying the functional obeys the positivity we require. The key ideas of the certification pipeline are:
\begin{itemize}[nosep]
\item a strict separation of \emph{search} and \emph{certificate}
  (Section~\ref{sec:searchcert}): all floating-point machinery (semidefinite- and linear-program (SDP/LP) solvers,
  grids, heuristics) serves only to propose candidates; soundness rests only on
  the certificate layer, so no property of the search needs to be trusted;
\item exact arithmetic over $\Q(\sqrt2)$ (Section~\ref{sec:qsqrt2}): the one
 irrationality of the problem, $\sqrt2$, enters through the radial coordinate
 of the crossing-symmetric point, $\rstar=3-2\sqrt2$; representing every
 quantity as $a+b\sqrt2$ with rational $a,b$ makes all sign tests exact
 rational comparisons;
\item \emph{Bernstein-basis positivity certificates}
  (Section~\ref{sec:bernstein}): a classical way to make
  the positivity of a polynomial on a box \emph{visible} in its coefficients,
  with exact-arithmetic conversion, built-in counterexample detection, and a
  subdivision loop that terminates on any strictly positive polynomial.\footnote{We also initially pursued certification through other methods such as interval arithmetic and numerical sums-of-squares followed by rationalization. Our attempts using these methods failed where the Bernstein-basis succeeded.};
\item the entire spin tail as \emph{one} certificate
  (Section~\ref{sec:allspins}): the truncated blocks are rational functions of
  a \emph{real} spin variable (we explain why the pole recursion's spin-dependent
  selection rules do not spoil this), so positivity for all real is just one additional direction in which we need to certify positivity;
\item a \emph{truncation-error envelope} for the true blocks
  (Section~\ref{sec:trueblocks}): subtracting the certified majorant from the computable object proves positivity of the uncomputable true-block value;
\item the region direction as ``one more continuous parameter direction'' for our Bernstein certificate  (Section~\ref{sec:regioncert}).
\end{itemize}
The \emph{trusted computing base} of the resulting certificates --- the list of statements,
cited theorems and programs that a reader must accept without the
certificates themselves checking them --- is small and explicitly listed
(Section~\ref{sec:trustedbase}): exact arithmetic, Sturm's theorem, the
two-line Bernstein positivity implication, and three analytic facts about
conformal blocks --- the convergence of the operator product expansion
\cite{Pappadopulo:2012jk}, the positivity and convergence of the radial
expansion \cite{Hogervorst:2013sma}, and the structure of the pole recursion
\cite{Kos:2013tga,Kos:2014bka,Penedones:2015aga} --- together with the
CFT axioms, a few formulas entered by transcription, and the programs
that run the checks. The remaining analytic input, a bound on the
blocks' diagonal radial coefficients, is not trusted but proved in
Section~\ref{sec:telescope}, from a zero-sum form of their recurrence.

\paragraph{Flagship application.} We apply the pipeline to the single-correlator
$\langle\sigma\sigma\sigma\sigma\rangle$ bootstrap to show:

\begin{flagshipresult}
Any unitary $\Zt$ symmetric 3d CFT with
\begin{enumerate}[nosep]
	\item Exactly one relevant $\Zt$-odd scalar $\sigma$, with $\Delta_{\sigma} \in [0.505,0.5405]$;
	\item At most one relevant $\Zt$-even scalar;
	\item A stress tensor at $\Delta=3$;
	\item A gap until the next $\Zt$-even spin-2 operator $\Delta_{T'}\ge11/2$;
\end{enumerate}
must contain \emph{exactly} one relevant $\Zt$-even scalar $\epsilon$, with
\[(\Delta_{\sigma},\De)\ \in\ (0.5055,\,0.540)\times(1.23,\,1.55).\]
In particular, in terms of the critical
exponents ($\eta=2\Ds-1$, $\nu=1/(3-\De)$):
\[\eta\ \in\ (0.011,\ 0.080)\qquad\text{and}\qquad
  \nu\ \in\ \big(\tfrac{100}{177},\ \tfrac{20}{29}\big)
  \ \approx\ (0.5650,\ 0.6897).\]
\end{flagshipresult}
Of course it is generally believed that this class of CFTs contains
exactly one CFT, the \textit{3d Ising CFT}. That uniqueness is not
proven; the belief rests on renormalization-group universality --- the
Ising, $\varphi^4$ and liquid--vapor critical points are all described
by a single fixed point. Even the conformal (rather than merely scale) invariance of the critical 3d Ising model is strictly a conjecture, though one with sharp 
theoretical support and dedicated precision tests \cite{Delamotte:2015aaz,Meneses:2018xpu}. None of this is consumed by our theorem, which quantifies over the class: the bounds hold for
every CFT satisfying the assumptions, whether that is one theory or
several.

\paragraph{Relation to previous work.} The bootstrap community has long
been aware of these loopholes, and partial treatments exist.
Certification over \emph{regions} appears in \cite{Chang:2024whx}
(Section~4.2 there) for OPE-coefficient space only. The crossing
equations depend on these parameters only quadratically, so
that region certificates are easy and inexpensive; the continuum directions
treated here --- external dimensions, all twists and spins, and the
true blocks --- are not addressed by that construction. Exact
positivity proofs for \emph{rationalized functionals} were brought to our field,
in the spectral-geometry bootstrap, by Kravchuk,
Maz\'a\v c and Pal \cite{Kravchuk:2021akc}: a numerically found (SDPB)
functional is replaced by a nearby rational one, evaluated on the
crossing vectors in exact arithmetic, and the resulting rational
polynomials are proven nonnegative on the required half-lines by an
exact recursive algorithm (their Appendix~B; a M\"obius-substitution
and coefficient-sign scheme of the same family as the Bernstein
certificates we run at scale here). To our knowledge that was the
first bootstrap-style computation whose final positivity certificate
was certified entirely in exact rational arithmetic; the same
machinery has since been reused in the bootstrap of Euclidean
lattices \cite{Bertucci:2026eul}, with the exact rational functionals
shipped as ancillary files. In mathematics this
discover-numerically-then-certify-exactly pattern has a longer
history: sphere-packing and lattice-optimality proofs routinely
convert numerical LP/SDP solutions into exact certificates by
rational reconstruction \cite{CohnKumar:2009leech,Dostert:2021exact}.
Our block approximation error bounds build on previous work on the convergence of the block
expansions in
\cite{Pappadopulo:2012jk,Hogervorst:2013sma,Rychkov:2015lca}.

Our contributions are: introducing the Bernstein basis as an efficient certification tool, elevating point bounds to region bounds, new methods for optimizing margins, new methods for ensuring positivity on the spin tail, new ways to certify positivity on the true underlying block structure and not just the rational approximation, an analytic decoupling argument that closes exclusion windows at the unitarity bound itself, certified coefficient majorants and exact external polynomiality for the mixed-correlator system that put its rigorous treatment within reach (Section~\ref{sec:mixedoutlook}), and the first pipeline that closes all four loopholes \emph{simultaneously} in order to produce the first theorem-grade bounds on the 3d Ising CFT.

\paragraph{Structure of the paper.}
Section~\ref{sec:setup} describes the bootstrap problem: the crossing
equation, the functional space, and the pole recursion.
Section~\ref{sec:reduction} is the logical core of the paper: it states
the two-part target --- (A) a short list of analytic inputs, (B)
finitely many polynomial sign conditions with exact $\Q(\sqrt2)$
coefficients, decidable by finite computation --- and derives, loophole
by loophole, how every claim of the flagship theorem takes this form.
The two sections after it solve the conditions (B):
Section~\ref{sec:regionpractice} describes the search for candidate
functionals (the layer of which nothing needs to be trusted), and
Section~\ref{sec:certarch} the certificates that decide the conditions
in exact arithmetic. Section~\ref{sec:trueblocks} constructs the one
object the reduction took as given --- the certified truncation
envelope --- and proves the coefficient bounds behind it.
Section~\ref{sec:trustedbase} then lists, exhaustively, everything a
reader must trust for the final theorem, and
Section~\ref{sec:rigresults} assembles the certified exclusions into the
bounds on the 3d Ising critical exponents.
Section~\ref{sec:experiments} discusses how the methods we developed can be useful even to the practical bootstrapper who is not interested in the additional rigor they help achieve, and
Section~\ref{sec:mixedoutlook} proves and machine-checks the analytic inputs for the extension to mixed correlators, paving the way towards an Ising critical exponent theorem that needs no assumption about the spin-2 sector.

The appendices collect the block conventions and the pole recursion
(App.~\ref{app:blocks}); the trade-offs among the search reductions
(App.~\ref{app:search}); the implementation of
the Bernstein certifier (App.~\ref{app:bernstein}); the closed form of
the diagonal block (App.~\ref{app:closedform}); the list of exclusion
tiles (App.~\ref{app:tiles}); the unitarity-bound endpoint lemma
(App.~\ref{app:endpoint}); the bookkeeping proofs for the
mixed-correlator extension (App.~\ref{app:mixedproofs}); and the data
release, with verification commands and timings
(App.~\ref{app:release}).

\section{Setup: the single-correlator Ising bootstrap}
\label{sec:setup}

We work in the simplest setting that already isolates the 3d Ising model in an island: the
four-point function of a $\Zt$-odd scalar $\sigma$ in a unitary 3d
CFT. Crossing symmetry of
$\langle\sigma\sigma\sigma\sigma\rangle$ reads
\begin{equation}
  \sum_{\mathcal O\in\sigma\times\sigma} p_{\mathcal O}\,
  F_{\Delta_{\mathcal O},\ell_{\mathcal O}}(u,v;\Ds)=0,\qquad
  F_{\Delta,\ell}=v^{\Ds}g_{\Delta,\ell}(u,v)-u^{\Ds}g_{\Delta,\ell}(v,u),
  \label{eq:crossing}
\end{equation}
with $p_{\mathcal O}=\lambda_{\sigma\sigma\mathcal O}^2\ge0$ with the identity
separated out as a term with a known OPE coefficient (set to 1 by convention). Acting with the standard derivative functionals at the crossing-symmetric point $z=\bar z=\half$,
\begin{equation}
  \alpha=\sum_{\substack{m+n\ \mathrm{odd},\ m+n\le\Lambda\\ m>n}} \alpha_{(m,n)}\,
  \partial_z^m\partial_{\bar z}^n\big|_{z=\bar z=1/2}.
\end{equation}
Components with $m\leftrightarrow n$ coincide by the $z\leftrightarrow\bar z$
symmetry of $F$, so we take $m>n$; this leaves $N_\Lambda$ independent
components, e.g. $N_9=15$. This turns \eqref{eq:crossing} into the vector equation
$-\vec{\mathrm{Id}}(\Ds)=\sum_{\mathcal O}p_{\mathcal O}\Vvec(\Ds;\Delta,\ell)$,
where $\Vvec$ collects the derivative components of $F$. If a single
$\alpha$ satisfies
\begin{equation}
  \alpha\cdot\vec{\mathrm{Id}}(\Ds)>0,\qquad
  \alpha\cdot\Vvec(\Ds;\Delta,\ell)\ \ge\ 0\ \ \text{on the assumed spectrum},
  \label{eq:funcconds}
\end{equation}
then no CFT with that spectrum exists by contradiction: applying $\alpha$ termwise to the
(componentwise absolutely convergent \cite{Pappadopulo:2012jk}, see
Section~\ref{sec:trustedbase}) sum \eqref{eq:crossing} would yield a negative number
as a sum of nonnegative ones.

The spectrum assumptions that enter our setup:
\begin{enumerate}[nosep,label=(\alph*)]
	\item the CFT contains a single relevant $\Zt$-odd scalar $\sigma$, of dimension $\Ds$;
	\item it contains \emph{at most one} relevant $\Zt$-even scalar
	(relevant means $\Delta<3$); when one is present we call it
	$\epsilon$, of dimension $\De$. Note that existence is not assumed:
	that a relevant even scalar \emph{must} be present is part of what
	the certificates prove;
	\item there is a conserved stress tensor, a $\Zt$-even spin-2 operator at $\Delta=3$;
	\item the next $\Zt$-even spin-2 operator has $\Delta_{T'}\ge11/2$.\footnote{For
	the 3d Ising CFT the best available determinations are
	$\Delta_{T'}=5.50915(44)$ from the extremal-functional spectrum of
	\cite{Simmons-Duffin:2016wlq}, whose error estimate is not rigorous,
	and $\Delta_{T'}=5.499(17)$ from the navigator bounds of
	\cite{Reehorst:2021hmp}.  The first lies above $\tfrac{11}2$ by
	$9\times10^{-3}$; the second interval contains $\tfrac{11}2$.
	Assumption (d) is therefore close to the true value, and it is the
	assumption that the mixed-correlator extension of
	Section~\ref{sec:mixedoutlook} is meant to remove.}
\end{enumerate}
Everything this paper proves is a statement about CFT scaling
dimensions under the assumptions above. The translation to statistical
physics is the standard dictionary: \emph{for a critical system whose
scaling limit is described by a CFT satisfying (a)--(d)} --- as the 3d
Ising universality class is universally believed to be --- the leading
thermodynamic critical exponents are determined by the two scalar
dimensions through the scaling relations
\begin{equation}
  \eta=2\Ds-1,\qquad \nu=\frac1{3-\De},
  \label{eq:exponents}
\end{equation}
so two-sided bounds on $(\Ds,\De)$ become two-sided bounds on
$(\eta,\nu)$. The dictionary itself (existence and conformality of the
scaling limit) is an input from outside the bootstrap, and we keep it
separate from the theorems: dimension bounds are what is proved,
exponent bounds are their translation under it. The uniqueness parts of (a) and (b) are not consumed by
any solver --- the linear programs below never see them --- but they are
essential premises for reading the certified exclusions as bounds on
\emph{dimensions}; we return to this logical point in
Section~\ref{sec:ring}.

\paragraph{Blocks and their approximation.} The conformal blocks
$g_{\Delta,\ell}$ are computed, as in all modern bootstrap numerics, from the recursion
relation over the poles of the block in $\Delta$ --- the \emph{pole
recursion}, Zamolodchikov-type
\cite{Kos:2013tga,Kos:2014bka,Penedones:2015aga}, the ``recursion relation
from analytic structure'' of \cite[Sec.~III.F.4]{Poland:2018epd}: in radial coordinates
\cite{Hogervorst:2013sma} $(r,\upeta)$ --- we write $\upeta$ for the standard radial angle to reserve $\eta$ for the critical exponent --- we use the stripped block
\begin{equation}
  h_{\Delta,\ell}(r,\upeta):=(4r)^{-\Delta}g_{\Delta,\ell}(r,\upeta).
  \label{eq:strippedblock}
\end{equation}
This stripped block satisfies a recursion whose truncation at pole level
$\kappa$ yields a closed-form \emph{rational} function
$h^{(\kappa)}$ of $\Delta$ --- and, as we exploit heavily, of $\ell$ as well
(Section~\ref{sec:allspins}). The \emph{truncation order} $\kappa$ is
the one free precision parameter on the block side of the whole
certification. For our certificates we used $\kappa_{\rm cert}=28$, a choice explained in
Section~\ref{sec:marginproof} (all numerical parameter choices of the
proof are collected in Table~\ref{tab:params} of Section~\ref{sec:trueblocks}). Two properties of the
radial frame matter throughout:
\begin{itemize}[nosep]
\item at the crossing-symmetric point the radial coordinate is
  \begin{equation}
    \rstar\;=\;\rho(\tfrac12)\;=\;\frac{1/2}{(1+\sqrt{1/2})^2}\;=\;3-2\sqrt2
    \;<\;\tfrac15,
    \label{eq:rstar}
  \end{equation}
  which is why $\sqrt2$ is the single irrationality of the entire problem
  (Section~\ref{sec:qsqrt2}) and why each further radial order of any expansion
  is suppressed by $\rstar$;
\item the radial expansion of $h_{\Delta,\ell}$ has nonnegative coefficients
  for unitary $(\Delta,\ell)$ \cite{Hogervorst:2013sma}, the positivity that
  ultimately powers every tail estimate in this paper.
\end{itemize}
Our conventions, the explicit recursion, and the derivative ``crossing vector''
tables are collected in App.~\ref{app:blocks}.

\paragraph{The radial Casimir recurrence.} One more object plays a central role in this paper: a recurrence in the radial level for the diagonal block. On the
diagonal $z=\bar z$, write the stripped block as a radial power series,
\begin{equation}
  H(\rho)\;:=\;h_{\Delta,\ell}(\rho,1)\;=\;\sum_{m\ge0}
  c_m(\Delta,\ell)\,\rho^{m},
  \qquad c_0=1 .
  \label{eq:Hdiag}
\end{equation}
Diagonal blocks satisfy the fourth-order ordinary differential equation
of Hogervorst, Osborn and Rychkov \cite{Hogervorst:2013kva}; rewriting
that ODE in the radial variable through $z=4\rho/(1+\rho)^2$ and the
stripping \eqref{eq:strippedblock}, and clearing denominators, turns it
into what we call the \emph{radial Casimir recurrence} for the
coefficients --- an eight-term recurrence (order seven), the diagonal
analogue of the Casimir recursion relations of
\cite[Sec.~III.F.3]{Poland:2018epd}; it relates the
eight consecutive coefficients $c_m,\dots,c_{m+7}$:
\begin{equation}
  \sum_{k=0}^{7}\tilde A_k(m;\Delta,\ell)\,c_{m+k}\;=\;0
  \qquad\text{for all }m\ge-7\quad(c_j:=0\text{ for }j<0),
  \label{eq:seventerm}
\end{equation}
with the eight coefficients the explicit polynomials
\begingroup\footnotesize
\begin{align}
  \tilde A_{0} &= m^2 \left(-5 \Delta^2-21 \Delta+\ell^2+\ell-11\right)-(2 \Delta+3) m \left(\Delta^2+9 \Delta-\ell^2-\ell+2\right)\nonumber\\
  &\qquad -6 \Delta (\Delta-\ell+1) (\Delta+\ell+2) -2 (2 \Delta+3) m^3-m^4 \nonumber\\
  \tilde A_{1} &= m^2 \left(-5 \Delta^2-33 \Delta+\ell^2+\ell-35\right)-(2 \Delta+5) m \left(\Delta^2+13 \Delta-\ell^2-\ell+10\right)\nonumber\\
  &\qquad -4 (2 \Delta+1) (\Delta-\ell+2) (\Delta+\ell+3) -2 (2 \Delta+5) m^3-m^4 \nonumber\\
  \tilde A_{2} &= 3 m^2 \left(5 \Delta^2+33 \Delta-\ell^2-\ell+39\right)+3 (2 \Delta+5) m \left(\Delta^2+13 \Delta-\ell^2-\ell+14\right)\nonumber\\
  &\qquad +2 \left(12 \Delta^3+71 \Delta^2-12 \Delta \ell^2-12 \Delta \ell+135 \Delta-7 \ell^2-7 \ell+72\right)+6 (2 \Delta+5) m^3+3 m^4 \nonumber\\
  \tilde A_{3} &= 3 m^2 \left(5 \Delta^2+45 \Delta-\ell^2-\ell+75\right)+3 (2 \Delta+7) m \left(\Delta^2+17 \Delta-\ell^2-\ell+26\right)\nonumber\\
  &\qquad +2 \left(15 \Delta^3+125 \Delta^2-15 \Delta \ell^2-15 \Delta \ell+330 \Delta-16 \ell^2-16 \ell+252\right)+6 (2 \Delta+7) m^3+3 m^4 \nonumber\\
  \tilde A_{4} &= -\tilde A_{3} \nonumber\\
  \tilde A_{5} &= -3 m^2 \left(5 \Delta^2+57 \Delta-\ell^2-\ell+123\right)-3 (2 \Delta+9) m \left(\Delta^2+21 \Delta-\ell^2-\ell+42\right)\nonumber\\
  &\qquad -4 \left(9 \Delta^3+97 \Delta^2-9 \Delta \ell^2-9 \Delta \ell+330 \Delta-14 \ell^2-14 \ell+330\right)-6 (2 \Delta+9) m^3-3 m^4 \nonumber\\
  \tilde A_{6} &= m^2 \left(5 \Delta^2+57 \Delta-\ell^2-\ell+119\right)+(2 \Delta+9) m \left(\Delta^2+21 \Delta-\ell^2-\ell+38\right)\nonumber\\
  &\qquad +6 (2 \Delta+3) (\Delta-\ell+4) (\Delta+\ell+5) +2 (2 \Delta+9) m^3+m^4 \nonumber\\
  \tilde A_{7} &= m^2 \left(5 \Delta^2+69 \Delta-\ell^2-\ell+179\right)\nonumber\\
  &\qquad +(2 \Delta+11) m \left(\Delta^2+25 \Delta-\ell^2-\ell+58\right)\nonumber\\
  &\qquad +14 (\Delta+2) (\Delta-\ell+5) (\Delta+\ell+6)\nonumber\\
  &\qquad +2 (2 \Delta+11) m^3+m^4 \label{eq:Atilde}
\end{align}
\endgroup
(here written at $m$ shifted so that \eqref{eq:seventerm} reads
$\sum_k\tilde A_k(m)\,c_{m+k}=0$; note the exact antisymmetry
$\tilde A_4=-\tilde A_3$ of the middle pair).\footnote{These displays
are generated mechanically (\code{TeXForm}) from the machine-derived
recurrence and, before export, validated by checking
$\sum_k\tilde A_k\,c_{m+k}=0$ exactly against block coefficients from
the independent pole-recursion evaluator at rational probes --- they
document the exact object of input [C7], not a hand transcription.}
The derivation is an exact rational computation, performed and verified
independently by two released checkers (Section~\ref{sec:certinputs},
input [C7]).

\paragraph{Three recursions, three roles.} Three recursions appear in
this paper; they are the three ingredients of the standard rational-approximation pipeline
\cite[Sec.~III.F]{Poland:2018epd}. The \emph{pole recursion} recurses
in $\Delta$: it assembles the block at $(\Delta,\ell)$ from the seed and
from blocks at shifted quantum numbers sitting at its poles, and with
the spin kept symbolic it returns each radial coefficient as a rational
function of $(\Delta,L)$ whose denominators are the linear pole-tower
factors. It computes the tables the certificates consume through level
$22$. The \emph{radial Casimir recurrence}
\eqref{eq:seventerm} recurses in the radial level at fixed
$(\Delta,\ell)$: it relates eight consecutive coefficients of one block.
In its zero-sum form (Lemma~\ref{thm:telescope}) it supplies the levels
$24$, $26$ and $28$ of the tables exactly, and it is what bounds the
levels beyond $\kappa_{\rm cert}=28$ that the truncation omits, which are never
computed. The \emph{transverse Casimir recursion} (App.~\ref{app:blocks})
reconstructs the derivatives transverse to the diagonal from diagonal
data. The first two hold for the coefficients of the one true block ---
the pole recursion converges to it (input T2 below) and the block
satisfies the HOR ODE --- so a bound derived from the recurrence bounds
the error of the pole-recursion tables; their agreement on the
computed levels is cross-checked at rational probe points (the footnote
above); the pole-recursion tables themselves are trusted as transcribed
(Section~\ref{sec:trustedbase}). The unequal-external generalization of
the recurrence powers the mixed-correlator coefficient majorants of
Section~\ref{sec:mixedoutlook}.

\section{The reduction: analytic inputs and finitely many polynomial
sign conditions}
\label{sec:searchcert}
\label{sec:reduction}

\subsection{The two-part target, and what has to be trusted}
\label{sec:target}

Everything that has to be proven is reduced to
\begin{enumerate}[nosep,label=(\Alph*)]
\item the \emph{analytic part}: a short, fixed list of cited theorems
  about conformal blocks, plus coefficient bounds for the block
  recursion that we prove ourselves, plus the handful of finite exact
  computations that are not sign conditions (symbolic identities, and
  the check that the certified tiles cover the claimed region) ---
  collected in Section~\ref{sec:analyticinputs}; and
\item a finite list of \emph{polynomial sign conditions}
\begin{equation}
  P(\vec x)\ \ge\ 0\quad\text{on }[0,1]^{n},\qquad n\le3,\qquad
  P\in\Q(\sqrt2\,)[\vec x]\ \text{explicit}.
  \label{eq:masterform}
\end{equation}
\end{enumerate}
Statements of kind (B) are \emph{decidable}: each one is settled,
affirmatively or negatively, by a finite exact computation on the
coefficients of $P$ (Section~\ref{sec:bernstein}). Within (A), the
theorems and coefficient bounds are the same for every application of
the pipeline and never refer to any particular functional. Everything
else specific to a given exclusion --- the tile, the functional, the
truncation order --- enters only through the explicit polynomials
in (B).

Table~\ref{tab:loopholemap} records where each of the four loopholes of
Section~\ref{sec:intro} goes under this reduction; the remainder of
this section derives the table row by row. (The unitarity-bound
endpoint caveat of Section~\ref{sec:intro} cannot take the form
\eqref{eq:masterform} at all; it is resolved analytically, by
Lemma~\ref{lem:corner}.)

\begin{table}[t]
\centering
\begin{tabular}{llc}
\toprule
loophole & becomes & where \\ \midrule
1.\ points $\to$ regions & one polynomial direction in $\Ds$ & \S\ref{sec:regionreduction} \\
2.\ finite $\to$ all spins & one polynomial direction in real $L$ & \S\ref{sec:allspins} \\
3.\ truncated $\to$ true blocks & subtract a certified error envelope & \S\ref{sec:envelopestatement} \\
4.\ floats $\to$ exact & coefficients in $\Q(\sqrt2)$, signs decided exactly & \S\ref{sec:exactstatement} \\ \bottomrule
\end{tabular}
\caption{The four loopholes of Section~\ref{sec:intro} and the
mechanism by which each is reduced to polynomial sign conditions of
the form \eqref{eq:masterform}.}
\label{tab:loopholemap}
\end{table}

With the analytic part (A) separated out, handling the sign conditions
(B) takes two layers, search and then certify:
\begin{enumerate}[nosep]
	\item \textbf{The search layer} (Section~\ref{sec:regionpractice}) may use
	anything --- floating point SDP/LP,
	grids, heuristics, the sample-and-audit loop of Section~\ref{sec:auditloop}.
	Its output is a \emph{candidate}: a tile $T$ and a functional with exactly
	rational entries $\hat\alpha\in\Q^{N}$. Nothing about it needs to be trusted.
	\item \textbf{The certificate layer} (Section~\ref{sec:certarch}) verifies,
	in exact arithmetic, finitely
	many inequalities about explicitly constructed rational functions. If it
	reports \code{True}, those inequalities \emph{hold}. If the search proposed a bad candidate, this layer certifies its failure; soundness cannot be compromised even if the search provides a bad candidate.
	Every check in this layer is \emph{fail-closed}: any interruption, error, or comparison the program cannot decide results in a failure verdict, never in a false success.
\end{enumerate}
The conversion of the certified statements (B)
into the physical statement of interest, i.e. Equation \ref{eq:funcconds} --- but for the
true blocks and all spins --- and ultimately Theorem~\ref{thm:ring} is the content of the analytic part (A).

Once the functional is fixed, every claim to certify is of the form
\begin{equation}
	F(\vec x)\ \ge\ 0\qquad\text{for all }\vec x\in D,
	\label{eq:posclaim}
\end{equation}
with $F$ a \emph{fixed} rational function of one to three variables (twist,
spin, external dimension) with coefficients in $\Q(\sqrt2)$, and $D$ a product of intervals and half-lines.
The next section turns every claim \eqref{eq:posclaim} into the form \eqref{eq:masterform} by mapping each domain $D$ onto $[0,1]^n$ by exact substitutions and eliminating the denominators of $F$ in a sign preserving manner. That section also fixes the row normalization that keeps every claim finite on the compactified domains. Sections~\ref{sec:regionreduction}, \ref{sec:allspins}, \ref{sec:envelopestatement} and~\ref{sec:exactstatement} then close the four loopholes, in order, through such positivity conditions.

\subsection{Domains: intervals, half-lines, and the point at infinity}
\label{sec:compact}
A positivity claim for a rational function on an unbounded domain is converted into the
sign claim for a single polynomial on a closed box by four operations,
applied in order. Each operation either reparametrizes the domain or
multiplies the claim by a quantity of known sign, so each preserves
the claim exactly.
\begin{enumerate}[nosep,label=(\roman*)]
\item \emph{Normalize.} All quantities are expressed in the stripped
  normalization \eqref{eq:stripconv}, specified below in operation
  (i): a known positive scale is divided
  out of each \emph{constraint row}. A row is one constraint of the
  positivity problem: the crossing vector $\Vvec(x)$ at a single
  constraint label $x$ --- one choice of sector, external dimension,
  twist and spin, or the identity --- whose functional components form
  one row of the constraint matrix of the linear program. The columns
  of that matrix are the derivative components $(m,n)$, weighted by
  the unknowns $\alpha_{(m,n)}$; each row contains every component of
  its label's vector, and the constraint is the single inequality
  $\alpha\cdot\Vvec(x)\ge0$. The family of constraints is thus the
  family of rows: a sector contributes a continuum of them, one per
  operator label. The
  scale is chosen so that the row stays finite as any variable grows
  without bound. A constant bound or margin on the normalized quantity
  therefore always means: a bound by the same constant \emph{times}
  the known positive scale.
\item \emph{Compactify.} Every half-line $[c,\infty)$ in a variable
  $v$ is mapped to the compact interval $[0,1]$ by the substitution
  $t=c/v$: the points of the half-line land in $t\in(0,1]$, and the
  added endpoint $t=0$ is the image of the point at infinity --- now
  an ordinary boundary face of a closed box. Bounded intervals map to
  $[0,1]$ affinely.
\item \emph{Clear the negative powers.} The substitution produces
  negative powers of $t$; multiplying by a sufficiently high power of
  $t$, which is positive on $(0,1]$, removes them.
\item \emph{Clear the denominator.} What remains is a rational claim
  $N/D\ge0$. The denominator $D$ is factored, each factor is
  \emph{certified} positive on the closed domain, and the claim
  becomes the sign of the polynomial $N$ alone. If a factor's
  positivity cannot be certified, the run aborts rather than
  certifying.
\end{enumerate}
The output is one polynomial with exact coefficients on one closed
box, and the certified statement is its sign there. On the image of
the original domain ($t>0$) this is the original claim, by the
sign preservation of (i)--(iv). The added faces $t=0$ deserve one
clarification. The polynomial extends continuously to them, and its
values there are its leading-coefficient data --- exact algebraic
objects. The face value is not defined as --- and need not equal --- a limit
of the original normalized function: that function may fail to
converge along paths into the face, and where the clearing power of
operation (iii) vanishes on the face, even a convergent path's limit
can differ from the face value. No limit is ever taken or needed.
Certifying the sign on the \emph{closed} box proves the original claim
on the open domain and, in addition, the exact leading-coefficient
statement that the face value encodes. The rest of this subsection first
fixes the normalization (i), then carries out (ii)--(iv) in detail, and
closes with one worked example of the constant-versus-scale semantics.

\paragraph{Operation (i): the row normalization.}
\label{sec:normalization}
Operation (i) is fixed first, because the face values that operations
(ii)--(iv) attach to the points at infinity are only meaningful once
the growth of the rows has been divided out.

\paragraph{Asymptotic anatomy of the crossing vectors.}
\label{sec:asymanatomy}
The choice of normalization of our positivity constraints is dictated by the asymptotic behavior of the crossing vectors. Throughout we parametrize operators by
twist and spin, $(\tau,\ell)$ with $\Delta=\tau+\ell$.\footnote{The unitary sector can be written as the union of fixed-twist rays $\tau\ge\tau_{\min}$ along which $\ell$ runs to
infinity --- and it is at fixed twist that large-spin operators
organize (and that the block data below stays uniformly controlled).}

\emph{Large $\Delta$ at fixed $\ell$.} In radial coordinates the block
is $(4r)^\Delta$ times a stripped function with a
$\Delta$-independent limit $h_\infty$ (App.~\ref{app:blocks}), and at
the crossing-symmetric point each $z$-derivative acts dominantly on
the exponential factor, multiplying by
$\partial_z\log(4\rho)|_{1/2}\cdot\Delta=2\sqrt2\,\Delta$. Two
further contributions enter, both visible in \eqref{eq:crossing}.
First, a derivative may act on the leg factor $v^{\Ds}$ instead of on
$(4\rho)^\Delta$; such a derivative produces no factor of $\Delta$,
so these contributions are suppressed by one power of $1/\Delta$
each, and at leading order the leg factor enters only through its
value $v^{\Ds}|_{z=\bar z=1/2}=4^{-\Ds}$. Second, $F_{\Delta,\ell}$
is a difference of two terms: call $v^{\Ds}g_{\Delta,\ell}(u,v)$ the
\emph{direct term}, and call $u^{\Ds}g_{\Delta,\ell}(v,u)$ --- the
image of the direct term under the substitution $z\to1-z$,
$\bar z\to1-\bar z$, which swaps $u\leftrightarrow v$ --- the
\emph{mirror term}. The crossing-symmetric point is a fixed point of
that substitution, so every $(m,n)$-derivative of the mirror term
equals $(-1)^{m+n}$ times the same derivative of the direct term. In
the difference the two terms therefore add to twice the direct term
when $m+n$ is odd, and cancel when $m+n$ is even --- which is why
only odd pairs $(m,n)$ appear in the functional basis. Including both
contributions,
\[
V_{(m,n)}(\Delta,\ell)\;\simeq\;2\cdot4^{-\Ds}\,
(2\sqrt2)^{m+n}\,\Delta^{m+n}\,(4\rstar)^{\Delta}\,
\big[h_\infty(\rstar)+O(1/\Delta)\big].
\]
The common factor $(4\rstar)^\Delta$ is positive and drops from
positivity; what remains is the \emph{polynomial hierarchy}
$\Delta^{m+n}$: the sup-normalized vector converges, at rate
$O(1/\Delta)$, to the ray of the top-order components, and positivity
degenerates to a half-space condition on finitely many leading
coefficients. One negative leading coefficient therefore violates
infinitely many constraints beyond any sample --- and that leading
coefficient is the homogenized boundary value the compactified
certificates test.

\emph{Large $\ell$ at fixed $\tau$.} Here only the holomorphic half of
the block carries the large weight $h=\ell+\tau/2$; the other half is
the \emph{fixed} light collinear block $k_\tau$ of weight $\tau$.
Each component receives one term of degree $m$ and one of degree $n$
from the $z\leftrightarrow\bar z$ symmetrization:
\begin{equation}
\begin{aligned}
V_{(m,n)}(\tau,\ell)\;\simeq{}&2\cdot4^{-\Ds}\,(4\rstar)^{\ell+\tau/2}\Big[
 (2\sqrt2)^m\big(\ell+\tfrac\tau2\big)^m k^{(n)}_\tau(\tfrac12)\\
&\qquad\pm(2\sqrt2)^n\big(\ell+\tfrac\tau2\big)^n
 k^{(m)}_\tau(\tfrac12)\Big],\big(1+O(1/\ell)\big).
\end{aligned}
\end{equation}
Again an exponential (now in $\ell$), times a polynomial in $\ell$ of
degree $\max(m,n)\le\Lambda$ --- \emph{not} $m+n$ --- whose leading
coefficient carries the $n$-th derivative of the light block. The
large-$\Delta$ ray is universal ($\tau$ drops out at leading order).
At fixed $\tau$ the row approaches the spin face through
$\tau$-dependent data --- the degree-$m$ coefficient of each
component carries the collinear derivative $k^{(n)}_\tau(\tfrac12)$
--- and what survives on the face itself is stated below, once the
normalization has been fixed.

\paragraph{Factoring out the largest growth.} To factor out the largest growth we normalize the positivity constraints by the following positive factor:
\begin{equation}
  \widehat\Vvec(\Ds;\Delta,\ell)
  \;:=\;
  \frac{\Vvec(\Ds;\Delta,\ell)}{(4\rstar)^{\Delta}\,\bar h^{\Lambda}},
  \qquad \bar h=\tfrac{\Delta+\ell}2,
  \label{eq:stripconv}
\end{equation}
i.e.\ the raw crossing vector equals the stripped one \emph{times} the
positive scale. A number quoted in \emph{stripped units} anywhere below
is the value of a quantity normalized in this way. This is the only positive scale that is chosen as a
normalization convention for crossing-vector rows. Its purpose is to remove
the universal exponential decay and the worst polynomial growth, so that rows
and margins have finite, comparable values on compactified domains.

The scale is one positive scalar per row, common to all components
$(m,n)$, and this uniformity is forced by what the constraint at $x$
\emph{is}: the half-space condition $\alpha\cdot\Vvec(x)\ge0$ depends
on the row only through its \emph{ray} --- the vector modulo one
overall positive factor. Multiplying the row by a positive scalar
$S(x)$ moves along the ray and leaves the constraint unchanged; a
component-dependent, $x$-dependent weight moves to a \emph{different}
ray and changes the constraint.\footnote{Concretely: for
$\Vvec=(-1,2)$ and $\alpha=(1,1)$ one has $\alpha\cdot\Vvec=1>0$, but
with component weights $S=(1,4)$ the weighted contraction is
$-1+\tfrac12<0$. Nor can the weights be absorbed into a redefined
functional: $\alpha_m/S_m(x)$ differs operator by operator, while the
crossing contradiction applies one \emph{fixed} linear functional to
the vanishing sum $\sum_O p_O\Vvec(x_O)$ --- with operator-dependent
coefficients there is no linearity step left. An $x$-\emph{independent}
component rescaling, by contrast, \emph{is} a valid redefinition of
the functional (a change of basis absorbed into $\alpha$), but a
constant weight cannot alter the $x$-dependent face behavior. In the
matrix sectors of the mixed system the invariance group is larger ---
positive-semidefiniteness survives any $x$-dependent diagonal
congruence $D(x)M(x)D(x)$ --- which is why leg factors drop there
component by component; a scalar inequality has only the overall
scalar as its modulus.} Thus, components of degree below $\Lambda$ tend to zero
at the infinity faces. The certified statement is
the sign of a polynomial on the \emph{closed} box, and in the
compactified variable the lower-degree components are its
subleading coefficients --- dominant at interior points near the face,
where the Bernstein tensor sees them in full. Only the face value
itself reduces to the top-degree data: the face value \emph{is} the leading asymptotic
coefficient, the limit-row constraint.

With the convention \eqref{eq:stripconv} fixed, the asymptotics
above become statements about concrete table data. All component
degrees are bounded by $\Lambda$, so the single common factor
$\bar h^{\Lambda}$ renders every stripped component finite on the
infinity faces: top-degree components attain nonzero limits,
lower-degree ones tend to zero, and in the compactified charts
($w=\tau_{\min}/\tau$ for the twist direction, $1/\ell$ for the
spin direction) the stripped components are rational functions with
finite boundary values --- the ``limit rows'' are exact finite data.
The strict face values are correspondingly simple. On a fixed-$\ell$
face ($\Delta\to\infty$) every component with $m+n=\Lambda$
survives with one common positive limit, so the face value is a
positive multiple of the single combination
$\sum_{m+n=\Lambda}\alpha_{(m,n)}$: one sign condition on the
functional's top-degree block. On a fixed-$\tau$ face
($\ell\to\infty$) the polynomial degree is $\max(m,n)$, and the
only component of full degree $\Lambda$ is $(m,n)=(\Lambda,0)$, so
the face value is $\alpha_{(\Lambda,0)}$ times a positive
$\tau$-dependent scalar (proportional to the collinear block value
$k_\tau(\tfrac12)>0$): one sign condition on one coefficient. These
sign conditions are the limit-row constraints. Of course `near' the faces the subleading coefficients play a role.

One row scale suffices for
\emph{all} directions to infinity because every unbounded direction
in the unitary quadrant forces $\Delta\to\infty$ --- large spin at
fixed twist, large twist at fixed spin, and every path into the
doubly-compactified corner --- so the single exponential
$(4\rstar)^\Delta$ is the relevant radial scale on every infinity
face. Second, the polynomial factor tames the entire derivative basis
at once because a derivative can pull down at most one power of the
large scale per order: a component of total derivative order
$p\le\Lambda$ grows at worst like $\bar h^{\,p}$ relative to the
exponential.

The claims \eqref{eq:posclaim} live on products of intervals and half-lines.
Intervals map to $[0,1]$ affinely. For a half-line, suppose the claim is
$Q(v)\ge0$ for all $v\in[c,\infty)$, with $Q$ polynomial of degree $d$ and
$c>0$. Substitute $v=c/t$: as $t$ runs over $(0,1]$, $v$ runs over
$[c,\infty)$, so the claim is equivalent to $Q(c/t)\ge0$ on $(0,1]$. The
function $Q(c/t)$ contains negative powers of $t$; multiplying by $t^{d}>0$
--- which cannot change any sign on $(0,1]$, and a sign is all we are asking
about --- produces a genuine polynomial $R(t):=t^{d}\,Q(c/t)$, and the claim
has become: $R\ge0$ on $(0,1]$. Bernstein certificates want a \emph{closed}
box, so we certify the statement on $[0,1]$ instead, $t=0$ included. This is
not a strengthening: if $R\ge0$ on $(0,1]$ then $R(0)\ge0$ follows by
continuity, so the two statements are equivalent --- but certifying the
closed interval makes the boundary value explicit, and that value is
meaningful: $R(0)$ is the leading coefficient of $Q$ --- the direction
of $Q$'s growth at $v\to\infty$, an exact algebraic datum (for a
univariate $R$, also the limit of the normalized $Q$; in several
variables the face value is still exact leading-coefficient data even
where the normalized function has no limit, per the roadmap above). For
the spin variable this means the $t=0$ face of the certified box
carries the lightcone data $\ell=\infty$, handled not by a separate
large-spin argument but as one more coefficient inspection.

This compactification is not an asymptotic approximation: no $1/\ell$
expansion is truncated and no remainder is estimated. The half-line is
reparametrized exactly, and the face $t_\ell=0$ is an ordinary boundary
point of the polynomial's domain. When both spin and twist are
unbounded, the closed square covers every path to the corner at once,
so no order of limits is chosen.

The claims are rational, $F=N/D$, so denominators must be dispatched. The structure of the 
pole recursion makes this easy: every denominator it generates is a
product of \emph{linear} functions of the twist $\tau$ and the spin
variable $L$, times an explicitly known algebraic constant. At the first
level of the recursion the three pole families of App.~\ref{app:blocks}
contribute, at equal externals (where, in families I and III, only
even residue indices survive), the factors
\begin{equation}
  \tau+2j,\qquad 2L+\tau+(2j+1),\qquad 2L+2\tau+(2j-1),
  \qquad j\ge0,
  \label{eq:towers}
\end{equation}
from the families III, I and II respectively, with $j$ bounded by
$\kappa/2$. Nested residues are blocks at shifted quantum numbers, so
the denominators they generate are the same three forms evaluated at
those shifted arguments: again linear in $(L,\tau)$, with integer or
half-integer offsets bounded by $\kappa$ (some lose the $\tau$- or the
$L$-dependence altogether). The factor $2L+2\tau-1$, the $j=0$ member
of the last family, is the one that reaches the unitarity bound on the
scalar line, $\tau=\tfrac12$ at $L=0$, and is the origin of the cleared
certificates of Lemma~\ref{lem:strip}. A worry naturally arises here: the pole recursion's poles
include locations \emph{on} the unitarity bound (the natural boundary of
some sectors), where the associated residues can vanish --- naively a
$0/0$ where positivity is being demanded. The resolution is that
all arithmetic is exact and rational: when a residue vanishes
identically (as the twist-one residues do for equal external scalars),
the offending factor cancels between numerator and denominator
\emph{algebraically} in the reduced normal form $F=N/D$, meaning a
single rational expression after all exact cancellations have been
performed --- an exact cancellation performed once, before any
evaluation, not a numerical near-cancellation --- and the certified
objects are built from the reduced form. What remains in $D$ after reduction are factors whose
zeros lie strictly \emph{below} the sector domains (at twists $\le0$ or
below the relevant gaps). On the actual domains
$\{L\ge L_0,\tau\ge\tau_0\}$ each surviving factor is then
\emph{manifestly} positive: substitute $L=L_0+s$, $\tau=\tau_0+t$ with
$s,t\ge0$ and observe nonnegative coefficients with positive constant
term. The certifier performs this check factor by factor: if reduction
ever left a factor whose positivity on the
domain is not manifest --- for instance because a pole really did
survive at the boundary --- the check fails and the certificate aborts,
rather than the factor being assumed harmless. Two corollaries when it
passes: $F$ is pole-free on the closed domain (no hidden $0/0$
anywhere, including between integer spins --- relevant below), and
$\mathrm{sign}\,F=\mathrm{sign}\,D\cdot\mathrm{sign}\,N$ reduces
\eqref{eq:posclaim} to a polynomial claim about $N$.

\paragraph{Denominator clearing.} After the row
normalization, the entries are still rational functions. For each certified
inequality these rational functions are placed over a common denominator; the
denominator factors are certified positive on the relevant closed domain and
then cleared, leaving a polynomial numerator for the Bernstein check. This is a
per-certificate algebraic reduction.

In the implementation the two factors in \eqref{eq:stripconv} are removed at
different points in the release code: The exponential never enters at all --- every table is built from
the stripped block $h_{\Delta,\ell}=(4r)^{-\Delta}g_{\Delta,\ell}$ of
Section~\ref{sec:setup}, so table entries are rational functions with no
exponential left to remove --- while the polynomial factor
$\bar h^{\Lambda}$ is dropped from the search rows and \emph{reinstated}
in the certificate chain as an explicit factor (code variable
$\code{hb}^{\Lambda}$ with $\code{hb}=L+\tau/2$), where it is manifestly
positive on the domains and hence sign-irrelevant. All these variants
are positivity-equivalent.

\paragraph{What a constant bound on the compactified box asserts.}
Throughout the paper, bounds and margins on half-line domains are
\emph{constants}: a constant sampled margin $t^\star$, constant envelope
bounds, and constant comparisons between truncated values and their error
majorants. Crossing vectors and truncation errors visibly \emph{grow} in the raw variables --- polynomially and exponentially in $\Delta$ and $\ell$ (the asymptotics worked out in operation (i) above). 

A constant bound is therefore only possible on the \emph{normalized}
function on the \emph{compactified} domain; unwound to the raw variables it
is not a constant bound at all, but a bound by the same constant
\emph{times a known positive scale}.

A worked example shows the whole mechanism. Suppose the raw quantity
$E(v)=v^{2}+3v+2$ is to be dominated on $v\in[1,\infty)$, and the row
carries the known positive scale $S(v)=v^{2}+v+1$. No constant bounds
$E$ on the half-line. The claim $E\le 2S$, however, compactifies: with
$t=1/v$ and after multiplying by $t^{2}>0$ (the half-line substitution
above),
\begin{equation}
  2\,t^2S(1/t)-t^2E(1/t)
  \;=\;2(1+t+t^{2})-(1+3t+2t^{2})\;=\;1-t\;\ge\;0
  \quad\text{on }[0,1],
\end{equation}
a polynomial inequality on a closed interval whose Bernstein coefficients
$(1,\tfrac12,0)$ are nonnegative --- certified on the first box. At the infinity face $t=0$ the inequality holds with slack $1$, the gap of the leading coefficients ---
the $v\to\infty$ direction, checked exactly, not asymptotically. Every
``constant'' in the certificate layer is of this kind: constant in the
normalized, compactified coordinates; scale-shaped in the raw ones. Which
positive scale has been taken out of the rows, and why one scale
suffices for \emph{all} directions to infinity, was fixed in operation (i)
above.

\subsection{Loophole 1 (regions): the external dimension enters
polynomially}
\label{sec:regionreduction}

A region bound must control every external dimension $\Ds$ in an
interval $[\slo,\shi]$ at once. This costs no sampling and no remainder
terms, because of a structural fact: the external dimension enters the
constraints \emph{exactly polynomially}. The blocks are
$\Ds$-independent and $\Ds$ appears only through derivatives of the
prefactors $u^{\Ds}$, $v^{\Ds}$ at the symmetric point, so
\begin{equation}
  \Vvec(\Ds;\Delta,\ell)\big|_{(m,n)}
  =\Big(\tfrac14\Big)^{\Ds}\times\big[\text{polynomial in }\Ds\text{ of
  degree }\le m+n\le\Lambda\big],
  \label{eq:exactpoly}
\end{equation}
with the common strictly positive prefactor $(1/4)^{\Ds}$ irrelevant to any
sign condition. Demanding positivity for all $\Ds\in[\slo,\shi]$ is
therefore one more polynomial-positivity direction on a compact
interval: after the affine map of Section~\ref{sec:compact}, one more
coordinate of the box in \eqref{eq:masterform}. The claims
\eqref{eq:posclaim} for a rectangular tile $T=[\slo,\shi]\times[\elo,\ehi]$ carry the
$\Ds$-interval exactly, and the historical alternative --- elaborate
corner-plus-remainder constructions around a sampled grid of external
dimensions --- is not needed anywhere.

\subsection{Loophole 2 (all spins): one polynomial direction in real
spin}
\label{sec:allspins}

Loophole 2 asks for positivity at every even spin $\ell\in2\Z_{\ge0}$ and
every allowed twist. Finitely many spins are handled individually: the
low sectors ($\ell=0$ window and continuum, $\ell=2$ stress tensor and
$T'$ gap) carry their gap assumptions, and every even spin up to
$\ell=100$ gets its own two-variable Bernstein certificate at
unitarity, from the real-spin table of the next sentences evaluated at
$L=\ell$, which equals the standard per-spin evaluator by
Lemma~\ref{lem:interp}. The infinitely many
spins beyond $\ell_*=100$ are handled by one uniform positivity claim. This is possible because the $\kappa$-truncated crossing-vector components are components of a single \emph{rational function of a real spin variable} $L$
(and of $\tau$) obtained by computing the pole recursion with $L$ symbolic. Rationality is
constructive: the pole recursion's seed data (the
$\ell=\infty$ collinear series) is $L$-independent; its residue
coefficients are explicit ratios of Pochhammer symbols, and for each fixed
pole label appearing in the finite $\kappa$-truncated recursion these
Pochhammers have fixed integer length, hence are polynomial in their
arguments and rational functions of $L$; its pole positions are
\emph{linear} in $L$ (they are the factors \eqref{eq:towers}); and the
recursion terminates after $\le\kappa$ levels. A
finite composition of rational operations is a rational function. We then
certify
\begin{equation}
  W(L,\tau)\;=\;\sum_m\hat\alpha_m \widehat V^{(\kappa)}_m(L,\tau)\ \ge\ 0
  \qquad\text{for \emph{all real} }L\in[\ell_*,\infty],\ \tau\in[1,\infty],
  \label{eq:tailclaim}
\end{equation}
(the hat denoting the positive row-normalization of Section~\ref{sec:normalization}) by
Bernstein tensors on the compactified domain. The physical requirement
--- even integers --- is a subset of \eqref{eq:tailclaim}: relaxing to
real $L$ costs nothing at certification time, because we are verifying
positivity on the larger set, not optimizing over it. Every even spin $0\le\ell\le100$
is certified individually from the real-spin table at $L=\ell$ (equal to the standard per-spin evaluator by Lemma~\ref{lem:interp}), and the real-$L$ certificate
\eqref{eq:tailclaim} is applied to the tail $L\ge\ell_*=100$ only,
where the certified object is the interpolated trivariate polynomial of
Section~\ref{sec:regioncert}, certified in the eleven twist pieces listed in Table~\ref{tab:inventory} of Section~\ref{sec:regioncert}. We next address some possible objections:

\paragraph{Objection 1: does the interpolant mean anything at non-integer
spin?} It does not need to. No physics is claimed between integers; the
interpolant is a certification device. Its only load-bearing property is that
\emph{at} integer $L=\ell$ it coincides with the standard truncated evaluator
for spin $\ell$ --- see Objection~2.

\paragraph{Objection 2: the pole recursion is not manifestly analytic in spin ---
it contains conditionals.}
Answering this objection needs one piece of recursion anatomy that we
have not used so far (details in App.~\ref{app:blocks}). The pole recursion
builds $h^{(\kappa)}$ by adding, level by level, residue terms drawn
from three families of poles --- the three families whose linear
denominator factors were listed in \eqref{eq:towers}. Each residue
term multiplies a nested block with \emph{shifted} quantum numbers,
and the three families shift them differently: the first two leave
the spin of the nested block alone or raise it, while the third
\emph{lowers} it --- its $n$-th member replaces a spin-$\ell$ block by
one of spin $\ell-n$. A spin can only be lowered $\ell$ times, so the
standard integer-spin evaluator cuts this third tower off at length
$\min(\kappa,\ell)$: implemented, quite literally, as a conditional
``keep the term only if $n\le\ell$.'' An
object assembled through such conditionals is not obviously the
restriction of a single rational function of a \emph{continuous} spin
variable $L$ --- which is what the all-spin certificates of
this section evaluate.

The mechanism that defuses the conditional sits inside the residue
itself. The third-family residue at index $n$ carries the Pochhammer
factor
\begin{equation}
  \gamma(n,L)\ \propto\ \frac{(1+L-n)_n}{(1+\zeta+L-n)_n},
  \qquad \zeta:=\tfrac{d-2}2=\tfrac12,
\end{equation}
whose numerator is the product of the $n$ consecutive values
$(1+L-n)(2+L-n)\cdots L$: at any integer $L=\ell<n$ one of these
values is zero, so the factor \emph{vanishes identically} --- while
the denominator, offset by the half-integer $\zeta$ (the Gegenbauer index $\tfrac{d-2}2$, written $\nu$ in much of the block literature; we use $\zeta$ to keep $\nu$ for the critical exponent), is never zero at
integer $L$. In other words, the terms that the conditional
would remove at integer spin remove themselves. This
suggests discarding the conditional altogether. Define the
\emph{saturated} build $\widehat V^{(\kappa)}_m(L,\tau)$ by keeping
all $\kappa$ third-family terms regardless of $L$: with no
conditionals left. This is guaranteed to truly be a rational function of $(L,\tau)$.
Write $\widehat V^{(\kappa),\ell}_m(\tau)$ for the standard per-spin
evaluator, with its conditional enforced. The claim to be established
is that saturating loses nothing:
\begin{lemma}[exact interpolation]\label{lem:interp}
For every truncation order $\kappa\le46$ --- which includes every build
consumed in this paper --- and every even integer spin $\ell\ge0$ (the
only spins that occur in $\sigma\times\sigma$),
$\widehat V^{(\kappa)}_m(\ell,\tau)=\widehat V^{(\kappa),\ell}_m(\tau)$
identically in $\tau$. The same holds level by level for the radial
coefficients $c_{\mathrm{lev}}$, $\mathrm{lev}\le46$. (For
$\ell\ge\kappa$ the statement holds at any order; the restriction comes
from the spins $\ell<\kappa$, which are settled by an exhaustive check
and not by an argument.)
\end{lemma}
\begin{proof}
For $\ell\ge\kappa$ the two builds are termwise identical, by the
following invariant of the nested recursion. At a recursion node write
$j$ for the current spin and $b$ for the remaining radial budget (the
number of levels still to be generated below the current one). At the
root $(j,b)=(\ell,\mathrm{lev})$ with
$b=\mathrm{lev}\le\kappa\le\ell$, so $j\ge b$. A first-family
residue raises the spin while consuming budget, the second family
leaves the spin unchanged, and a third-family residue at index $n$
maps $(j,b)\mapsto(j-n,\,b-n)$: all three preserve $j\ge b$. Hence
at every node --- including arbitrarily deeply nested ones --- every
third-family index satisfies $n\le b\le j$, the conditional never
triggers, and saturation changes nothing. The
nontrivial content is confined to the finitely many even spins
$\ell<\kappa$, where the saturated build retains terms with
$\ell<n\le\kappa$ that the per-spin build drops --- and these must
vanish. The vanishing residue factor alone does not prove
that they do: $\gamma(n,\ell)=0$ multiplies a nested shifted block,
itself produced by further recursion, and if that nested object had a
pole at the same specialization, the product --- a $0\cdot\infty$ ---
could survive as a finite nonzero term. Excluding this analytically
would require tracking orders of zeros against orders of poles through
the entire nested recursion. We therefore do not argue this case; we
\emph{check} it, exhaustively: the saturated
build specialized to $L=\ell$ is compared against the per-spin
evaluator, as exact rational functions of $\tau$, at every level up to
$46$ and for \emph{every} even spin below that level --- the finite list
of all cases in which the two objects could differ. The comparison is
the released checker \code{satcheck\_L9.wl}, run with \code{KAPPA=46};
together with the index-range argument for the remaining spins it
proves the lemma at every level a certificate consumes: the levels up
to $\kappa_{\rm cert}=28$ of the tile certificates and the levels $30$--$46$ of the
coefficient windows [C3].  The comparison is of the radial coefficients
$c_{\mathrm{lev}}$; equality of the derivative components
$\widehat V^{(\kappa)}_m$ follows because both builds feed the same
block-independent level maps of Lemma~\ref{lem:leveldec}.
\end{proof}
Positivity of the saturated $\widehat V^{(\kappa)}(L,\tau)$ for all
real $L\ge\ell_*$ therefore implies positivity at every physical even
spin, simply by restriction to the integers. Between integers the
saturated build is a regular interpolating function --- no physics is
claimed there (Objection~1) --- and its pole-freeness on the certified
domain is not assumed but certified, through the factor-by-factor
positivity of the denominators after reduction to normal form.  Before
reduction, the residue coefficients of App.~\ref{app:blocks} contribute
further linear factors at real spin --- the Pochhammer denominators
$(1+\zeta+L-n)_n$ and their nested relatives, which vanish at half-integer
$L$ below $\kappa-\tfrac32$, inside the certified spin range --- and only
the tower factors \eqref{eq:towers} survive the reduction: the saturated
coefficients agree with the recurrence solution on every integer slice,
hence identically, and the recurrence's denominators are tower factors
only.  The certifier does not assume this cancellation; its denominator
check aborts on any surviving factor that is not manifestly positive. In production,
two machine checks back the lemma on complementary spin ranges: the
exhaustive $\ell<\kappa$ comparison just described
(\code{satcheck\_L9.wl}), and, on the large-spin-tail domain $L\ge100$ ---
where $\ell\ge\kappa$, so saturation is trivially inactive and the
question is instead whether the \emph{interpolated large-spin-tail
polynomials} faithfully represent the evaluator --- the large-spin-tail
scripts prove that each tile's polynomial equals one fixed
\emph{nonzero} multiple of the object assembled from the same
symbolic tables, by comparing \emph{spin slices} --- for an integer
$\ell$, the spin slice of a function $F(L,\tau,\sigma)$ at $L=\ell$ is
the two-variable function $F(\ell,\tau,\sigma)$ --- identically in
$(\tau,\sigma)$ on more slices than the degree in $L$ --- the multiple's
sign is certified per tile (negative, in fact, on all production
tiles) and absorbed into the polynomial sign claim --- a consistency check of the
interpolated polynomial against its generating data; the bridge from
that data to the physical per-spin vectors is the interpolation lemma
of Section~\ref{sec:allspins}, not this comparison (the spin-slice and
scale-identity stages of
\code{certify\_tile\_fartail\_L9.wl}, described in
Section~\ref{sec:regioncert}). The exhaustive comparison is load-bearing in three places. First, through [C3] of Section~\ref{sec:certinputs}, whose coefficient functions come from the saturated real-spin build and whose agreement with the per-spin coefficients at every even spin through level $46$ is the exhaustive part of Lemma~\ref{lem:interp}. Second, in every sector at a fixed spin $\ell\le20$ of every tile --- the scalar window and continuum, the stress tensor, the spin-2 gap, and the spins $4$ to $20$: these certificates substitute $L=\ell$ into the saturated real-spin table, whose pole-recursion levels run to $22$, so for $\ell<22$ the identification with the per-spin evaluator is the exhaustive part of the lemma, while for $\ell\ge22$ it is the index-range part (the levels $24$--$28$ come from the recurrence, which holds at every spin). Third, the strip tiles certify the three-variable piece \code{tail3d.aa} of Table~\ref{tab:inventory} on $\ell\in[14,100]$ with the real-$L$ build; this piece is redundant with their per-spin certificates on the same spins. The exhaustive check compares the gated per-spin tower with the saturated coefficients produced by the released C++ recursion; the tile tables are produced by the Wolfram implementation of the same saturated recursion, and the released checker \code{tablecheck\_L9.wl} proves that the two implementations give identical radial coefficients, as rational functions of $(L,\tau)$, at every level up to $22$.

\paragraph{Objection 3: is demanding positivity at all \emph{real} $L$ too
strong?} It is stronger than what physics requires, and demanding more
of a functional can only weaken the bounds, never invalidate them. How
much stronger is quantifiable, because the certified object is a
polynomial in the compactified spin variable. Let $P(u)$, $u=\ell_*/L$,
be the cleared large-spin-tail polynomial at fixed $(\tau,\Ds)$, of
degree $d$ in $u$ and with Bernstein coefficients bounded by $B$ in
absolute value, and suppose $P\ge m>0$ at every even integer
$\ell\ge\ell_*$. On an interval $[\ell,\ell+2]$ between consecutive even
spins, $P$ lies above its linear interpolant minus
$\tfrac12\max|\partial_L^{2}P|$; the Bernstein-form bounds
$|P'(u)|\le2dB$, $|P''(u)|\le4d(d-1)B$ and the chain rule with
$du/dL=-\ell_*/L^{2}$ give $\tfrac12\max_{[\ell,\ell+2]}|\partial_L^{2}P|
\le2d^{2}B\,\ell_*/\ell^{3}$ for $\ell\ge\ell_*$. Hence
\begin{equation}
  P(L)\;\ge\;m-\frac{2d^{2}B\,\ell_*}{\ell^{3}}
  \qquad\text{for all real }L\in[\ell,\ell+2],
  \label{eq:realLdip}
\end{equation}
so positivity at the even integers alone already implies positivity at
every real $L\ge\ell_c:=(2d^{2}B\ell_*/m)^{1/3}$. With the production
values $d\le41$, $\ell_*=100$, and a margin at the integers of a tenth
of the largest Bernstein coefficient, $\ell_c\approx150$: beyond a few
hundred the real-$L$ condition adds nothing to the integer one, and only
on $[\ell_*,\ell_c]$ is it a strengthening at all. Solving
$\ell_c\le\ell_*$ for the cutoff, the real-$L$ condition is implied by
the margined integer condition on the whole tail once
$\ell_*\ge d\sqrt{2B/m}$, about $190$ for the production polynomials,
and then costs no bound strength. This threshold is an estimate rather
than a prescription: $d$ is fixed by the truncation order, but $B$ and
$m$ belong to the polynomial certified at a given $\ell_*$ --- the
Bernstein coefficients in $u=\ell_*/L$ and the margin at the integers
beyond $\ell_*$ are both recomputed when $\ell_*$ changes --- so the
inequality can only be checked for a chosen cutoff after the fact, not
used to choose it in advance. 

In practice, the searches only imposed positivity at sampled integer spins and
at $\ell=\infty$ and yet the candidate functionals certified on the full real-$L$ domain with margins comparable to its integer-spin margins.

\subsection{Loophole 3 (true blocks): subtract a certified error
envelope}
\label{sec:envelopestatement}

Loophole 3 is the only place where any analytic input beyond exact algebra
enters, so we spell the logic out step by step. The required implication is
simple: if the true value differs from the computable truncated value
$W^{(\kappa)}$ by at most a certified error bound $M$, and if
$W^{(\kappa)}\ge M$ everywhere, then the true value is nonnegative
everywhere:
\begin{equation}
  W^{\rm true}\;=\;W^{(\kappa)}+\big(W^{\rm true}-W^{(\kappa)}\big)
  \;\ge\;W^{(\kappa)}-M\;\ge\;0.
  \label{eq:ribbon}
\end{equation}
All three quantities are in stripped units (Section~\ref{sec:normalization}),
so \eqref{eq:ribbon} compares finite quantities on the compactified
domains.
The margin $M$ is a certified upper bound on
our ignorance of the true blocks, and ``demanding positivity by the margin''
(one more Bernstein certificate, for
$\mathrm{num}(W^{(\kappa)})-M\cdot\mathrm{den}(W^{(\kappa)})\ge0$) is
the price of replacing an uncomputable function by a computable one.

The bound $M$ --- the \emph{truncation envelope} --- is constructed in
Section~\ref{sec:trueblocks}, and the properties it satisfies are
stated and proved there as Theorem~\ref{thm:envelope}. Nothing else in
the reduction depends on \emph{how} $M$ is built: once $M$ is fixed,
the margined claim above is one more statement of the form
\eqref{eq:posclaim}, and after the compactifications of
Section~\ref{sec:compact} one more polynomial sign condition
\eqref{eq:masterform}.

\subsection{Loophole 4 (floats): exact coefficients in \texorpdfstring{$\Q(\sqrt2)$}{Q(sqrt2)}}
\label{sec:exactstatement}
\label{sec:qsqrt2}
Every quantity entering the positivity conditions --- block-recursion data, crossing
vectors, functional entries --- lives in the field $\Q(\sqrt2)=\{a+b\sqrt2:\ a,b\in\Q\}$, because the single irrational input is the radial coordinate \eqref{eq:rstar} of the symmetric point (and powers of $\sqrt2$ from derivatives of the map $z\mapsto\rho$). Thus we just need $\Q(\sqrt2)$ arithmetic acting on pairs of rationals:
\begin{equation}
	(a_1+b_1\sqrt2)(a_2+b_2\sqrt2)=(a_1a_2+2b_1b_2)+(a_1b_2+a_2b_1)\sqrt2.
\end{equation}
Note that \emph{sign determination} is an exact rational comparison: $a+b\sqrt2>0$ iff ($a\ge0$ and $b\ge0$, not both zero) or ($a>0>b$ and $a^2>2b^2$) or ($b>0>a$ and $2b^2>a^2$). No decimal expansion of $\sqrt2$ is ever needed.

The output of any floating point search used to find a candidate (untrusted) functional is \emph{rationalized} before it is used in any certification.

\subsection{The analytic inputs (A), collected}
\label{sec:analyticinputs}

Everything the proof uses beyond statements of kind (B) is the
following, stated here in one line each. Section~\ref{sec:trustedbase}
audits the same list exhaustively, with references and with the exact
role each item plays.
\begin{description}[nosep]
\item[T0: OPE convergence.] The block decomposition converges
  absolutely near the symmetric point, so the functional may be applied
  term by term.
\item[T1: radial positivity.] The stripped block's radial expansion has
  nonnegative coefficients and converges for $r<1$.
\item[T2: recursion structure.] The pole recursion converges to the
  true block, and its truncation at level $\kappa$ agrees with the
  block through radial order $\kappa$; the same block satisfies the
  diagonal Casimir ODE, hence its coefficients obey the radial Casimir
  recurrence.
\item[Coefficient bounds.] The one analytic statement proved in this
  paper rather than cited: the radial coefficients obey the bounds
  behind the truncation envelope (Theorem~\ref{thm:envelope}), derived
  from the zero-sum form of the radial Casimir recurrence
  (Section~\ref{sec:trueblocks}).
\sloppy
\item[Exact identities.] Finitely many symbolic identities consumed by
  the constructions, each machine-checked in exact arithmetic and shipped
  as a released artifact: the two identities behind the
  zero-sum reduction among the certified inputs [C1]--[C7]
  (Section~\ref{sec:certinputs}), and the all-spin saturation identity
  as its own checker (\code{satcheck\_L9.wl},
  Section~\ref{sec:allspins}).
\item[Coverage.] The certified tiles exactly cover the claimed
  exclusion region: one finite computation in exact rational
  arithmetic (eq.~\eqref{eq:coverage},
  Section~\ref{sec:rigresults}).
\item[Unitarity-bound endpoint.] A scalar saturating the unitarity
  bound has a null level-two descendant (cited), and a three-line exact
  computation then shows it decouples from $\sigma\times\sigma$; used once,
  to close the single boundary point $\De=\tfrac12$ of the covering
  (Lemma~\ref{lem:corner}).
\end{description}
\fussy
T0--T2 are cited theorems; everything else in this list is either
proved in this paper (the coefficient bounds, in
Section~\ref{sec:trueblocks}, with every mechanical step of the
derivation machine-checked and shipped) or is itself a finite exact
computation with a released artifact. Nothing in (A) involves a limit,
an estimate, or a floating-point number; floating-point evaluations
appear only as cross-checks of transcribed formulas
(Section~\ref{sec:trustedbase}), never as inputs.

\section{The search: Finding a region functional that is positive with sufficient margin}
\label{sec:regionpractice}

Section~\ref{sec:reduction} fixed \emph{what} has to be certified. This
section describes the layer that produces the candidate functionals ---
the layer of which nothing needs to be trusted --- and the practices
that make its candidates survive certification on the first attempt.

\subsection{The native problem, and three ways to make it finite}
\label{sec:threeways}
\label{sec:searchengine}

We are searching for a finite linear functional $\alpha$
satisfying a \emph{continuum} of linear inequalities: positivity of
$\alpha\cdot\Vvec(x)$ for every constraint label
$x=(\text{sector};\Ds,\Delta,\ell)$. I.e.\ every external $\Ds$ in some interval of interest; every
dimension $\Delta$ in the allowed range (typically from the lower bound set by unitarity $\Delta_{\text{uni}}$ to infinity); every even spin in $2\mathbb Z_{\geq0}$ (only even spins appear in the OPE of identical scalars), plus positivity of $\alpha\cdot\vec{\mathrm{Id}}(\Ds)$ on the $\Ds$-interval.
The external dimension enters these constraints exactly polynomially
(eq.~\eqref{eq:exactpoly} of Section~\ref{sec:regionreduction}), with a
common strictly positive prefactor irrelevant to any sign condition.
There are various methods to deal with continuum directions in these searches in finite numerical computations.
\begin{enumerate}[nosep,leftmargin=1.6em]
\item \textbf{Sample everything $\Rightarrow$ fully sampled LP.} Pick finitely many
  values of every continuous label; each choice gives one linear
  inequality in $\alpha$; solve the resulting linear program. Positivity
  \emph{between} the chosen values is not imposed, but one can hope that for a well-chosen set of samples this positivity occurs automatically.
\item \textbf{Keep $\Delta$ exact, sample the rest $\Rightarrow$
  PMP/SDPB.} This is the standard bootstrap reduction for point bounds: for fixed
  externals, each sector's constraint is a polynomial (matrix) in the
  single variable $\Delta$ times a positive prefactor, and SDPB imposes
  its positivity on the entire half-line, rather than at sampled points,
  through the classical equivalence between one-variable polynomial
  positivity and sums of squares. For a region bound involving externals, one must still sample those.
\item \textbf{Keep $\Delta$ and the externals exact $\Rightarrow$
  SOS-extended PMP.} The classical Luk\'acs theorem states that a
  one-variable polynomial $P(\Delta_\sigma)$ is nonnegative on an interval
  $[\slo,\shi]$ \emph{if and only if} it can be written as
  \begin{equation}
    P(\Ds)=(\Ds-\slo)\,m(\Ds)^{T}G_0\,m(\Ds)+(\shi-\Ds)\,m(\Ds)^{T}G_1\,m(\Ds)
    \label{eq:lukacs}
  \end{equation}
  for odd $\deg P=2k{+}1$, where $m(\Ds)=(1,\Ds,\dots,\Ds^{k})^{T}$ is
  the vector of monomials and $G_{0},G_{1}$ are positive-semidefinite
  $(k{+}1)\times(k{+}1)$ Gram matrices; for even $\deg P=2k$ the
  analogous form is $m^TG_0m+(\Ds-\slo)(\shi-\Ds)\,\tilde m^TG_1\tilde m$
  with $\tilde m$ one degree shorter. (Each term $m^TGm$ with
  $G\succeq0$ is a sum of squares of polynomials --- write
  $G=\sum_i\lambda_i w_iw_i^T$ --- and the interval factors carry the
  endpoint information; the nontrivial direction of Luk\'acs is that
  \emph{every} interval-nonnegative polynomial arises this way, with no
  loss of degree.) Now recall that each constraint row of our problem,
  $P(\Ds)=\alpha\cdot\Vvec(\Ds;\Delta,\ell)$, is a polynomial in $\Ds$
  of degree $\le\Lambda$ (eq.~\ref{eq:exactpoly}; e.g. at
  $\Lambda=9$ and $k=4$ the Gram matrices are $5\times5$ ), \emph{and}
  polynomial in $\Delta$. Promote the Gram entries to decision
  variables alongside $\alpha$ --- more precisely, to polynomials in
  $\Delta$ of matched degree with free coefficients, $G_{0,1}(\Delta)$
  --- and impose \eqref{eq:lukacs} as an identity in $\Ds$: comparing
  coefficients of each power of $\Ds$ gives finitely many
  \emph{linear} equations relating the Gram coefficients to $\alpha$.
  What remains to demand is $G_0(\Delta)\succeq0$ and
  $G_1(\Delta)\succeq0$ for all $\Delta$ in the sector's range --- the
  positivity of a polynomial \emph{matrix} in the single variable
  $\Delta$, which is the constraint type SDPB natively imposes.
  The net effect: two $(k+1) \times (k+1)$ polynomial-matrix constraints per
  sector replace that sector's scalar constraint, and positivity holds
  exactly in $\Ds$ \emph{and} $\Delta$ by construction, with nothing
  sampled.
\end{enumerate}
The trade-offs between these different approaches are discussed in App.~\ref{app:search}, we currently run the first option for single-correlator regions because of existing LP support for exact rational arithmetic and unrestricted constraint types.\footnote{And simply because we got it working first.}

One central fact in either approach: Whatever the reduction,
the returned $\alpha$ is only a \emph{candidate}: the final theorem rests
exclusively on the certificate layer, which proves continuum
positivity on the true blocks with no reference to how $\alpha$ was
found. A bad candidate costs a re-run, never soundness.

\subsection{Anatomy of one tile exclusion (via the fully sampled LP of Section~\ref{sec:threeways})}

\subsubsection{An overview}
Target: Finding a functional that is positive on a rectangular tile
\begin{equation}
  T=\Isig\times\Ieps,\qquad
  \Isig=[\slo,\shi]\ni\Ds,\qquad
  \Ieps=[\elo,\ehi]\ni\De ,
  \label{eq:tiledef}
\end{equation}
in the $(\Ds,\De)$ plane. Throughout the paper $\Isig$ denotes a
tile's $\Ds$-range and $\Ieps$ its $\De$-range --- the
\emph{even-scalar window} in which the one relevant even scalar is
allowed to sit; the subscripts $\pm$ mark the lower and upper endpoints
of each range.
\begin{enumerate}[nosep,leftmargin=1.6em]
	\item \textbf{Data.} Truncated crossing vectors
	$\Vvec_{\kappa}(\Ds;\Delta,\ell)$ --- exact rational tables at radial
	order $\kappa$, polynomial in $\Ds$ --- and their exact collinear
	$\ell\to\infty$ limit vectors, finite in the stripped normalization
	$\widehat\Vvec:=\Vvec/\big[(4\rstar)^{\Delta}\bar h^{\Lambda}\big]$
	defined in Section~\ref{sec:normalization}.
	\item \textbf{Initial sample set.} For each sector (the scalar in the window
	$[\elo,\ehi]$; scalars in the continuum above the gap; each even spin to $12$; the
	stress tensor; the $\ell=\infty$ limit): a ladder of twist values of
	absolute spacing $\tfrac12$, in the twist $\tau:=\Delta-\ell$, from the
	sector's gap through gap${}+20$, plus sparse anchors out to
	$5\times10^{3}$ (window sectors use the analogous refinement of their
	width), each evaluated at five values of $\Ds$ spanning
	$[\slo,\shi]$. Finally it turns out that samples at extremely large twist, 
	$\tau\sim10^4,10^6$ and $\ell\sim10^4,10^6$, are included as well:
	the limit rows constrain the boundary value itself, while the
	near-corner rows constrain the \emph{approach} to it (a functional
	can satisfy the limit row exactly and still dip below zero at large
	finite spin through a negative subleading coefficient). With both
	row families imposed, every production tile certified on the first
	attempt; we have not run the searches with either family removed, so
	their individual contributions are not separated. See
	Section~\ref{sec:allspins}. The identity vector $\vec{\mathrm{Id}}(\Ds)$ is sampled at $65$ equally spaced values of $\Ds$. All told a few thousand
	inequalities for the unknown coefficients $\alpha_{(m,n)}$ multiplying the
	derivative basis at the crossing-symmetric point. $N_\Lambda=15$ at the used $\Lambda=9$.
	\item \textbf{The search.} We maximize the \emph{uniform relative margin}:
\begin{equation}
  t^\star=\max_\alpha\ t
  \quad\text{s.t.}\quad
  \frac{\alpha\cdot\Vvec(x_i)}{\|\Vvec(x_i)\|_\infty}\ \ge\ t\ \ \forall i,
  \qquad \alpha\cdot\vec{\mathrm{Id}}(\Ds^c)=1,
  \label{eq:marginmax}
\end{equation}
solved by simplex in exact rational arithmetic, where
\begin{equation}
  \|\Vvec(x)\|_\infty\;:=\;\max_{1\le m\le N_\Lambda}\big|\Vvec(x)\big|_m
  \label{eq:supnorm}
\end{equation}
is the largest absolute component of the constraint vector at $x$;
dividing by it makes $t$ dimensionless and comparable across rows whose
overall scales differ by orders of magnitude.

The solve is a linear program in the lifted variables
$(\alpha,t)\in\Q^{N_\Lambda+1}$: each sampled constraint contributes the
row $\alpha\cdot\Vvec(x_i)-t\,\|\Vvec(x_i)\|_\infty\ge0$, affine in
$(\alpha,t)$; the identity normalization is one equality row; the
objective $t$ is capped at $t\le10$\footnote{We use a cap $t\le10$ because once the achievable margin exceeds a level comfortably above the error budget, there is no point in optimizing further.}.

Every \emph{input} is exact: the derivative tables are
rational functions of $(\Ds,\tau,\ell)$ with rational coefficients,
and every sample point $x_i$ is rational. The one place non-rational
numbers appear is row assembly: the tables are \emph{evaluated} at the
sample points in high-precision floating arithmetic ($35$ significant
digits) rather than exact fractions, purely for speed. That single
inexact input is lifted immediately. Each row is divided by its
sup-norm, the columns are rescaled to comparable magnitude (a diagonal
change of variables on $\alpha$, undone on output), and every entry is
then replaced by a nearby rational at relative accuracy $10^{-12}$.
The rationalized rows define a program that gets solved
exactly, by the simplex method in exact fraction arithmetic
(Mathematica's exact \code{LinearProgramming}), returning the optimal
vertex $(\alpha,t^\star)$ as exact rationals. Exact pivoting here is a
robustness choice rather than a rigor requirement: after scaling, the
identity rows still span some twelve orders of magnitude, and
floating-point LP solvers were observed to declare feasible programs
infeasible; an exact simplex cannot. Undoing the column rescaling
reintroduces the floating scale factors of the assembly stage, so the
functional is \emph{stored} as arbitrary-precision reals.

This is safe because the search's arithmetic carries no weight in the
proof. Its output is frozen at the certificate boundary: the
certification scripts consume one fixed rational vector per tile,
shipped in the release as a file of exact rationals
(\code{functionals/rational/}). That vector was obtained from the stored
arbitrary-precision reals by \code{Rationalize[.,0]}, which returns for
each stored number the rational of smallest denominator lying within its
recorded precision --- a deterministic conversion of the stored
representation, not a re-rounding to a chosen number of digits --- and every checker aborts unless the frozen file
agrees with that conversion of the decimal file it is handed. Every
statement the theorem relies on is then proven about that one fixed
rational vector, against the true blocks,
in exact arithmetic. The search is a heuristic that proposes
candidates; a defective search can produce a functional that fails
certification, but it cannot produce a false theorem.

Why maximize a margin as the objective? A functional that barely satisfies the sampled constraints has no reason to satisfy the unsampled ones; one that clears every sampled row by a uniform relative slack $t^\star$ can only violate an unsampled constraint if
$\alpha\cdot\Vvec$ varies by more than $t^\star\|\Vvec\|_\infty$
between samples.

One might however worry that optimizing produces an extremal functional
with explicit zeros, i.e.\ zero margin. The margin objective changes
which function is allowed to have zeros. A bound-optimizing or
bare-feasibility solve has constraints $\alpha\cdot\Vvec\ge0$; its
optimizer typically sits on the boundary of that feasible cone, so
$\alpha\cdot\Vvec$ itself touches zero on the sampled extremal spectrum.
That is the fatal kind of zero, since downstream certification must
prove $\alpha\cdot\Vvec\ge\mathrm{Env}$ pointwise for a positive error
envelope $\mathrm{Env}$, and no envelope survives a genuine zero.

In the margin-maximizing solve the constraints are instead
\begin{equation}
  \alpha\cdot\Vvec(x_i)-t\|\Vvec(x_i)\|_\infty\ge0 .
\end{equation}
Thus active rows at the optimum mean
$\alpha\cdot\Vvec(x_i)=t^\star\|\Vvec(x_i)\|_\infty$, not
$\alpha\cdot\Vvec(x_i)=0$. If $t^\star>0$, the unshifted sampled profile
$\alpha\cdot\Vvec$ is strictly positive at every sampled row. The zeros
have moved to the shifted profile
$\alpha\cdot\Vvec-t^\star\|\Vvec\|_\infty$, which is what we
want: the positive gap $t^\star\|\Vvec\|_\infty$ is the room available
for truncation error and for between-sample variation. As a byproduct,
$t^\star\to0$ as a tile approaches the true boundary of excludability
--- a useful diagnostic.

Between samples the slack can dip below $t^\star$ --- the audit of step 4 exists
to hunt such dips --- and the continuum statement is delivered
exclusively by the certificate layer, which proves the pointwise
inequality $\alpha\cdot\Vvec_\kappa(x)\ge\mathrm{Env}(x)$ everywhere,
with no reference to $t^\star$. Empirically the room
is vast: the smallest sampled margin we accepted --- $t^\star=1.9\times10^{-4}$, on the tile neighbouring the island --- sits eight
orders of magnitude above the envelope tax ($\sim10^{-12}$ relative).

\item \textbf{Audit and search again if necessary}: hunt for places where the
candidate fails \emph{between} samples; add each confirmed failure
location to the sample set as a new inequality and re-solve; repeat
until a full-resolution pass finds no failure.
\item \textbf{Certificate.} Hand the exact rational value of the final
$\alpha$ to the certificate layer, which proves  $\alpha\cdot\Vvec_{\rm true}\ge0$ on the continuum --- all $\Ds\in B$, all twists including $\tau=\infty$, all spins including the tail to $\ell=\infty$, true blocks --- by exact Bernstein certificates. If this step fails, its output is localized (see below): the failure location re-enters step 4 as one more sample point, or --- if it sits on a boundary of the certification domain --- prompts a new constraint family, and the search reruns. With the audit loop of Section~\ref{sec:certfailure} in place, every candidate that cleared the audit also passed certification.
\end{enumerate}

\subsubsection{The audit loop, in detail}
\label{sec:auditloop}

The loop is needed because the optimizer sees only the sampled
inequalities and is being \emph{optimized against them}: nothing
prevents $\alpha\cdot\Vvec(x)$ from dipping negative at unsampled $x$,
and since the optimizer trades such unseen freedom for
objective value, between-sample dips are generic, not exceptional.
Refining the sample grid uniformly does not cure this --- the
re-optimized functional relocates its dips between the new points. What
does cure it is the classical cutting-plane strategy: locate an actual failure of the current candidate, add \emph{that location} to the sample set as one more inequality, and
re-solve, so that sampling density is spent where the current 
functional fails. The loop terminates once a full-resolution pass finds
no failure.

One round looks like this. The current candidate $\alpha$ is evaluated
along each sector's twist axis on a ladder far finer than the
constraint sample spacing --- absolute twist step $1/128$ in early
rounds and $1/1024$ in the final full-resolution sweep, i.e.\ $64$ to
$512$ times finer than the round-0 constraint ladder's step of
$\tfrac12$ (window sectors are scanned at the same refinement of their
width) --- at several values of $\Ds$. Local minima of $\alpha\cdot\Vvec/\|\Vvec\|_\infty$ below a
threshold are candidate violations; each is re-evaluated at $35$-digit
precision, and every confirmed violation is added to the sample set at
the violating position in $(\Ds,\Delta,\ell)$. The identity
component gets a stronger treatment: $\alpha\cdot\vec{\mathrm{Id}}(\Ds)$
is a \emph{polynomial in $\Ds$ alone} (degree $\le\Lambda$, no $\Delta$
dependence), so instead of sampling we isolate its real roots on
$[\slo,\shi]$ exactly (Sturm sequences) and add sample points at the
isolated roots and at the midpoints between consecutive roots ---
wherever the polynomial is actually negative. The tile is declared
\emph{excludable} when a round produces no confirmed violation at the
finest scan resolution and the identity root count on the interval is
zero.

The identity plays a special role because:
\begin{itemize}[nosep]
\item \emph{The LP pins it at one point.} We normalize
  $\alpha\cdot\vec{\mathrm{Id}}(\Ds^{c})=1$ at the interval's central
  value $\Ds^{c}$ (some normalization is necessary --- the constraints
  are homogeneous in $\alpha$); positivity at every other $\Ds$ must be
  imposed as explicit constraints.
\item \emph{Empirically it is the most dip-prone direction, and there
  is a plausible mechanism.} At the LP optimum the sampled identity
  constraints tend to be \emph{active} --- satisfied with equality ---
  because lowering $\alpha\cdot\vec{\mathrm{Id}}$ wherever it is not
  pinned frees the components to buy margin in the sectors, so the
  optimizer pushes the identity values down onto their floor at many
  sample points simultaneously. But
  $\alpha\cdot\vec{\mathrm{Id}}(\Ds)$ is one polynomial of degree
  $\le\Lambda$: forced to touch a floor at many points, it generically
  oscillates \emph{below} the floor between them (the same rigidity
  that underlies equioscillation in approximation theory). Sector rows
  do not fight the objective this way --- their slack \emph{is} the
  objective --- so their active set is sparse.
\item \emph{Exactness is uniquely cheap here.}
  $\alpha\cdot\vec{\mathrm{Id}}(\Ds)$ is a polynomial in the single
  variable $\Ds$ (no $\Delta$, no spin), so complete real-root
  isolation by Sturm sequences costs milliseconds, and there is no
  reason to tolerate sampling at all.
\end{itemize}

\subsubsection{A sampling failure mode, detected by certification}
\label{sec:certfailure}

The certificate layer rejects unless the requested sign statement has
actually been proved: any undecided sign, unproved denominator
positivity, or exhausted subdivision budget is treated as failure rather
than as a certificate. Its failures are also \emph{localized}, not bare
verdicts. A Bernstein certificate either succeeds on a sub-box, or
exhibits an \emph{exactly negative corner value} --- a specific point
$(\Ds,\tau,\ldots)$ at which the claimed inequality provably fails,
i.e.\ an exact counterexample --- or exhausts its subdivision budget, in
which case the surviving sub-boxes localize where positivity is in
doubt. The one-variable sign checks likewise return the
actual roots on failure. So when certification fails, we learn \emph{where}, and
the location dictates the repair: a counterexample in a sector's
interior gets added as a sample point; a counterexample on a boundary
of the certification domain reveals a whole \emph{family} of
constraints the sample set never represented.

An example of the latter kind is the face $w=0$ of the compactified
twist direction (Section~\ref{sec:compact}), which is $\tau=\infty$: its
value is the leading large-$\tau$ coefficient of $\alpha\cdot\Vvec$
(Section~\ref{sec:asymanatomy}), and no finite sample $\tau_i<\infty$
constrains it, at any density. Now see what a
counterexample there means. An interior counterexample says
$\alpha\cdot\Vvec<0$ at one point: one violated inequality, repaired by
adding one sample point. A negative value at $w=0$ says the \emph{asymptotic
coefficient} is negative, and therefore $\alpha\cdot\Vvec(\tau)<0$ for
\emph{all} $\tau$ beyond some finite $\tau^\dagger$ --- infinitely many
violated inequalities at once, all of them beyond the last sample, none
of them representable by finitely many more interior sample points (each
new interior sample just pushes the violation region further out without
shrinking it). The repair is therefore structural: the sample set must contain rows that
\emph{stand for the limit itself} --- the exact collinear $\ell=\infty$
vectors, and rows at $\tau,\ell\sim10^4,10^6$ near the compactified
corner --- after which the LP controls the asymptotic coefficients like
any other direction. With this coverage in place all
candidate tiles of the exclusion region passed certification on the
first pass.

\section{The certificate}
\label{sec:certarch}

This section describes the layer that decides the statements
\eqref{eq:masterform}: the Bernstein certificates themselves
(Section~\ref{sec:bernstein}), the exact arithmetic they run in
(Section~\ref{sec:qsqrt2}), and how the certificates of a whole tile
are assembled from them (Section~\ref{sec:regioncert}).

\subsection{Bernstein certificates: making positivity visible}
\label{sec:bernstein}

Let $P(t)$ be a polynomial on $[0,1]$. Expand it
in the \emph{Bernstein basis} of degree $d$,
\begin{equation}
  P(t)=\sum_{i=0}^{d}b_i\,B^d_i(t),\qquad
  B^d_i(t)=\binom{d}{i}t^i(1-t)^{d-i}.
\end{equation}
Each $B^d_i$ is manifestly nonnegative on $[0,1]$, and they sum to $1$ (they
are the terms of $(t+(1-t))^d$): a nonnegative partition of unity. Hence the implication
\begin{equation}
  \boxed{\ b_i\ge0\ \text{ for all }i
  \quad\Longrightarrow\quad P\ge0\ \text{ on }[0,1].\ }
  \label{eq:bernpos}
\end{equation}
The conversion from monomial coefficients $a_k$ to Bernstein coefficients is a
fixed, exact, rational linear map, $b_i=\sum_{k\le i}\binom{i}{k}
\binom{d}{k}^{-1}a_k$: no rounding, no evaluation, no optimization.

A worked example makes the point. $P(t)=1-t+t^2$ is strictly positive on
$[0,1]$ (minimum $3/4$), but its monomial coefficients $(1,-1,1)$ do not show
it. Its Bernstein coefficients are $(b_0,b_1,b_2)=(1,\tfrac12,1)$: all
nonnegative, so
$P=1\cdot(1-t)^2+\tfrac12\cdot 2t(1-t)+1\cdot t^2$ is now \emph{visibly} a
nonnegative combination of nonnegative functions. Positivity has been made
syntactic.

Three further properties turn \eqref{eq:bernpos} into a certification
procedure that is \emph{complete for strictly-signed instances}: every
polynomial whose minimum on the box is nonzero receives, in finite
time, either a positivity certificate or an exact counterexample.
One class of inputs falls outside this dichotomy: a polynomial whose
minimum on the box is \emph{exactly} zero (a touching zero). Such an
input may still certify --- $P(t)=t^{2}$ has Bernstein coefficients
$(0,0,1)$, nonnegative as \eqref{eq:bernpos} requires, and passes on
the first box, as does any claim whose zeros sit on boundary faces
with otherwise nonnegative coefficients --- but it need not: a zero
at a non-dyadic interior point keeps the coefficients mixed in sign at
every subdivision depth, and the loop then terminates only by
exhausting its subdivision budget --- not a verdict in either
direction but a recognizable diagnostic stall, never a wrong answer.
The load-bearing functional inequalities avoid the stall mode by
construction: each is \emph{margined}, meaning the certified quantity
is a functional value \emph{minus} a strictly positive error envelope,
and the search stage maximized the slack $t^\star$ that keeps
this difference bounded away from zero. The auxiliary claims that do
attain zero --- boundary-face values, as in the worked example of
Section~\ref{sec:compact} --- are of the benign kind and certify
directly. The stall mode therefore never appeared in any production
run:
\begin{itemize}[nosep]
\item \textbf{Exact subdivision.} The Bernstein basis --- and with it
  the criterion \eqref{eq:bernpos} --- is adapted to the interval: the
  coefficients $b_i$ certify positivity on $[0,1]$ and on $[0,1]$
  only. To examine the left half $[0,\tfrac12]$ one therefore expands
  the reparametrized polynomial $P_L(s)=P(s/2)$ --- the left half
  mapped back onto $[0,1]$ by an exact affine change of variables ---
  in the Bernstein basis on $[0,1]$, and likewise
  $P_R(s)=P\bigl(\tfrac{1+s}{2}\bigr)$ for the right half. So each
  subdivision step is, conceptually, a change of coordinates followed
  by a re-expansion. Neither needs to be performed explicitly: the
  \emph{de Casteljau recursion} computes the new coefficients directly
  from the old ones. Starting from $b^{(0)}_i=b_i$, form successive
  averages of neighbouring entries,
  $b^{(r)}_i=\tfrac12\bigl(b^{(r-1)}_i+b^{(r-1)}_{i+1}\bigr)$; then
  the Bernstein coefficients of $P_L$ are
  $b^{(0)}_0,b^{(1)}_0,\dots,b^{(d)}_0$ and those of $P_R$ are
  $b^{(d)}_0,b^{(d-1)}_1,\dots,b^{(0)}_d$. One can verify this by
  carrying out the substitution and re-expansion by hand and watching
  the binomials collapse; the de Casteljau recursion is that computation with the
  cancellations already performed, leaving nothing but halvings and
  additions --- a single exact rational linear map on the coefficient
  list. The averaged form also makes a structural property visible
  that the explicit substitution would hide: every new coefficient is
  a repeated average of old ones, i.e.\ a weighted mean with
  nonnegative weights, so subdivision can only narrow the range of
  the coefficients, never widen it --- in particular it can never turn
  an all-nonnegative list into one with a negative entry.
  Quantitatively, Bernstein coefficients converge to the function
  values at the corresponding grid points quadratically in the box
  width, so for any $P$ \emph{strictly} positive on the closed box
  there is a finite depth at which every sub-box's coefficients are
  all nonnegative.
\item \textbf{Built-in counterexamples.} The endpoint coefficients equal the
  polynomial's values at the box corners ($b_0=P(0)$, $b_d=P(1)$; likewise at
  every subdivision). An exactly negative corner coefficient is therefore a
  certified \emph{counterexample}: on strictly-signed instances the procedure
  refutes as reliably as it confirms, rather than merely sometimes
  confirming.
\item \textbf{Tensor products.} In $n$ variables, apply the transform in each
  tensor mode; \eqref{eq:bernpos} holds unchanged for the tensor coefficients
  on $[0,1]^n$. The object being transformed is the coefficient tensor of
  the polynomial, with one index per variable running over that variable's
  degree plus one, so its size is the product of the per-variable index
  ranges (each degree plus one).
  We use $n\le3$ (twist, spin, external dimension); the largest production
  instances are the three-variable spin-tail sectors, with degrees around
  $(39,54,9)$ in $(\ell,\tau,\Ds)$, i.e.\ $40\times55\times10\approx
  2\times10^4$ rational entries --- trivial for exact arithmetic.
  Subdivision, too, goes one axis at a time: bisecting a box along the
  $k$-th coordinate is the change of variables $t_k\to t_k/2$ (or
  $t_k\to(1+t_k)/2$) with all other coordinates unchanged, so the de
  Casteljau recursion is applied along the $k$-th index of the tensor
  alone, leaving the other indices untouched.
\end{itemize}

\paragraph{The subdivision algorithm.} These three properties assemble
into the complete procedure --- the certifier's decision logic is
nothing beyond the following loop, run on the exact Bernstein tensor of
the prepared polynomial over $[0,1]^n$. The boxes it visits form a
binary tree: the full box $[0,1]^n$ is the \emph{root}, the two halves
of each bisected box are its \emph{children}, and a box that is not
bisected further is a \emph{leaf}. For the current box:
\begin{enumerate}[nosep,leftmargin=6em,labelwidth=5em,labelsep=1em,labelindent=0pt,align=left]
\item[\textsc{(certify)}] If every entry of its Bernstein tensor is
  $\ge0$, the box is certified by \eqref{eq:bernpos}; it becomes a leaf
  and the loop moves on to the next unfinished box.
\item[\textsc{(refute)}] If a \emph{corner} entry is negative, that
  entry \emph{equals} the value of $P$ at the corresponding corner
  point, so the corner is emitted as an exact counterexample and the
  entire claim fails immediately.
\item[\textsc{(split)}] Otherwise --- some non-corner coefficient is
  negative while all corner values are nonnegative, so on this box
  the coefficients are still too far from the function values they
  approximate to settle the sign either way --- choose one axis,
  bisect the box along it by the de Casteljau recursion, and recurse
  on the two halves.
\end{enumerate}
The boxes are processed depth-first: after a bisection both halves are
fully finished (certified, or subdivided further and finished) before
any other box is taken up, so at any moment only the chain of boxes
from the root down to the current one is held in memory --- memory
grows with the depth of the tree, not with its total size. The choice
of split axis in
\textsc{(split)} affects efficiency only, never soundness: any policy
that keeps refining every axis inherits the termination guarantee of
the first bullet, because the sub-boxes shrink in all directions and
their corners become dense. A strictly positive $P$ reaches an
all-certified tree at finite depth, and a $P$ that is negative
somewhere eventually places a corner inside the negative region and is
refuted. (The only inputs that stall are the touching zeros excluded
above.) 

Two split policies are used in production. The default
bisects the longest edge of the current box, which refines every axis
in rotation. For the large-spin-tail sectors --- the certificates with by far
the largest trees --- the split axis is instead chosen adaptively:
bisect the axis along which the current tensor has the greatest number
of sign changes between entries that are adjacent in that index, since
separating coefficients of mixed sign into different halves is
what lets \textsc{(certify)} succeed. Because this adaptive rule alone
could in principle keep selecting the same axis, every fourth
bisection reverts to the longest edge, so all axes keep refining and
the termination argument above still applies. On the hardest large-spin-tail
pieces the adaptive rule shrinks the certified tree by orders of
magnitude --- on one production piece from $134{,}695$ boxes to $37$
(App.~\ref{app:bernstein}). Each certificate's log records which
policy produced it; the distinction carries no weight in the proof,
because acceptance rests on accounting that is blind to the policy.
Concretely, before \code{True} is emitted both implementations require:
the work queue is empty; every generated box received an explicit
verdict (the fail-closed control flow); and the box counters satisfy
the binary-tree identities
$\text{certified leaves}=\text{splits}+1$ and
$\text{processed}=\text{leaves}+\text{splits}$ --- which, together
with the empty queue, hold precisely when the processed boxes form a
complete binary subdivision of the root box with every leaf certified,
so a dropped or unfinished sub-box cannot pass silently. The log
records the tree size and the minimal Bernstein coefficient over all
certified leaves; on the production large-spin-tail pieces the independent
C++ implementation and the reference implementation were additionally
cross-validated against each other (App.~\ref{app:bernstein}).

Why this tool and not the alternatives? \emph{Interval arithmetic}
bounds the range of an expression over a box by propagating an interval
through the expression tree one operation at a time. Its classic
weakness --- the \emph{dependency problem} --- is that every
occurrence of a variable is treated as if it varied over its interval
independently of every other occurrence: the correlation between them
is forgotten at each step. Our
certified quantities are extreme instances of this shape: sums
of terms with coefficients spanning tens of orders of magnitude that
cancel down to a small positive margin, so the width of the returned
enclosure is set by the sizes of the individual \emph{terms}, not by
the small residual that is the actual answer. Measured on our tail
claims, the enclosure of the minimum is wider than the true minimum by
factors of $10^{16}$--$10^{28}$; since it therefore straddles zero, it
certifies nothing.
The Bernstein transform simply does not have the problem: it is one
exact rational-linear map applied to the polynomial's \emph{global}
coefficient vector, so every cancellation happens exactly, once,
during the transform, and what is inspected afterwards are exact
coefficients of the same polynomial in a different basis. 

\emph{Numerical sums-of-squares} would prove a target polynomial
$P\ge0$ by finding, with a floating-point SDP solver, a
positive-semidefinite Gram matrix $G$ with $P=m^{T}Gm$, where $m$ is
the vector of monomials up to half the degree; to turn that into a
proof one must then repair the numerics --- replace the float $G$ by a
nearby exact rational $G'\succeq0$ reproducing $P$ \emph{exactly},
absorbing the rounding error into the semidefinite slack. In our initial tests of this method both processes failed: the auxiliary SDP inherits the problem's
$10^{70}$ coefficient spread and fails to converge reliably, and when
it does, the rounding error exceeds the available slack, so no exact
$G'$ survives the repair.\footnote{We stopped pursuing these alternative methods when we found the Bernstein basis method worked very well. There could be solutions to making these alternative methods work that we never explored.}

The Bernstein route has no optimization, no evaluation, and no rounding anywhere --- it is a
single exact linear transformation followed by sign inspection, with
subdivision as a fallback that, in our applications, is almost never
needed (most certificates below pass on the first box). And it is
classical (de Casteljau, B\'ezier, half a century of computer-aided
design): every property the certificates rely on --- non-negativity of
the basis, endpoint coefficients equal to corner values, subdivision
as an exact rational map, coefficients converging to values under
subdivision and hence termination on strictly positive inputs --- is a
textbook fact. Everything the certificates rely on here is classical.

Implementation details are in App.~\ref{app:bernstein}; three disciplines
apply throughout. \emph{Fail-closed}: every code path that does not end in
an explicit exact verdict --- an unexpected symbol in a coefficient, a factor
whose sign is not decided, a crash --- cannot be interpreted as a success; \emph{Falsified-input tests}: each
certifier is routinely run on deliberately falsified claims (the same
object with its sign flipped, or a box known to contain an allowed point)
and must \emph{fail} on them --- a checker that could never say no would
otherwise go unnoticed. \emph{Coverage check}: before the final
\code{True} conclusion is emitted, it is re-verified that the sub-certificates tile the
original claim with no gaps.

\paragraph{Complexity: the general problem is hard.} Deciding whether
a polynomial is nonnegative on a box is co-NP-hard already for
quadratics on the unit cube, and globally already for quartics, by the
co-NP-completeness of matrix copositivity \cite{MurtyKabadi1987}; no
subdivision scheme has a polynomial worst case on arbitrary inputs, and
on an input whose minimum is exactly zero the loop of
Section~\ref{sec:bernstein} never certifies and ends only by exhausting its subdivision budget. We never solve
the general problem. Every instance we certify has a favorable
property: it is strictly signed, with a \emph{relative margin}
\[
  t\;:=\;\frac{\min_{[0,1]^n}P}{\max_i|b_i|}
\]
bounded away from zero, where the $b_i$ are the Bernstein coefficients
on the root box (the search stage of Section~\ref{sec:regionpractice}
maximized this kind of slack). The released certifiers do not
record $t$ itself; each log records \code{minB}, the smallest Bernstein
coefficient over all certified leaves, an absolute number in the units
of the cleared polynomial (printed as a float, since it is a diagnostic
and enters no verdict). Because $P$ is bounded below by the smallest
coefficient on every leaf, \code{minB} is a lower bound on
$\min P$, and $t\ge\code{minB}/\max_i|b_i|$.

\paragraph{Complexity: the guarantee for margined instances.} For such
inputs the size of the certificate is bounded by a theorem, not a
heuristic. Bernstein coefficients on a sub-box of width $h$ differ from
the values of $P$ at the corresponding grid points by at most $C\,h^{2}$,
where $C$ is set by the degree and by the second differences of the
root-box coefficients, measured in the same units as $t$. All
coefficients on a sub-box are therefore nonnegative as soon as
$C h^{2}<t$, i.e.\ at a bisection \emph{depth} of
$\tfrac12\log_2(C/t)=O(\log(1/t))$ per axis, and the uniform grid at
that depth has at most $(C/t)^{n/2}$ boxes in $n$ variables (explicit
certificate-size bounds of this form: univariate
\cite{BoudaoudCarusoRoy2008}, multivariate on simplices
\cite{Leroy2012}). Adaptive subdivision can only do better, since it
refines only where coefficients are still of mixed sign. For our
polynomials the constant $C$ is enormous --- the coefficients span
$10^{70}$ in magnitude --- so the uniform-grid bound is astronomically
loose, and it is adaptivity, not the theorem, that keeps the runs cheap:
only neighbourhoods of near-minima refine.

\paragraph{Complexity: the practical cost.} The work of one run is
\[
  \sum_{\text{splits}}(\text{tensor volume})\times d_k
  \;+\;\sum_{\text{leaves}}(\text{tensor volume}),
\]
in big-integer operations. Here the \emph{tensor volume}
$\prod_{i=1}^{n}(\deg_i+1)$ is the number of stored coefficients, set by
the cleared polynomial's per-variable degrees (a large-spin-tail
example: degrees $(39,54,9)$ in $(\ell,\tau,\Ds)$ give
$40\cdot55\cdot10\approx2\times10^{4}$ entries); $d_k$ is the degree of
the one direction bisected at a split, because subdividing a fiber of
length $d_k+1$ along that direction costs $d_k(d_k+1)/2$ averagings,
i.e.\ $\sim d_k$ per coefficient; and a leaf costs one pass of sign
inspection. The number of directions $n$ therefore enters through the
volume (a product of $n$ factors) and through the worst-case exponent
$n/2$ above, not through the per-split cost. Each averaging is an
integer operation whose bit-length grows a priori linearly with depth
(each split homogenizes by $2^{d_k}$); the per-child content-GCD
extraction of App.~\ref{app:bernstein} keeps it flat in practice. The
recorded spread on production pieces: the hardest double-tail piece of
the large-spin tail took $134{,}695$ boxes and $259$\,s on 32 cores under the longest-edge
policy, and $37$ boxes and $2$\,s under the sign-guided policy; most
claims certify on the first box.  The Bernstein checks of a tile take
minutes; what makes a full tile cost about $25$ core-hours
(App.~\ref{app:timings}) is the exact assembly of the large-spin-tail
objects, not the certification. The one-variable
Sturm checks cost $O(d^{2})$ arithmetic operations on exact rationals,
and the coverage check is trivial.

\subsection{Regions: one more Bernstein direction}
\label{sec:regioncert}

Everything so far was at fixed $\Ds$. The region upgrade just reuses the
above tools. By \eqref{eq:exactpoly} the external dimension enters the
crossing vectors polynomially (degree $\le\Lambda$), and the expensive
part of the symbolic build (the block-derivative tables) is
$\Ds$-independent, so a \emph{single} build with $\Ds$ symbolic yields the
master object
\begin{equation}
	\widehat V_m(\Ds;L,\tau)\ \in\ \Q(\sqrt2)(L,\tau)[\Ds],\qquad
	\deg_{\Ds}\le\Lambda,
\end{equation}
from which every sector of every tile follows by substitution. A tile
$T=[\slo,\shi]\times[\elo,\ehi]$ is then certified by the same list of
certificates as a point, with the $\Ds$-interval as one extra affine
Bernstein direction\footnote{For the simplest, one-variable certificates the released scripts use Sturm's theorem instead of Bernstein subdivision: the number of real
roots of an exact polynomial in an interval is computable by exact sign
counting, so ``no roots in $[\slo,\shi]$ and positive at one point''
certifies positivity. This is an equally exact, finite decision
procedure, faster on these tiny instances; a one-variable Bernstein
certificate decides the same claims. Sturm's theorem is listed
in the trusted computing base (Section~\ref{sec:trustedbase}) because it is what
the production scripts run.} The certified truncation-error argument of Section~\ref{sec:trueblocks} runs unchanged with the extra variable.

\begin{sloppypar}
\paragraph{What one tile's certificate consists of.} Table~\ref{tab:inventory}
lists the certificates of one tile, with the sector domain each one
covers and the assumption of Section~\ref{sec:setup} that fixes that
domain. The frame, band and gap tiles (Table~\ref{tab:tiles}) run them in three \emph{layers}, one
script invocation each: the first six rows (seventeen certificates,
\code{SECTORS=2d} of \code{certify\_ring\_tile\_L9.wl}), the per-spin row
(eighty-eight, \code{SECTORS=spins}), and the large-spin tail (eleven,
\code{certify\_tile\_fartail\_L9.wl}, the last piece certified as the
diagonal-split pair of App.~\ref{app:bernstein} through the C++
checker). The strip tiles run the first seven rows through
\code{certify\_strip\_tile\_L9.wl}, with the window certificate in the
cleared form of Lemma~\ref{lem:clearedwindow} below, and replace the last
row by four three-variable pieces, \code{tail3d.aa}, \code{.ab},
\code{.ba}, \code{.bb}, covering $\ell\in[14,100]$ or $\ell\ge100$ times
$\tau\in[1,\tausplit]$ or $\tau\ge\tausplit$.
\end{sloppypar}

\begin{table}[t]\centering\footnotesize
\begin{tabular}{@{}l p{5.9cm} l l r@{}}
\toprule
certificate & sector domain (assumption) & variables & envelope & count\\
\midrule
identity & $\alpha\cdot\vec{\mathrm{Id}}(\Ds)>0$ on $\Isig$ & $\Ds$, Sturm & --- & 1\\
stress tensor & $(\Delta,\ell)=(3,2)$ \ (c) & $\Ds$, Sturm & window & 1\\
scalar window & $\ell=0$, $\tau\in\Ieps$ \ (b) & $(\tau,\Ds)$ & window & 1\\
scalar continuum & $\ell=0$, $\tau\in[3,\tausplit]$; $\tau\ge\tausplit$ \ (b) & $(\tau,\Ds)$ & window; tail & 2\\
spin $2$ & $\tau\in[\tfrac72,\tausplit]$; $\tau\ge\tausplit$ \ (c),(d) & $(\tau,\Ds)$ & window; tail & 2\\
spins $4,6,\dots,12$ & $\tau\in[1,\tausplit]$; $\tau\ge\tausplit$ \ (unitarity) & $(\tau,\Ds)$ & window; tail & 10\\
spins $14,16,\dots,100$ & same, one pair per even spin & $(\tau,\Ds)$ & window; tail & 88\\
large-spin tail & real $L\ge\ell_*=100$; $\tau$ in $[1,3]$, $[3,10]$, $[10,30]$, $[30,10^2]$, six ratio-$3$ pieces from $10^2$ to $10^5$, and $[10^5,\infty]$ & $(L,\tau,\Ds)$ & tail & 11\\
\bottomrule
\end{tabular}
\caption{The certificates of one frame, band or gap tile ($\tausplit=100$,
$\Lambda=9$), $116$ in all. ``Window'' and ``tail'' name the two
branches of the truncation envelope (Theorem~\ref{thm:envelope}); every
domain is closed and compactified as in Section~\ref{sec:compact}, so
the tail pieces include the faces $\tau=\infty$ and $L=\infty$. The
strip tiles replace the last row as described in the text.}
\label{tab:inventory}
\end{table}

\paragraph{The large-spin-tail certificate.} The pieces of the last row
are built differently from the others, because the symbolic-$L$ normal
form of the three-variable object is prohibitively expensive (the
\code{Together} step takes about $19$ hours per sector).
Instead, \code{certify\_tile\_fartail\_L9.wl} reconstructs the cleared
polynomial by exact interpolation in $L$. Let $D(L,\tau)$ be the common
denominator of the level increments, shared by all components and
factored into pole-tower factors that are certified positive on
$\{L\ge100,\tau\ge1\}$, and let
$G=\sum_m\hat\alpha_m\widehat V^{(\kappa)}_m
-16\sum_m|\hat\alpha_m|\widehat T^{(\kappa),5/2}_m$ be the margined
combination of Proposition~\ref{prop:margin} (tail branch). The script
first verifies, componentwise and symbolically, that $G\,D$ is a
polynomial in $(L,\tau,\Ds)$,\footnote{This step is what turns the slice
comparison below into a proof. The argument there is a degree count: two
polynomials in $L$ of bounded degree that agree on more spin slices than
that degree are equal. It fails for a rational function, which can agree
with a polynomial at any finite number of points without being equal to
it. Had one tower factor, say $2L+\tau+5$, been missing from $D$, it
would be invisible on every spin slice --- at $L=101$ it is just the
$L$-free factor $207+\tau$ --- so each slice could look like a perfectly
good polynomial in $(\tau,\Ds)$ and the interpolated $P$ could match all
of them, while differing from the true object between and beyond the
lines, in particular towards the face $L=\infty$. The check is
structural and is made on the shipped expressions themselves: for each
of the fifteen components, $D$ times the truncated vector and $D$ times
the tail majorant are brought to a single reduced fraction with $L$,
$\tau$ and $\Ds$ symbolic, and the script requires that the denominator
is a number and the numerator a polynomial, aborting otherwise. The
reduction is an exact identity, so a weak reduction could cause a false
abort but never a false pass. That $D$ clears every pole could also be
argued from the recursion, whose poles are all tower factors
\eqref{eq:towers}; the multiplicities make that argument delicate, and
the a priori degree it gives is loose ($68$ against the actual $37$),
which is why the statement is verified on the data instead.} and records its $L$-degree $D_L$ ($37$ on
the production tiles, against the a priori bound $68$ read off the table
degrees). It then evaluates $G\,D$ exactly on the $42$ spin slices
$L=101,\dots,142$ --- the spin slice at an integer $L=\ell$ is the
restriction $(G\,D)(\ell,\tau,\Ds)$, a polynomial in $(\tau,\Ds)$ with
$\Q(\sqrt2)$ coefficients --- and solves the Vandermonde system for every
$\tau^{i}\Ds^{j}$ coefficient, which yields a polynomial $P(L,\tau,\Ds)$
of degree at most $41$ in $L$. Both $P$ and $G\,D$ are polynomials in
$L$ of degree at most $\max(41,D_L)$, so agreement on more slices than
that proves them identical up to the one constant the normalization of
each slice can introduce: the script checks the $(\tau,\Ds)$-identity on $27$
further slices not used in the construction, $69$ slices in all, and confirms that $P/(G\,D)$ is a single nonzero constant $d_0$ (negative on every production tile). \footnote{The
witnesses are needed because each slice polynomial is normalized by an
exact constant of $\Q(\sqrt2)$ left over after the reduction (typically
of the form $c/(2+\sqrt2)^k$), so the slice comparison by itself
establishes only $P(\ell,\tau,\Ds)=s_\ell\,(G\,D)(\ell,\tau,\Ds)$ with
some number $s_\ell$ for each slice. If $s_\ell$ were allowed to vary
with $\ell$ the conclusion would fail in a way that matters for signs:
for instance $s_\ell=120-\ell$ would make $P=(120-L)\,G\,D$, a
polynomial that passes every slice comparison and changes sign at
$L=120$, so that positivity of $P$ on $L\ge100$ would say nothing about
$G$. The script therefore evaluates the ratio exactly, at a rational
point of each slice where the denominator does not vanish, and requires
all $69$ values to be equal and nonzero; then $P-d_0\,G\,D$ is a
polynomial of $L$-degree at most $\max(41,D_L)$ vanishing on $69$ slices,
hence zero. The sign of $d_0$ is decided by the exact rule of
Section~\ref{sec:qsqrt2} and fixes which sign of $P$ is certified. No
floating-point number enters any of these steps; decimal values appear
only in log messages.}
The eleven Bernstein certificates then run on
$\mathrm{sign}(d_0)\,\mathrm{sign}(D)\,P$ in the compactified variable
$u=100/L\in[0,1]$, whose face $u=0$ is $L=\infty$. The shipped
\code{fartail/} polynomials are stored copies of $P$ that spare a
verifier the Vandermonde solve and nothing else: a stored polynomial is
not trusted, and when the script loads one it verifies the
$(\tau,\Ds)$-identity against freshly assembled slice polynomials on all
$69$ slices, the $42$ construction slices included, before any Bernstein
certificate runs. There is no mode that skips this check.

For tiles whose scalar window reaches close to the unitarity bound, the window certificate runs in the cleared form forced by the pole at the unitarity bound:
\begin{lemma}[cleared window certificate]\label{lem:clearedwindow}
Let $\alpha$ be a candidate functional for a tile with scalar window
$[w,\ehi]\subset[\tfrac12,1]$, and let $V^{(28)}$, $\widehat T$ be the
truncated vector and tail majorant of the certificate chain. If the
single Bernstein certificate
\[
  (2\tau-1)\,\big(\alpha\cdot V^{(28)}\big)
  \;-\;\What\;\big(\textstyle\sum_m|\alpha_m|\,\widehat T_m\big)\;\ge\;0
  \qquad\text{on } \{\ell=0\}\times[w,\ehi]\times\Isig
\]
holds, then $\alpha\cdot V^{\rm true}\ge0$ for every scalar with
$\tau\in[w,\ehi]$ if $w>\tfrac12$, and with $\tau\in(\tfrac12,\ehi]$ if
$w=\tfrac12$, and every $\Ds$ in the tile's external-dimension interval.
\end{lemma}
\begin{proof}
By Lemma~\ref{lem:strip} the true tail obeys
$|\alpha\cdot(V^{\rm true}-V^{(28)})|\le
 c_\infty\sum_m|\alpha_m|\widehat T_m$ on the strip, and
$(2\tau-1)c_\infty\le\What$. Multiplying the tail bound by
$(2\tau-1)>0$ and combining with the display gives
$(2\tau-1)\,\alpha\cdot V^{\rm true}\ge0$ on $[w,\ehi]$, and the
clearing factor is strictly positive there when $w>\tfrac12$, which
gives the closed conclusion for the two upper bands. The certificate
runs on the \emph{closed} box $[w,\ehi]$ with $w=\tfrac12$ for the bottom
band, but there the conclusion is the open interval: at the corner the
clearing factor vanishes, so the certified inequality no longer
implies anything about $\alpha\cdot V^{\rm true}$ there. What the
certificate proves on that face is positivity of the finite limit
$\lim_{\tau\to1/2^+}(2\tau-1)\,\alpha\cdot V^{(28)}$ --- the pole's
residue, which is the vector of the would-be null descendant rather
than the contribution of any physical operator at the bound (see the
discussion after Lemma~\ref{lem:corner}). The endpoint
$\De=\tfrac12$ --- a scalar at the unitarity bound --- is excluded by
a separate, purely analytic argument
(Lemma~\ref{lem:corner}), so the covering's total claim is closed.
\end{proof}

The external-dimension continuum,
historically the motivation for elaborate corner-plus-remainder
constructions at the search stage, is at the certificate stage
\emph{just another compact direction}.

\section{The truncation envelope: construction and proof}
\label{sec:trueblocks}

Section~\ref{sec:envelopestatement} reduced loophole 3 to a single
certified inequality per claim --- positivity by the margin,
$\mathrm{num}(W^{(\kappa)})-M\cdot\mathrm{den}(W^{(\kappa)})\ge0$ ---
for a certified error bound $M$.
It remains to build $M$, and this section does this. It is
organized so that every \emph{conceptual} ingredient of the bound
appears explicitly in the text, while everything that is instead a
mechanical exact computation --- a symbolic identity expanded term by
term, or the sign of one explicit function certified on one closed box
--- is stated once and delegated to a named, released machine
proof. Section~\ref{sec:truncationerrorwhat} introduces the objects and
the three pieces of technology the construction uses.
Section~\ref{sec:telescope} develops the conceptual core and proves the
main lemmas. Section~\ref{sec:certinputs} lists, as numbered statements
[C1]--[C7], every machine-certified input the proof consumes, each with
the artifact that proves it. Section~\ref{sec:marginproof} assembles
lemmas and inputs into the envelope theorem and the margin bound in
the form the production certificates consume.

\subsection{Concepts and technology}
\label{sec:truncationerrorwhat}

\paragraph{The object to be bounded.} The truncated evaluator
$h^{(\kappa)}$ agrees with the true block $h$ through radial order
$r^{\kappa}$: adding pole levels beyond $\kappa$ changes only radial orders
beyond $\kappa$ \cite{Kos:2013tga,Penedones:2015aga}. The componentwise error
is therefore the \emph{$\rho$-truncation tail}: the omitted radial levels
$\mathrm{lev}>\kappa$ of the pole recursion. It
starts at order $r^{\kappa+1}$ and is evaluated at
$r=\rstar=3-2\sqrt2<1/5$, so each successive omitted order is
suppressed by one further factor of $\rstar$ relative to the previous
one --- each pair of even orders by $\rstar^2<1/25$ (only even orders
occur for identical externals, for the reason explained at the start of
Section~\ref{sec:telescope}). Geometric smallness is thus structural; the work of this
subsection is to turn it into an exact inequality with explicit,
certified constants. The objects to be bounded are the diagonal
coefficients $c_m(\Delta,\ell)$ of the Setup, eq.~\eqref{eq:Hdiag},
governed by the radial Casimir recurrence \eqref{eq:seventerm}.
Everything below is about bounding these coefficients and then
propagating the bound, through explicit linear maps, to the derivative
tables the certificates consume.

\paragraph{Technology 1: sign-discarded majorants.} The basic device is
elementary: given an object built by a finite composition of linear
steps with signed, rational weights, rerun the \emph{same} construction
with every signed numerator weight replaced by its absolute value,
keeping every denominator symbolically with its factors certified
positive on the domain; the
triangle inequality, applied at each step, then proves that the
resulting nonnegative object dominates the true one in absolute value.
For the production tail maps the replacement rules are, concretely:
(i) the signed Pochhammer polynomials of the leg-factor assembly ---
polynomials in the \emph{external} dimension $\Ds$ --- are replaced by
their coefficientwise absolute values, a majorant for $\Ds\ge0$ since
every monomial is then nonnegative there; (ii) every fixed constant of
the construction is replaced by its exact absolute value: the
derivative-map constants, elements of $\Q(\sqrt2)$ whose signs the
generator decides exactly (aborting if it cannot), and the integer
coefficients of the transverse Casimir recursion; (iii) the one signed
$(\Delta,\ell)$-dependent combination, the Casimir eigenvalue, is
majorized as
$|\ell(\ell+1)+\Delta^2-3\Delta|\le\hat\lambda:=\ell(\ell+1)+\Delta^2+3\Delta$,
with $\hat\lambda$ entering the rerun construction as a numerator
weight; (iv) denominators are never majorized: they are kept
symbolically, reduced, and their factors certified positive on the
closed domain, so dividing by them preserves domination. No other
operation intervenes. The released generator (\code{build\_That\_L9\_k28.wl})
implements these rules, so the load-bearing inequality
$|\omega|\le\hat\omega$ is inherited step by step.
In the single-correlator proof this device is used twice, both times
on \emph{finite} constructions: applied to the finite linear maps that
carry diagonal coefficients into derivative tables, it will produce
the majorant tail maps of Lemma~\ref{lem:majorant}; applied to the
polynomial coefficients of the block recurrence, it produces the
majorized induction step behind the geometric tail bound
(Proposition~\ref{prop:geo3} through the certified input [C4]).

\paragraph{Technology 2: compactified domains and certified signs.}
Every parameter region we use is one of the closed, compactified boxes
of Section~\ref{sec:compact}: the substitutions $w=\tau_{\min}/\tau$,
$1/\ell$ map each sector, \emph{including its points at infinity}, onto
a closed box on which all our functions are rational with certified
nonvanishing denominators. On such a box, the sign of any explicit
rational function is machine-decidable by the exact Bernstein
certificates of Section~\ref{sec:certarch}. Each analytic requirement of
Technology~1 --- a denominator separation, a numerator majorant, a
coefficient bound at a fixed level --- will be posed as the positivity
of one explicit function on one closed box, and thereby becomes a
finite computation. The finitely many such facts the proof
consumes are collected as [C1]--[C7] in Section~\ref{sec:certinputs};
nothing else is taken from a machine.

\paragraph{Technology 3: two regimes, two kinds of coefficient bound.}
Recall from Section~\ref{sec:regionpractice} that the certificates
organize the positivity constraints into \emph{sectors}: the scalar
window, the scalar continuum above the gap, and one sector per spin,
each a family of exchanged operators parameterized by the twist
$\tau=\Delta-\ell$ and the (real) spin $\ell$. The coefficient bounds
below are stated once for all sectors, on the closed quadrant
$\{\ell\ge0,\ \tau\ge1\}$ that contains every one of them (extended
below twist one, where only the scalar interval lives, in
Lemma~\ref{lem:strip}).
The certificates split every sector at a fixed twist $\tausplit$ into a
compact \emph{twist window} $\tau\le\tausplit$ and a noncompact
\emph{large-twist tail} $\tau\ge\tausplit$, and the two parts need different
coefficient bounds (we take $\tausplit=100$; any comparably large value
would serve, see Table~\ref{tab:params}). On a compact window the best
possible bound is $(\Delta,\ell)$-dependent: an explicit function
$\hat c(\Delta,\ell)$ with $0\le c_m\le\hat c$ for all $m$. On the noncompact tails the same sharp envelopes are unbounded (they grow with $\Delta$); what survives is
a \emph{uniform geometric} bound $c_m\le C_0\lambda^m$, with one
numeric constant $C_0$ valid for every operator of the closed quadrant
at once. The conceptual core of Section~\ref{sec:telescope} is these two
bounds: the sharp envelope, obtained from the zero-sum form of the
coefficient recurrence that we establish there, and the geometric bound,
obtained from a finite family of certified levels plus a certified
induction. Both are stated for the \emph{diagonal} coefficients only, and that suffices: in
the pipeline the entire derivative table is a linear image of the
diagonal coefficients alone --- transverse derivatives are
reconstructed from diagonal data by the transverse Casimir recursion, with
explicit rational weights and no other radial input --- so bounds on
the $c_m$ control, through computable linear maps, everything the
certificates consume. (Positivity of the full radial--angular expansion
\cite{Hogervorst:2013sma} additionally makes the diagonal a majorant of
every angular slice.) The discussion below concerns identical external
scalars in $d=3$; the later mixed-correlator
Section~\ref{sec:mixedoutlook} describes the
unequal-external extension.

\paragraph{The numerical parameters, collected.} The proof below fixes
a handful of numerical thresholds --- a truncation order here, a level
bound there --- and quoted in isolation these values can look
arbitrary. They are not: each is the value of a named parameter, fixed
once, by the one consideration cited at the place listed in
Table~\ref{tab:params}, and the statements below quote the parameters
by name. Only one comparison connects two rows of the table: the
truncation order $\kappa$ is chosen by comparing the truncation
envelope, which shrinks like $\rstar^{\kappa+2}$ (the first omitted
even level), against the
functionals' certified margin floor $t^\star$ --- and that single
comparison, made at the end of Section~\ref{sec:marginproof}, is what
this whole subsection builds toward. Every other value is either forced
by the certified inputs it summarizes ($\levstar$, $\mgeo$) or a
convenient choice that nearby values would serve equally well
($\lambda$, $\tausplit$), as discussed where each is
fixed.

\begin{table}[t]\centering\small
\begin{tabular}{@{}l p{4.9cm} p{3.1cm} p{4.6cm}@{}}
\toprule
symbol & meaning & production value & where fixed and used\\
\midrule
$\Lambda$ & derivative order of the functionals ($N_\Lambda=15$
components) & $9$ & search stage,
Section~\ref{sec:regionpractice}\\[2pt]
$\kappa$ & radial truncation order & $\kappa_{\rm cert}=28$,
the order at which every certified statement is made ([C6]);& ``Choosing the truncation order $\kappa$'',
Section~\ref{sec:marginproof}\\[2pt]
$\tausplit$ & twist splitting every sector into window and large-twist tail &
$100$ & Technology~3 above; Theorem~\ref{thm:envelope}\\[2pt]
$\levstar$ & contraction threshold: smallest omitted level at which the
weighted tail ratio $\hat q<\tfrac1{10}$ is certified & $21$ (the
$\tau\ge1$ quadrant and the cleared scalar strip alike) &
Lemmas~\ref{lem:ratio}, \ref{lem:strip};
[C2]; eq.~\eqref{eq:geomdom}\\[2pt]
$\lambda$, $C_0$ & base and constant of the uniform geometric bound
$c_m\le C_0\lambda^m$ & $\lambda=3$,\ \ $C_0=\tfrac1{2\cdot10^{13}}$ &
Proposition~\ref{prop:geo3} and the discussion following it\\[2pt]
$\mgeo$ & first radial level covered by the geometric bound & $29$
& Proposition~\ref{prop:geo3}: window certificates [C3] (even levels
$30$--$46$), induction [C4] (levels $\ge47$)\\[2pt]
$t^\star$ & sampled relative margin floor of the searched
functionals (a search-stage quantity, not certified) & $\ge1.9\times10^{-4}$ (worst tile) &
Section~\ref{sec:regionpractice}; consumed by the choice of
$\kappa$\\
\bottomrule
\end{tabular}
\caption{The numerical parameters of the truncation-error proof, their
production values, and where each is fixed. Every threshold quoted in
the lemmas and certified inputs of this section is one of these.}
\label{tab:params}
\end{table}

\paragraph{Routes to coefficient control.} There are several ways to
obtain the coefficient bounds themselves. We list three in ascending
tightness; the production certificates use only the third. The first
two are included for context: they give a less technical, more
intuitive way to see why rigorous coefficient control is available at
all.
\begin{itemize}
	\item \emph{A majorant from unitarity alone.} Any single block is bounded
	by any unitary correlator that contains it, divided by the corresponding
	OPE coefficient. Embedding $(\Delta,\ell)$ as a double-twist operator of a
	generalized free field with external dimension
	$\Delta_\phi=(\tau-2n)/2\in[\tfrac12,\tfrac32)$, $n=\lfloor(\tau-1)/2\rfloor$,
	gives $c_m\le K(\tau,\ell)\,(\Delta+m)^{4\Delta_\phi}$ with everything
	closed-form: polynomial growth at the true geometric rate, uniformly in
	twist. This is already enough to run the entire certification with no
	unproven input, at the price of computing a high-order-in-$\kappa$ reference truncation. We will however not use this and instead continue to prove stronger bounds that lower the required reference block order $\kappa$.
	\item \emph{Dimensional reduction.} Every $3d$ block is a
	\emph{positive} sum of $2d$ blocks of equal or lower spin
	\cite{Hogervorst:2016hal},
	$g^{(3d)}_{\Delta,\ell}=\sum_{n\ge0}\sum_{j\le\ell}\mathcal A_{n,j}\,
	g^{(2d)}_{\Delta+2n,j}$, with $\mathcal A_{n,j}\ge0$ \emph{proven} above
	the unitarity bound and, for $3d\to2d$, given in closed form as ratios of
	Pochhammer symbols. The $2d$ blocks are elementary hypergeometrics whose
	radial coefficients are explicit and lie in $(0,1]$. \footnote{We verified both the
		positivity and the closed form against an independent construction (an
		angular extension of the \cite{Kos:2013tga} recursion).}
	\item \emph{A zero-sum recurrence.} The diagonal coefficients satisfy a
	three-term recurrence whose coefficients sum to zero, so that their
	\emph{differences} obey a first-order recurrence with a closed-form
	solution, proving the $c_m$ monotone and bounded --- no growth at all. This is the sharpest
	statement of the three, it is proven in full below, and it is the one
	the production margin consumes.
\end{itemize}

\subsection{The conceptual core: the zero-sum recurrence, the difference
recurrence, and the two coefficient bounds}
\label{sec:truncationproof}
\label{sec:telescope}

One preliminary disposes of the odd radial levels. For equal external
dimensions the odd-level coefficients vanish identically: every
odd-shift residue of the pole recursion carries the external-difference
factor $\big(\tfrac{1-n}{2}\big)_n$ (evaluated at
$\Delta_{12}=\Delta_{34}=0$), and for odd $n$ the $n$ consecutive
arguments $\tfrac{1-n}2,\tfrac{3-n}2,\dots,\tfrac{n-1}2$ of this
Pochhammer symbol run through $0$, so the factor vanishes; the
remaining recursion shifts are even. The diagonal expansion is
therefore a series in $\rho^2$, and we work with the even subsequence
$b_s=c_{2s}$ throughout.

The main conceptual input of the coefficient bounds is then a hidden
structure in the radial Casimir recurrence \eqref{eq:seventerm} of the
Setup: restricted to the even levels, it reduces to a three-term
recurrence whose middle coefficient is exactly minus the sum
of the outer two. With $\tau=\Delta-\ell$,
\begin{equation}
	P_0(s)\,b_s-\big(P_0(s)+P_2(s)\big)\,b_{s+1}+P_2(s)\,b_{s+2}=0,
	\label{eq:telescoperec}
\end{equation}
\begin{equation}
	P_0(s)=(s+\tfrac32)(s+\Delta)(s+\tfrac{\tau+1}2)(s+\tfrac{\Delta+\ell+2}2),
	\quad
	P_2(s)=(s+2)(s+\Delta+\tfrac12)(s+\tfrac{\tau+2}2)(s+\tfrac{\Delta+\ell+3}2).
	\label{eq:telescopeP}
\end{equation}
Because the middle coefficient is $-(P_0+P_2)$, the three coefficients
sum to zero; we call this the \emph{zero-sum form}. A zero-sum recurrence
annihilates constants, so it is really a recurrence for the differences
$D_s=b_{s+1}-b_s$ --- the \emph{difference recurrence} --- which is of
first order and solvable in closed form:
\begin{equation}
	D_{s+1}=\frac{P_0(s)}{P_2(s)}\,D_s
	\quad\Longrightarrow\quad
	D_s=\big(c_2-1\big)\,
	\frac{(\tfrac32)_s(\Delta)_s(\tfrac{\tau+1}2)_s(\tfrac{\Delta+\ell+2}2)_s}
	{(2)_s(\Delta{+}\tfrac12)_s(\tfrac{\tau+2}2)_s(\tfrac{\Delta+\ell+3}2)_s}
	\;\sim\;\frac{{\rm const}}{s^{2}} .
	\label{eq:telescopeD}
\end{equation} Consequently, once the
sign of $c_2-1$ is known,\footnote{The passage from
\eqref{eq:telescoperec} to the solved form divides by $P_2(s)$, so we
record where that division is legitimate. On the closed
domains the certificates use --- the quadrant
$\{\ell\ge0,\ \tau\ge1\}$ and the closed scalar strip
$\{\ell=0,\ \tau\in[\tfrac12,1]\}$ --- every factor
of $P_2(s)$, and likewise of $P_0(s)$, is bounded below by an explicit
positive constant for every integer $s\ge0$ (the worst cases, attained
at the free-scalar corner $\Delta=\tau=\tfrac12$ of the strip:
$s+2\ge2$, $s+\Delta+\tfrac12\ge1$, $s+\tfrac{\tau+2}2\ge\tfrac54$,
$s+\tfrac{\Delta+\ell+3}2\ge\tfrac74$). On these domains the
solved form is therefore proven, finite, and positive (the anchor $c_2$
itself has a pole at the corner, so on the strip the statements are
made for the cleared quantities $(2\tau-1)b_s$, Lemma~\ref{lem:strip})
--- and it is
what Lemmas~\ref{lem:monotone}--\ref{lem:strip} and the closed
form of App.~\ref{app:closedform} consume. The division fails only on
the hyperplanes where a factor of $P_2$ vanishes at some integer
$s_0\ge0$ --- $\Delta\in\{-\tfrac12,-\tfrac32,\dots\}$,
$\tau\in\{-2,-4,\dots\}$, or $\Delta+\ell\in\{-3,-5,\dots\}$, all deep
in the non-unitary region --- and there the solution becomes
\emph{resonant}, in the same sense as in Frobenius' method for ODEs
when indicial exponents differ by integers. Concretely: at the level
$s_0$ with $P_2(s_0)=0$, the recurrence \eqref{eq:telescoperec} loses
its $b_{s_0+2}$ term, so instead of determining the next coefficient
it imposes a consistency condition on the previous two,
$P_0(s_0)\,(b_{s_0}-b_{s_0+1})=0$; solving for $b_{s_0+2}$ near such a
hyperplane divides by the vanishing $P_2(s_0)$, and the coefficient
generically diverges. This is not a pathology of the method but a
feature of the block itself: the three hyperplane families are
sublattices of the familiar lattice of block poles in $\Delta$ --- the
very poles whose residues drive the pole recursion
(App.~\ref{app:blocks}) --- so the divergence of the coefficients is
the block's own pole, seen level by level. In the closed form the same
loci appear automatically, as the standard requirement that
hypergeometric denominator parameters avoid nonpositive integers; the
one remaining parameter pole of the closed form, at $\Delta=\tfrac12$,
enters through the anchor $c_2$ and is the free-scalar pole discussed
in Lemma~\ref{lem:strip}.} the $c_m$
are \emph{monotone and bounded}:
\begin{equation}
	1=c_0\le c_2\le c_4\le\cdots\nearrow c_\infty(\Delta,\ell)
	=1+(c_2-1)\;{}_5F_4\!\Big[\genfrac{}{}{0pt}{}
	{1,\ \tfrac32,\ \Delta,\ \tfrac{\tau+1}2,\ \tfrac{\Delta+\ell+2}2}
	{2,\ \Delta+\tfrac12,\ \tfrac{\tau+2}2,\ \tfrac{\Delta+\ell+3}2};1\Big],
	\label{eq:cinfty}
\end{equation}
with the hypergeometric value finite (term ratio $1-2/s+O(s^{-2})$) and its
tail elementary to enclose. No growth at all: $3d$ diagonal blocks sit
between the $2d$ (decaying) and $4d$ (linearly growing) closed
forms, and behave as $h(r,1)\sim c_\infty/(1-r)$. Summed, the identity
in fact yields a closed form for the diagonal block itself --- an
independent re-derivation, as we learned afterwards, of a formula of
Rychkov and Yvernay \cite{Rychkov:2015lca}; the derivation of the closed form and its 
identification with the known form are discussed in App.~\ref{app:closedform}. 

We now formally prove the reduction from the radial Casimir recurrence of
the Setup to this three-term recurrence, then harvest its consequences
as the lemmas the margin bound consumes.

\begin{lemma}[reduction to a zero-sum three-term recurrence]
\label{thm:telescope}
	For all $(\Delta,\ell)$ in the open unitary region, and extending by
	continuity to the conserved-current boundary points as described below,
	the diagonal radial coefficients of the equal-external $3d$ block ---
	i.e.\ the solution of the radial Casimir recurrence \eqref{eq:seventerm}
	with the block's initial data --- satisfy
	the zero-sum three-term recurrence
	\eqref{eq:telescoperec}--\eqref{eq:telescopeP} for all integers $s\ge-1$,
	with $b_{-1}:=0$
	(evenness of the radial level is already absorbed in $b_s=c_{2s}$).
	For non-integer $\ell$ the coefficients mean the sequence defined by the
	recurrence \eqref{eq:seventerm} from the same initial data; every identity
	used below is rational in $\ell$, so it holds for that sequence at every
	real $\ell$, and this is the object the real-spin certificates of
	Section~\ref{sec:allspins} consume.
\end{lemma}
\begin{proof}
	Each step is an exact identity, machine-checkable from the printed
	formulas.
	(i) Start from the radial Casimir recurrence \eqref{eq:seventerm} of the
	Setup, satisfied by the coefficients $c_m$ of the stripped diagonal
	series \eqref{eq:Hdiag}, with the $\tilde A_k$ explicit polynomials
	in $(m,\Delta,\ell)$.

	(ii) As explained at the start of this subsection, for equal externals
	only even levels are populated, so we restrict the
		recurrence to the even subsequence $b_s:=c_{2s}$: at $m=2s$ the terms
	with odd shift multiply vanishing coefficients $c_{2s+k}$, $k$ odd, and
	the terms $k=0,2,4,6$ that remain involve $b_s,\dots,b_{s+3}$. Let $S$
	denote the shift operator in this even level
	variable, $(Sb)_s=b_{s+1}$. The restricted recurrence can then be written
	as a finite-difference operator
	\[
	(T_4 b)_s=0,
	\]
	meaning explicitly that $(T_4b)_s$ is a linear combination of
	$b_s,b_{s+1},b_{s+2},b_{s+3}$ with polynomial coefficients in
	$(s,\Delta,\ell)$. $T_4$ is \emph{defined} by this even-level
	restriction; the analogous restriction to the odd levels does not
	admit the factorization below (we checked), which is immaterial here
	because the odd coefficients vanish identically. Define the three-term operator $T_3$ by
	\[
	(T_3 b)_s
	:=
	P_0(s)b_s-\big(P_0(s)+P_2(s)\big)b_{s+1}+P_2(s)b_{s+2},
	\]
	with $P_0,P_2$ as in \eqref{eq:telescopeP}. Thus the equation
	$(T_3b)_s=0$ is the desired recurrence
	\eqref{eq:telescoperec}. The symbolic check proves the operator identity
	\[
	T_4=16\,(S-1)\circ T_3
	\]
	as an identity of rational functions in $(s,\Delta,\ell)$. Consequently any coefficient sequence satisfying the
	restricted recurrence $T_4b=0$ has
	\[
	(S-1)(T_3b)=0,
	\]
	so $w_s:=(T_3b)_s$ is independent of $s$. The initial value
	$w_0$ is computed symbolically from $c_0$, $c_2$ and $c_4$ --- in the
	released checker from the symbolic-spin pole recursion, which gives
	the block's coefficients by T2; the recurrence \eqref{eq:seventerm}
	started from $c_0=1$ with the zero padding gives the same values ---
	and is
	identically zero as a rational function of $(\Delta,\ell)$. Therefore
	$w_s=0$ for every $s\ge0$, which proves $(T_3b)_s=0$, i.e.\
	\eqref{eq:telescoperec}. The instance $s=-1$ follows in the same way:
	\eqref{eq:seventerm} holds for all $m\ge-7$ with the zero padding, so
	the factorized identity holds at $s=-1$ as well and gives
	$w_{-1}=w_0=0$, which with $b_{-1}=0$, $b_0=1$ reads
	$P_2(-1)\,(b_1-1)=P_0(-1)$ and fixes $c_2$ (eq.~\eqref{eq:c2closed}). The symbolic identities used in this proof are
	proven by the released file \code{hor\_telescope\_checker.wl}. That file
	derives the radial Casimir recurrence for the stripped coefficients from the HOR ODE, verifies the even-level factorization
	$T_4=16(S-1)\circ T_3$, verifies that constants are annihilated by
	$T_3$, and verifies the initial identity $w_0=0$, all as exact rational
	identities in the symbolic variables $(s,\Delta,\ell)$.\footnote{The
		release also includes an independent transcription check,
		\code{check\_telescope\_sympy.py}, which re-checks the exported exact
		objects in a second computer-algebra system. Versions and expected outputs
		for both files are listed in \code{scripts/CHECKERS.md}; their
		checksums are in \code{MANIFEST.sha256}.}
	
	The identities are proved first on the open unitary region where the
	recursion formulas are manifestly nonsingular. Boundary cases such as
	conserved currents are handled as in the exceptional Frobenius cases of
	\cite[App.~B]{Hogervorst:2013kva}: the block is defined by continuity in
	$\Delta$, after cancelling the common factor that makes the unreduced
	recursion look singular. For the certified objects this is
	implemented algebraically: the rational expressions are reduced before
	evaluation, and the denominator-factor checks certify that no uncancelled
	pole remains on the closed sector domains.
\end{proof}

The theorem's consequences, in the form the margin bound consumes, are
the following four lemmas. Where a lemma rests on a
machine-certified fact, the fact is cited as [C$\cdot$] and stated
precisely in Section~\ref{sec:certinputs}; everything else is proved
here.

\begin{lemma}[the monotone bounded branch]\label{lem:monotone}
On the closure of the sector $\{\ell\ge0,\ \tau\ge1\}$,
\[
  1=c_0\le c_2\le c_4\le\cdots\nearrow c_\infty(\Delta,\ell),
\]
with $c_\infty$ the finite hypergeometric value \eqref{eq:cinfty}.
\end{lemma}
\begin{proof}
The certified input [C1] gives $1\le c_2\le\tfrac32$ on the sector
closure. By \eqref{eq:telescopeD} every difference $D_s$ is the product
of $c_2-1\ge0$ with a ratio of Pochhammer symbols whose arguments are
all positive on the sector, so every $D_s\ge0$: the sequence is
monotone. Summing the closed-form differences gives \eqref{eq:cinfty},
whose series converges (term ratio $1-2/s+O(s^{-2})$), so the sequence
is bounded and increases to $c_\infty$.
\end{proof}

\begin{lemma}[the $(\Delta,\ell)$-dependent envelope]\label{lem:envelopeC}
On the same sector closure,
\[
  \sup_m c_m=c_\infty\;\le\;1+(c_2-1)(1+B)\;\le\;\tfrac32+\tfrac B2,
  \qquad B^2=\tfrac38\,\Delta(\tau+1)(\Delta+\ell+2).
\]
\end{lemma}
\begin{proof}
Writing
$t_s=\prod_{j<s}P_0(j)/P_2(j)$ for the terms of \eqref{eq:cinfty}, each
of the four Pochhammer factor pairs obeys the elementary square-root
bound
\begin{equation}
  \frac{j+a}{j+a+\tfrac12}\;\le\;\sqrt{\frac{j+a}{j+a+1}}
  \qquad(\text{the squares differ by }\tfrac14\ge0),
  \label{eq:sqrtstep}
\end{equation}
and in the product of the right-hand sides over $j<s$ consecutive
factors cancel, leaving $\sqrt{a/(a+s)}$. With the
four parameters $a\in\{\tfrac32,\Delta,\tfrac{\tau+1}2,
\tfrac{\Delta+\ell+2}2\}$, all $\ge1$ on the sector, this gives
$t_s\le\min\big(\sqrt{3/(2s+3)},\,B/(s+1)^2\big)$ with
$B^2=3\Delta(\tau+1)(\Delta+\ell+2)/8$; since
$\sum_{s\ge1}(s+1)^{-2}=\pi^2/6-1<1$ and $t_0=1$,
\begin{equation}
  \sup_m c_m=c_\infty\;\le\;1+(c_2-1)(1+B)\;\le\;\tfrac32+\tfrac B2 ,
  \label{eq:csupmain}
\end{equation}
the last step using $c_2-1\le\tfrac12$ from [C1].
\end{proof}

The derivative tables weight each radial level by nonnegative
combinations of the falling factorials
$\mathrm{FP}(\Delta{+}\mathrm{lev},M):=(\Delta{+}\mathrm{lev})
(\Delta{+}\mathrm{lev}{-}1)\cdots(\Delta{+}\mathrm{lev}{-}M{+}1)$,
$0\le M\le\Lambda$, times $\rstar^{\mathrm{lev}}$
(Lemma~\ref{lem:leveldec} below). The contraction constant that sizes the truncation order is the worst ratio of two consecutive
populated levels under these weights,
\begin{equation}
  \hat q(\mathrm{lev})\;:=\;\sup_{(\Delta,\ell)}\ \max_{0\le M\le\Lambda}\
  \frac{c_{\mathrm{lev}+2}}{c_{\mathrm{lev}}}\cdot
  \frac{\mathrm{FP}(\Delta{+}\mathrm{lev}{+}2,M)}{\mathrm{FP}(\Delta{+}\mathrm{lev},M)}
  \cdot\rstar^{2},
  \label{eq:qhatdef}
\end{equation}
the supremum taken over the closed domain in question, and
$\hat q_\omega(\mathrm{lev})$ the same quantity without the factor
$c_{\mathrm{lev}+2}/c_{\mathrm{lev}}$ (the ratio of the weights alone);
$\hat q$ without argument means $\sup_{\mathrm{lev}\ge\levstar}\hat
q(\mathrm{lev})$, and likewise for $\hat q_\omega$.

\begin{lemma}[the uniform two-level ratio]\label{lem:ratio}
On the sector closure, $c_{2s+2}/c_{2s}\le\tfrac{591}{500}$ for
$s\ge10$. Consequently the weighted two-level tail ratio
\eqref{eq:qhatdef} obeys, for $\Lambda=9$ and
$\mathrm{lev}\ge\levstar=21$ --- the contraction threshold of
Table~\ref{tab:params}, covering every omitted level of any build with
$\kappa\ge20$ and in particular of the production build ---
\[
  \hat q\;\le\;\tfrac{591}{500}\cdot\tfrac{92}{35}\,\rstar^{2}
  \;<\;\tfrac1{10}.
\]
\end{lemma}
\begin{proof}
Since $c_{2s+2}/c_{2s}=1+D_s/c_{2s}$ with $c_{2s}\ge1$, and
$D_s=(c_2-1)\,t_s\le\tfrac12\,t_s$ by [C1] and
\eqref{eq:telescopeD}, the $t_s$ bound of the previous proof gives,
at $s\ge10$,
$c_{2s+2}/c_{2s}\le1+\tfrac12\sqrt{3/23}<\tfrac{591}{500}$
uniformly on the sector. The weights' two-level ratio is bounded
factor-wise: each $\mathrm{FP}(\Delta{+}\mathrm{lev},M)$ is positive
on the domain (its arguments satisfy $\Delta+\mathrm{lev}\ge22>\Lambda$
there), and
\[
  \frac{\mathrm{FP}(x{+}2,M)}{\mathrm{FP}(x,M)}
  =\frac{(x{+}2)(x{+}1)}{(x{-}M{+}2)(x{-}M{+}1)}
\]
is increasing in $M$ and decreasing in $x$, hence at most
$\mathrm{FP}(24,9)/\mathrm{FP}(22,9)=\tfrac{24\cdot23}{15\cdot14}
=\tfrac{92}{35}$ for $x=\Delta+\mathrm{lev}\ge22$ ($\Delta\ge1$ on
the sector, $\mathrm{lev}\ge21$); a nonnegative combination of terms
whose individual ratios obey a common bound obeys it too. The product of the coefficient-ratio bound
$\tfrac{591}{500}$, the falling-factorial weight-ratio bound
$\tfrac{92}{35}$, and the radial factor $\rstar^{2}$ is the display.
\end{proof}
\noindent Every constant in the three lemmas above is either derived
in the proofs or is the certified input [C1]. The fourth lemma covers
the scalar strip below twist one. One new feature appears there and
nowhere else in the pavement: the crossing vectors have a pole at the
free-scalar corner $\tau=\tfrac12$ (the single $k{=}1$ $\beta$-tower
residue of the pole recursion), so on the strip every certified
quantity is \emph{cleared} --- multiplied by the factor $(2\tau-1)$,
positive on the open interval, which cancels the pole exactly and lets
the certificates run on the \emph{closed} interval down to
$\tau=\tfrac12$.

\begin{lemma}[the scalar strip, cleared]\label{lem:strip}
On the scalar strip $\{\ell=0,\ \tau\in(\tfrac12,1]\}$ the weighted
two-level tail ratio obeys $\hat q<\tfrac1{10}$ for every
$\mathrm{lev}=2s\ge20$ --- the same contraction threshold
$\levstar=21$ as the $\tau\ge1$ quadrant, so a single split level
serves both domains --- and the coefficient envelope obeys
\[
  \sup_m c_m(\tau)\;=\;c_\infty(\tau)\;\le\;\frac{\What}{2\tau-1},
  \qquad \What=\tfrac{5330666037}{5\cdot10^{9}}<1.0662 ,
\]
with both statements certified on the cleared closed strip
$[\tfrac12,1]$, i.e.\ after multiplication by $(2\tau-1)$. Together
with unitarity, which forces $\tau\ge1$ at $\ell\ge1$, the quadrant of
Lemma~\ref{lem:monotone} and this strip cover every unitary operator:
there is no gap between the two certificate domains.
\end{lemma}
\begin{proof}
For $\ell\ge1$ unitarity already enforces $\tau\ge1$; the only unitary
territory below twist one is the scalar interval, where everything is
univariate and the initial coefficient has the closed form
$c_2(0,\tau)=\tau(3\tau-1)/\big((1{+}\tau)(2\tau-1)\big)$ --- making
the pole at the free-scalar corner explicit. The diagonal coefficients
$b_s(\tau)=c_{2s}(0,\tau)$ obey the difference recurrence in the form
$b_s=b_{s-1}+(c_2-1)t_s$,
$t_{s+1}=t_s\,P_0(s{-}1,\tau)/P_2(s{-}1,\tau)$, with the quartics
$P_0,P_2$ of Lemma~\ref{lem:monotone}, so one certifies the
\emph{actual} pair ratios rather than an analytic majorization chain.
The clearing identity $(c_2-1)(1{+}\tau)(2\tau{-}1)=(1{-}\tau)^2$
shows every cleared quantity $w_s=(2\tau-1)b_s$ is pole-free on the
closed interval and increasing in $s$
($w_s-w_{s-1}=(1-\tau)^2t_s/(1+\tau)\ge0$), while $t_s$ is positive
and decreasing ($0<P_0<P_2$ coefficient-wise in
$(s,\tau-\tfrac12)$). The certificates [C2] establish
$\hat q\le\tfrac9{100}$ at $\mathrm{lev}=20,22,24,26$ uniformly on the
closed cleared strip. The certified object is the product of the
coefficient ratio $b_{s+1}/b_s$, the rational majorant
$\tfrac{2944}{10^5}>\rstar^2$, and the \emph{exact} two-level ratio of
the falling-factorial weights at the worst weight order,
$\mathrm{FP}(x{+}2,9)/\mathrm{FP}(x,9)=(x+2)(x+1)/\big((x-7)(x-8)\big)$
with $x=\tau+2s$; [C2] also certifies that this ratio increases with the
weight order $M\le9$ on the domain, so that $M=9$ majorizes every weight
that occurs, and that it decreases in $s$. (A power of $(s{+}1)/s$ is
not a majorant of this ratio, for the reason given in the proof of
Lemma~\ref{lem:ratio}.) The three monotonicity facts and the decrease of
the weight ratio then force $\hat q(s,\cdot)\le\hat q(13,\cdot)$ for
all $s\ge13$, extending the display to every $\mathrm{lev}\ge20$. For
the envelope, $c_\infty-b_{12}=(c_2-1)\sum_{j\ge13}t_j$; the certified
tail-domination inequality of [C2],
$P_2(s,\tau)(s+\tfrac32)^2\ge P_0(s,\tau)(s+\tfrac52)^2$ for all
$s\ge12$ \emph{including the $s=\infty$ boundary} (compactified, with
the double zero at the boundary divided out exactly), gives
$t_{j+1}/t_j\le(j+\tfrac12)^2/(j+\tfrac32)^2$; multiplied out from $j=13$,
with consecutive factors cancelling, this is
$t_j\le t_{13}\,(\tfrac{27}{2})^2(j+\tfrac12)^{-2}$ for $j\ge13$, and
the elementary integral bound
$\sum_{j\ge13}(j+\tfrac12)^{-2}\le\int_{12}^\infty(x+\tfrac12)^{-2}dx
=\tfrac2{25}$ gives $\sum_{j\ge13}t_j\le\tfrac{729}{50}\,t_{13}$.
Combining with the certified $(2\tau-1)b_{12}\le\tfrac{21}{20}$,
$(1-\tau)^2\le\tfrac{17}{100}(1+\tau)$ and
$t_{13}\le\tfrac{6509}{10^6}$:
\[
 (2\tau-1)\,c_\infty\;\le\;\tfrac{21}{20}
   +\tfrac{17}{100}\cdot\tfrac{729}{50}\cdot\tfrac{6509}{10^6}
   \;=\;\What .
\]
All inequalities are exact Bernstein certificates over $\Q(\sqrt2)$; the
same run also rejects a deliberately sign-flipped input, as a test that
the checker can fail.
\end{proof}
\noindent This strip-plus-quadrant shape is forced: at \emph{real}
$L>0$ and $\tau<1$ (outside unitarity) the initial coefficient bound
fails ($c_2$ is already below $\tfrac35$ by $L=10$), so no
real-$(L,\tau)$
quadrant below twist one can carry these certificates. The cleared
form is likewise forced: the stripped initial coefficient
$(2\tau-1)\,c_2=\tau(3\tau-1)/(1+\tau)$ is manifestly pole-free and
positive on the closed interval $[\tfrac12,1]$, and the same clearing
applies to every $c_m$ and to the window certificates of the strip
tiles, whose scalar windows reach below twist one
(Lemma~\ref{lem:clearedwindow}).\\

This completes the coefficient bounds on the twist windows: the pieces
$\tau\le\tausplit$ of every sector, which are compact in the twist
direction.
What remains are the noncompact pieces --- the twist tails
$\tau\ge\tausplit$ of each sector, and the spin tail, where $\ell$ is
unbounded. The $(\Delta,\ell)$-dependent envelope of
Lemma~\ref{lem:envelopeC} cannot be used: its constant $B$ grows without
bound along both noncompact directions, and no  $(\Delta,\ell)$--independent
constant can replace it, since pointwise $c_m\nearrow c_\infty$ and
$c_\infty$ is unbounded over the quadrant $\{\ell\ge0,\ \tau\ge1\}$. We must therefore find a new $(\Delta,\ell)$-independent bound --- and the price is that the bound must be allowed to grow with the radial level $m$: the geometric
form $c_m\le C_0\lambda^m$ of Technology~3 is forced. For
\emph{fixed} level $m$, the coefficient $c_m(\Delta,\ell)$ is one explicit rational function of $(\Delta,\ell)$, produced by finitely many steps of the (validated) recursion. A claim of the form ``$c_m\le C_0\lambda^m$ on the closed, compactified quadrant'' is therefore a finite machine-decidable computation. That disposes of any \emph{finite} set of levels. The \emph{infinite} remainder is disposed of by one more certified positivity statement, an
induction step: the recurrence expresses $c_m$ through the seven
preceding levels with explicit coefficient functions, so ``if levels
$m-7,\dots,m-1$ obey the geometric bound then so does level $m$'' is
again a sign condition on explicit recursion data, certified once,
uniformly in $m$ and in $(\Delta,\ell)$ over the quadrant. The
induction step is the following lemma, stated here because
Proposition~\ref{prop:geo3} consumes it (Section~\ref{sec:mixedoutlook}
later certifies the same inequality for unequal externals,
Lemma~\ref{prop:mixedmajorant}; nothing from that section is used here).

\begin{lemma}[the equal-external induction step]\label{lem:inductstep}
Write the eight-term recurrence \eqref{eq:seventerm} at equal externals
as $\sum_{k=0}^{7}\tilde A_k(m;\Delta,\ell)\,c_{m+k}=0$. Uniformly over
real $\ell\ge0$, $\tau\ge\tfrac12$ and all $m\ge40$,
\[
  \sum_{k<7}\big|\tilde A_k(m;\Delta,\ell)\big|\,3^k
  \;\le\;\big|\tilde A_7(m;\Delta,\ell)\big|\,3^7,
  \qquad \tilde A_7(m;\Delta,\ell)\neq0 .
\]
Consequently, if $|c_{m+k}|\le C\,3^{m+k}$ for $k=0,\dots,6$ then
$|c_{m+7}|\le C\,3^{m+7}$.
\end{lemma}
\begin{proof}
Machine-checked: after passing to the sector-nonnegative coordinates
$m=40+M$, $\tau=\tfrac12+t$, $\ell\ge0$ and replacing each
$|\tilde A_k|$, $k<7$, by its coefficientwise absolute-value majorant,
every coefficient of the certificate polynomial
$|\tilde A_7|\,3^7-\sum_{k<7}|\tilde A_k|\,3^k$ is a nonnegative rational
with positive constant term, and $|\tilde A_7|$ has the same manifest
positivity. The checker proves this with the external difference kept
as a variable $\delta\in[-\tfrac32,0]$, expanded in the exact Bernstein
basis in $\delta$, so that the statement holds uniformly on that
interval; the equal-external case is its endpoint $\delta=0$ (released
checker \code{mixed\_bounds\_checker.wl}, stage [M4], channel
$\Delta_{12}=\Delta_{34}=\delta$, with $\tau_0=\tfrac12$ and $m_0=40$;
this is the certified input [C4]). The consequence is the triangle inequality applied
to $c_{m+7}=-\tilde A_7^{-1}\sum_{k<7}\tilde A_k\,c_{m+k}$.
\end{proof}

\noindent The result:

\begin{proposition}[certified base-$3$ coefficient bound]
\label{prop:geo3}
Uniformly on the closed, compactified tail-sector superset
$\{\ell\ge0,\ \tau\ge1\}$,
\begin{equation}
  0\ \le\ c_m(\Delta,\ell)\ \le\ C_0\,3^{\,m}
  \qquad\text{for all }m\ge\mgeo=29,\quad
  C_0=\tfrac1{2\cdot10^{13}}.
  \label{eq:geo3}
\end{equation}
\end{proposition}
\begin{proof}
Nonnegativity at every \emph{real} $\ell\ge0$ is the paper's own
Lemma~\ref{lem:monotone}: the certified $c_2\ge1$ on the sector
quadrant (input [C1]), the difference recurrence, and the positivity of
its Pochhammer factors give $c_{2s}\ge1$ on the stated domain. (Radial positivity, trusted-base
input T1, asserts the same at \emph{physical integer} spins and serves
as a crosscheck; the statement consumed here is the real-$\ell$ one.)
Odd levels vanish identically by equal-external parity (the generating
code checks this and aborts otherwise). For the even levels $30\le m\le46$ the inequality
is the certified input [C3]: a pair of direct exact Bernstein
certificates per level (the two charts of the compactified spin
domain), every one passing on its first box. For $m\ge47$ it follows by
strong induction from Lemma~\ref{lem:inductstep}: the certified levels
$40\le j\le46$ of [C3] serve as the base window, and the lemma's
consequence advances the bound one level at a time; together with
$c_m\ge0$ from the first paragraph this is \eqref{eq:geo3}.
\end{proof}

\noindent The base $3$ is a convenient choice: any base
$\lambda<1/\rstar=3+2\sqrt2$ makes the tail sums below converge (they
carry $\rstar^{\,m}$ per level), and $3$ leaves both a comfortably
certified induction step and the geometric tail ratio
$3\rstar\approx0.51$. The constant $C_0$, in turn, is fixed directly
by the finitely many certified coefficient functions at the crossover
levels, with no chain of intermediate inequalities contributing slack
--- which is how a $(\Delta,\ell)$-uniform constant as small as $C_0$ is
possible at all. And the Proposition adds nothing to the trusted computing base:
both of its ingredients are sign checks on explicit recursion output,
the objects, domains, and checker that every other certificate
in the pipeline already uses.

The division of labor between the two coefficient bounds is thus by
domain. On every compact
twist window, production consumes the closed-form differences: the
monotone branch (Lemma~\ref{lem:monotone}) and its limit $c_\infty$
(Lemma~\ref{lem:envelopeC}) power the sharp $c_\infty$-envelope
certificates there, and the strip constants of Lemma~\ref{lem:strip}
extend them below twist one. On the noncompact twist and spin tails,
Proposition~\ref{prop:geo3} takes over. And the radial Casimir recurrence
whose induction step [C4] certifies is the same recurrence from which
the zero-sum form itself was distilled (its equal-external
even-level factorization $T_4=16(S-1)\circ T_3$ is
Lemma~\ref{thm:telescope}), each independently re-derived from the
printed HOR operator by its own released checker.

Two further lemmas carry the coefficient bounds from the diagonal
$c_m$ to the derivative tables the certificates actually consume. The
recursion is organized by radial level, and the truncated tables
inherit that structure exactly:
\begin{lemma}[level decomposition]\label{lem:leveldec}
For every derivative component $m$ and truncation order $\kappa$,
\begin{equation}
  \widehat V^{(\kappa)}_m \;=\; \sum_{\mathrm{lev}\le\kappa}
  c_{\mathrm{lev}}\;\omega_{\mathrm{lev},m},
  \label{eq:leveldec}
\end{equation}
where $c_{\mathrm{lev}}$ are the diagonal radial coefficients of the
stripped block (nonnegative by input T1) and the \emph{level maps} $\omega_{\mathrm{lev},m}$ are
explicit rational functions of $(\Delta,\ell)$, independent of the block:
the level-$\mathrm{lev}$ seed enters the derivative table through
factorial-power weights
$\mathrm{FP}[\Delta{+}\mathrm{lev},M]\,\rstar^{\mathrm{lev}-M}$, where
$\mathrm{FP}[x,M]=x(x-1)\cdots(x-M+1)$ is the falling factorial, and the
Casimir reconstruction of transverse derivatives, a finite composition
of linear steps with explicit rational coefficients.
\end{lemma}
\begin{proof}
In the pole recursion every residue term carries the explicit power
$r^{n}$ of its level shift, so the coefficient of $r^{\mathrm{lev}}$ in
the diagonal series --- the $c_{\mathrm{lev}}$ of \eqref{eq:Hdiag} ---
collects the terms of total level shift $\mathrm{lev}$, and the
truncation at $\kappa$ leaves every level $\mathrm{lev}\le\kappa$
intact. The derivative table is obtained from the diagonal series by
three operations, each linear and each acting level by level:
differentiation in $r$ at $r=\rstar$, which turns the level-$\mathrm{lev}$
monomial into the falling-factorial weights; the transverse Casimir
recursion, whose coefficients depend on $(\Delta,\ell)$ but not on the
level; and the leg-factor convolution. Applying this composition to the
single monomial of level $\mathrm{lev}$ defines $\omega_{\mathrm{lev},m}$,
and linearity gives \eqref{eq:leveldec}. The maps are constructed
exactly by the released generators.
\end{proof}
The discard-the-signs device of
Technology~1, applied to this finite linear construction, produces the
tail majorants:
\begin{lemma}[majorant domination]\label{lem:majorant}
Define the \emph{majorant level maps} $\hat\omega_{\mathrm{lev},m}$ by
rerunning the (finite, explicit) linear construction of
$\omega_{\mathrm{lev},m}$ with these replacements (as
implemented in the released generator \code{build\_That\_L9\_k28.wl}):
(i) every signed Pochhammer polynomial of the leg-factor assembly ---
a polynomial in the external dimension $\Ds$ --- is replaced by its
coefficientwise absolute value, a majorant on $\Ds\ge0$, where every
monomial is nonnegative; (ii) every fixed constant is replaced by its
exact absolute value: the derivative-map constants, elements of
$\Q(\sqrt2)$ whose signs the generator decides exactly and aborts if
it cannot, and the integer coefficients of the transverse Casimir
recursion; (iii) the one signed $(\Delta,\ell)$-dependent combination,
the Casimir eigenvalue, is majorized as
\[
  \big|\ell(\ell+1)+\Delta^2-3\Delta\big|\;\le\;
  \hat\lambda:=\ell(\ell+1)+\Delta^2+3\Delta ,
\]
the instance $|x-y|\le x+y$ for $x=\ell(\ell+1)+\Delta^2\ge0$,
$y=3\Delta\ge0$, with $\hat\lambda$ entering the rerun construction as
a numerator weight; (iv) denominators are never majorized: they are
kept symbolically --- each factor certified positive on the closed
domain by the factor check (at $d=3$ the transverse-recursion
denominators are positive integers) --- so dividing by them preserves
domination; and no other operation occurs. Then, for every level and component,
$|\omega_{\mathrm{lev},m}|\le\hat\omega_{\mathrm{lev},m}$ pointwise on
the certified domains, and the tails
$\widehat T^{(\kappa)}_m=\sum_{\mathrm{lev}>\kappa}\hat\omega_{\mathrm{lev},m}$
--- and their geometrically weighted variants
$\widehat T^{(\kappa),x}_m=\sum_{\mathrm{lev}>\kappa}
x^{\mathrm{lev}}\hat\omega_{\mathrm{lev},m}$ --- possess closed-form
Pochhammer--geometric expressions, shipped as the exact artifacts of
[C6] and re-derived independently of \code{Sum} by the finite
difference-identity checker \code{verify\_tail\_sums.wl} (both weights
$x=1$ and the geometric $x=\tfrac52$; since $\hat\omega_{\mathrm{lev}}$
already carries $\rstar^{\mathrm{lev}}$, the checker's summation
variable is $q=x\,\rstar$, i.e.\ $q=\rstar$ and $q=\tfrac52\rstar$;
see the artifact map).
\end{lemma}
\begin{proof}
$\omega_{\mathrm{lev},m}$ is a finite composition of linear steps
(Lemma~\ref{lem:leveldec}; the chain is printed in
App.~\ref{app:blocks}) whose weights are explicit \emph{rational
functions} on the sector domain --- rational in $(\Delta,\ell)$, not
fixed numbers, which the majorant
construction must respect. Each weight's denominator is first reduced
and its factors certified positive on the closed sector (the
factor-by-factor check of Section~\ref{sec:compact});
$\hat\omega_{\mathrm{lev},m}$ is then well defined, and the
triangle inequality applied at each step of the composition gives
the pointwise domination. The infinite level sums are summed
symbolically: each level enters through factorial-power weights times
$\rstar^{\mathrm{lev}}$ (Lemma~\ref{lem:leveldec}), so the sums over
$\mathrm{lev}>\kappa$ are Pochhammer--geometric series with explicit
closed forms; the released objects and the scripts that construct them
are fixed in [C6].
\end{proof}

\subsection{The certified inputs}
\label{sec:certinputs}

Everything the proof still needs beyond the lemmas above is one of the
following finitely many machine-certified statements. Each is the
positivity of an explicit rational object on a closed set of the form
$\{\ell\ge\ell_0,\ \tau\ge\tau_0\}$ or
$\{\ell=0,\ \tau\in[\tau_1,\tau_2]\}$: the checker applies the
substitutions of Section~\ref{sec:compact} to that set and
certifies positivity on the resulting closed Bernstein box. Closed
boxes include the infinity faces, so each certified claim holds on the
full closure of the stated set; each checker log records the box, the
substitution, and the denominator-factor positivity checks, and every
checker aborts on any undecided comparison rather than passing it. File names refer to the data release; expected
invocations and outputs are listed in its \code{CHECKERS.md}.

\begin{description}[nosep]
\item[{[C1]} initial coefficient bounds.]
$c_2-1\ge0$ and $\tfrac32-c_2\ge0$ on the compactified closure of
$\{\ell\ge0,\ \tau\ge1\}$, each certified on a single closed Bernstein
box. Proved by \code{c2\_bounds.wl}. (Consumed by
Lemmas~\ref{lem:monotone}--\ref{lem:ratio}.)
\item[{[C2]} strip certificates (cleared).]
On the \emph{cleared closed} strip $\{\ell=0,\ \tau\in[\tfrac12,1]\}$
--- every quantity multiplied by $(2\tau-1)$, which cancels the single
free-scalar pole, so the closed endpoint $\tau=\tfrac12$ is included
--- exact Bernstein certificates of: (i) the weighted two-level tail
ratio $\hat q\le\tfrac9{100}<\tfrac1{10}$ at
$\mathrm{lev}=20,22,24,26$, formed with the exact falling-factorial
weight ratio $\mathrm{FP}(x{+}2,9)/\mathrm{FP}(x,9)$ at $x=\tau+\mathrm{lev}$
(certified to be the largest over the weight orders $M\le9$ and
decreasing in the level), together with the monotonicity facts
($P_0,P_2>0$; $P_2-P_0\ge0$ coefficient-wise; the clearing identity
$(c_2-1)(1{+}\tau)(2\tau{-}1)=(1{-}\tau)^2$) that promote the last of
these to \emph{all} $\mathrm{lev}\ge20$; (ii) the cleared envelope
bound $(2\tau-1)\,c_\infty(\tau)\le\What$, assembled from
$(2\tau-1)\,c_{24}\le\tfrac{21}{20}$,
$(1-\tau)^2\le\tfrac{17}{100}(1+\tau)$, the certified
$t_{13}\le\tfrac{6509}{10^6}$, and the tail-domination certificate
$P_2(s,\tau)(s+\tfrac32)^2\ge P_0(s,\tau)(s+\tfrac52)^2$ for all
$s\ge12$ \emph{including the $s=\infty$ boundary} (compactified, with
the double zero at the boundary divided out exactly); and (iii) the
anchor $1\le c_2\le9$ on the \emph{uncleared} strip
$\{\ell=0,\ \tau\in[\tfrac{51}{100},1]\}$ (the search floor; $c_2$
itself diverges as $\tau\to\tfrac12$), with the closed form
$c_2(0,\tau)=\tau(3\tau-1)/((1{+}\tau)(2\tau-1))$. Items (i)--(ii) are
proved by \code{strip\_tight\_cert.wl} on the cleared closed strip
$[\tfrac12,1]$; item (iii) by \code{c2\_strip\_cert.wl}.
(Consumed by Lemmas~\ref{lem:strip} and~\ref{lem:clearedwindow}.)
\item[{[C3]} geometric window certificates.]
For each even level $30\le m\le46$:
$C_0\,3^m-c_m(\Delta,\ell)\ge0$ on the compactified closure of
$\{\ell\ge0,\ \tau\ge1\}$, with $C_0=\tfrac1{2\cdot10^{13}}$, certified
in two chart pieces per level ($\ell\in[0,1]$ and the compactified
$\ell\ge1$, both with $\tau\ge1$ compactified) --- eighteen Bernstein
certificates in all; the odd levels vanish identically (the generating
code checks this and aborts otherwise). Proved by the \code{cgeo\_window\_cert} artifacts (exact
generator, tensors, verdicts); the same directory ships the exact
coefficient-function dump \code{cm47\_rf.tgz} from which the window
tensors are regenerated. These coefficient functions come from the
saturated recursion in real spin; that they equal the per-spin
coefficients at every even spin, at every level through $46$, is
Lemma~\ref{lem:interp}, whose exhaustive part is the run of
\code{satcheck\_L9.wl} with \code{KAPPA=46}. (Consumed by
Proposition~\ref{prop:geo3}, levels and base window.)
\item[{[C4]} the induction step.]
The inequality of Lemma~\ref{lem:inductstep}: uniformly over real
$\ell\ge0$, $\tau\ge\tfrac12$ and all $m\ge40$,
$\sum_{k<7}|\tilde A_k|\,3^k\le|\tilde A_7|\,3^7$ with $\tilde A_7\ne0$,
for the equal-external eight-term recurrence. Proved by
\code{mixed\_bounds\_checker.wl} (stage [M4], whose interval
$\delta\in[-\tfrac32,0]$ contains the equal-external point $\delta=0$,
$\tau_0=\tfrac12$); the
eight-term recurrence itself is derived from the printed HOR diagonal
ODE independently by that file and by
\code{hor\_telescope\_checker.wl} (with a second-CAS transcription
check, \code{check\_telescope\_sympy.py}). (Consumed by
Proposition~\ref{prop:geo3}, all levels $\ge47$.)
\item[{[C5]} fold comparisons.]
For each of the ten factorial-power orders $r=0,\dots,9$ --- the
possible degrees of the falling-factorial weights
$\mathrm{FP}[\Delta{+}\mathrm{lev},M]$ of Lemma~\ref{lem:leveldec} at
$\Lambda=9$ --- consumed by the tail
construction:
$16\,S^{(5/2)}_r(\Delta)\ge C_0\,S^{(3)}_r(\Delta)$ on
$\Delta\in[1,\infty)$, where
$S^{(x)}_r(\Delta)=\sum_{\mathrm{lev}>\kappa}
\mathrm{FP}[\Delta{+}\mathrm{lev},r]\,(x\rstar)^{\mathrm{lev}}$
(at the production $\kappa_{\rm cert}=28$) are the exact closed-form
level sums at fold base $x\rstar$. One-box certificates in
\code{cgeo\_window\_cert}. (Consumed by
Theorem~\ref{thm:envelope}, tail branch.)
\item[{[C6]} closed-form tail artifacts.]
The production closed-form tails at $\Lambda=9$: the unfurled
\code{That\_L9\_k28.m}, the folded \code{That\_L9\_k28geo.m}
(constructed symbolically by the released script
\code{build\_That\_L9\_k28geo.wl} from the same finite absolute-value recipe,
with the geometric base folded into the level sums), and the exact
increment \code{INC\_22\_28\_L9.m}; all consumed directly by the
certificate scripts.\footnote{Separately --- and this distinction
matters for what is being trusted --- the release includes exact
export-chain checks for the large-spin-tail tensors actually certified. These
checks validate that the emitted Bernstein tensors faithfully
transcribe the shipped polynomial objects consumed by the certifier;
they are complementary to, not substitutes for, the symbolic argument
of Lemma~\ref{lem:majorant}.} (Consumed by
Lemma~\ref{lem:majorant} and every sector certificate.)
\item[{[C7]} the zero-sum reduction checks.]
The derivation of the eight-term recurrence \eqref{eq:seventerm} from the HOR ODE, the
even-level factorization $T_4=16(S-1)\circ T_3$, the annihilation of
constants, and the initial identity $w_0=0$, all as exact rational
identities in $(s,\Delta,\ell)$. Proved by
\code{hor\_telescope\_checker.wl} and re-checked by
\code{check\_telescope\_sympy.py}. (Consumed by the proof of
Lemma~\ref{thm:telescope}.)
\end{description}

\noindent Beyond [C1]--[C7] the proof consumes only the three analytic
trusted-base inputs of Section~\ref{sec:trustedbase}: convergence of
the operator product expansion (T0), radial positivity and convergence
$c_m\ge0$, $r<1$ (T1), and the recursion structure (T2: the pole recursion
converges to the true block and agrees with it through order
$\kappa$ --- this is what ties the certified tables to the physical
blocks in the first place).

\subsection{The margin bound, proved}
\label{sec:marginproof}

We now assemble the pieces. Combining the level decomposition
(Lemma~\ref{lem:leveldec}), the majorant maps
(Lemma~\ref{lem:majorant}), and either coefficient bound --- the sharp
$(\Delta,\ell)$-dependent envelope $0\le c_{\mathrm{lev}}\le\hat c$ with
$\hat c=\tfrac32+\tfrac B2$ from
Lemmas~\ref{lem:monotone}--\ref{lem:envelopeC}, or, on the tails, the
$\Delta$-uniform certified base-$3$ bound
$c_{\mathrm{lev}}\le C_0\,3^{\mathrm{lev}}$ of
Proposition~\ref{prop:geo3} --- yields the envelope
\begin{equation}
  \big|\widehat V^{\rm true}_m-\widehat V^{(\kappa)}_m\big|
  \;\le\;\sum_{\mathrm{lev}>\kappa}c_{\mathrm{lev}}\,
  \hat\omega_{\mathrm{lev},m}
  \;\le\;\hat c\;\widehat T^{(\kappa)}_m ,
  \label{eq:envelope}
\end{equation}
an inequality between \emph{nonnegative} objects at every step. One
notational precision: as displayed, \eqref{eq:envelope} covers the
first coefficient bound, where $\hat c$ is level-independent. When a
tail piece instead uses the geometric coefficient bound, the
level-dependent factor is \emph{absorbed into the tail majorant}, not
into the constant: the certified product takes the form
$C\cdot\widehat T^{(\kappa),\mathrm{geo}}_m$ with
$\widehat T^{(\kappa),\mathrm{geo}}_m	
=\sum_{\mathrm{lev}>\kappa}(5/2)^{\mathrm{lev}}\hat\omega_{\mathrm{lev},m}$,
which retains a Pochhammer--geometric closed form (the weighted
two-level ratio is bounded by $(5/2)^2\hat q<5/8<1$, so the infinite
weighted tail is bounded by a geometric series, equivalently by at most
$8/3$ times its first weighted two-level pair). The appearance of base
$\tfrac52$ rather than $3$ is an engineering artifact: the shipped
tables were assembled against a majorant of the form
$16\cdot(\tfrac52)^m$. The fold comparisons [C5] are
what makes those tables sound: the tables depend on the coefficient
majorant only through the closed-form level sums $S^{(x)}_r$, of which
they are a fixed nonnegative-weight linear map, so the ten
inequalities $16\,S^{(5/2)}_r\ge C_0\,S^{(3)}_r$ imply that the
shipped $16$-base-$(\tfrac52)$ tail polynomials majorize the certified
base-$3$ envelope term by term --- every shipped per-tile certificate
is the binding inequality of the certified chain, and the
number $16$ in the assembly scripts is an assembly parameter (any
constant passing the ten comparisons would serve), not a claim about
block coefficients. Each certified piece consumes exactly one of the
two coefficient bounds --- and hence one closed-form tail --- and its
log records which. Two rules fix the choice. Within each fixed-spin
sector the split at $\tau=\tausplit$ is the boundary between the two
regimes (window: the sharp $c_\infty$ envelope; twist tail: the
geometric bound). Every spin-tail piece --- every piece with
$\ell\ge\ell_*=100$ --- uses the geometric branch for \emph{all}
$\tau$, including $\tau\in[1,\tausplit]$: the window envelope
$\hat c_{\rm prod}=\tfrac74+\tfrac{B^2}4$ grows like $\ell^2$ through
$B^2\propto\Delta(\Delta+\ell+2)$, so it is not finite on the
$\ell=\infty$ face and could not be certified there, whereas the
geometric bound is $(\Delta,\ell)$-uniform. The released scripts
implement this. In \code{certify\_ring\_tile\_L9.wl} each
per-spin sector is split into a window piece (suffix \code{.a}, twist
window, $\hat c_{\rm prod}$) and a tail piece (suffix \code{.b},
$\tau\ge\tausplit$, geometric tail); \code{certify\_tile\_fartail\_L9.wl}
uses the geometric tail on every piece of the large-spin tail. The
three-variable pieces of the strip certifier
\code{certify\_strip\_tile\_L9.wl} (also available in the ring script)
follow the same rules: \code{tail3d.aa} ($\ell\in[14,100]$,
$\tau\in[1,100]$) is the only bounded piece and uses
$\hat c_{\rm prod}$, while \code{tail3d.ab}, \code{tail3d.ba} and
\code{tail3d.bb} (the pieces with $\tau\ge100$ or $\ell\ge100$) carry
the geometric tail. We state the envelope as a theorem, since every certificate consumes
it.

\begin{theorem}[truncation envelope]\label{thm:envelope}
For every certified sector domain and every derivative component $m$,
exactly one of the following two certified inequalities is consumed,
as recorded per piece in the release logs (``windows'': the bounded
pieces with $\ell\le100$ and $1\le\tau\le\tausplit$; ``tails'': the
large-twist tails $\tau\ge\tausplit$ and every large-spin-tail piece
$\ell\ge100$).  The scalar windows of the strip tiles, which lie below
twist one where $\hat c$ is unbounded, are the one exception: there the
cleared inequality of Lemma~\ref{lem:clearedwindow}, with the strip
constant $\What$ of Lemma~\ref{lem:strip} in place of $\hat c$, is
consumed instead.
\begin{align*}
  \text{(windows)}\quad
  &\big|\widehat V^{\rm true}_m-\widehat V^{(\kappa)}_m\big|(\Delta,\ell)
  \ \le\ \hat c(\Delta,\ell)\,\widehat T^{(\kappa)}_m(\Delta,\ell),\\
  \text{(tails)}\quad
  &\big|\widehat V^{\rm true}_m-\widehat V^{(\kappa)}_m\big|(\Delta,\ell)
  \ \le\ C_0\,\widehat T^{(\kappa),3}_m(\Delta,\ell)
  \ \le\ 16\,\widehat T^{(\kappa),5/2}_m(\Delta,\ell),
\end{align*}
pointwise on the domain. Here $\widehat T^{(\kappa)}_m=\sum_{\mathrm{lev}>\kappa}
\hat\omega_{\mathrm{lev},m}$ is produced by the absolute-value procedure
above (visibly majorant-based; no signed difference enters its
construction); $\widehat T^{(\kappa),x}_m$
denotes the fold-base-$x$ variant
$\sum_{\mathrm{lev}>\kappa}x^{\mathrm{lev}}\hat\omega_{\mathrm{lev},m}$
(all closed forms). In the tail chain, the first inequality is the
base-$3$ coefficient bound of
Proposition~\ref{prop:geo3} ([C3] plus [C4]), and the second is
the ten fold comparisons [C5] --- so the $16$-base-$(5/2)$ tail
polynomials that the shipped per-tile certificates were assembled with
dominate the certified base-$3$ envelope, and those certificates are the
binding inequalities.
\end{theorem}
\begin{proof}
Termwise: $|c_{\mathrm{lev}}\,\omega_{\mathrm{lev},m}|\le
\hat c\,\hat\omega_{\mathrm{lev},m}$ by the coefficient bounds ---
$0\le c_{\mathrm{lev}}\le\hat c$ on windows
(Lemmas~\ref{lem:monotone} and~\ref{lem:envelopeC}), the base-$3$
bound of Proposition~\ref{prop:geo3} on the tails --- and the
componentwise majorant property of Lemma~\ref{lem:majorant}; summing
over $\mathrm{lev}>\kappa$ and using convergence of the radial
expansion (input T1) gives the first inequality of each branch. The
closed forms are the shipped artifacts [C6]; and the
tail branch's second inequality is the termwise domination implied
by the ten fold comparisons [C5], as explained above.
\end{proof}

The size of the envelope is governed by the certified two-level ratio:
for each pair of radial levels, consecutive majorant terms obey
\begin{equation}
  \frac{\hat c\,\hat\omega_{\mathrm{lev}+2,m}}
       {\hat c\,\hat\omega_{\mathrm{lev},m}}
  \;=\;\frac{\hat\omega_{\mathrm{lev}+2,m}}{\hat\omega_{\mathrm{lev},m}}
  \;\le\;\hat q_\omega\;=\;\tfrac{92}{35}\,\rstar^{2}
  \;<\;\tfrac1{10}
  \qquad(\Lambda=9,\ \mathrm{lev}\ge\levstar=21),
  \label{eq:geomdom}
\end{equation}
uniformly over the certified domains including the $\ell=\infty$ face
(the constant envelope $\hat c$ cancels in the ratio, which is the
falling-factorial weight bound from the proof of
Lemma~\ref{lem:ratio}; when the true level coefficients $c_{\mathrm{lev}}$
appear instead of $\hat c$, their ratio costs the further factor
$\tfrac{591}{500}$ of that lemma, still clearing $\tfrac1{10}$)
--- so the envelope is, up to the safe rational factor
$1/(1-\hat q)<10/9$, its own leading pair, whose radial suppression is
$\rstar^{\kappa+2}$ (the first omitted even level; measured size
${\sim}10^{-13}$ in stripped units at $\kappa=20$, the remaining
factors being the polynomial level weights). At \emph{small} levels the weight factor wins and the
ratio exceeds $1$ --- which is harmless by construction: contraction
is only ever invoked for the omitted levels $\mathrm{lev}>\kappa$,
the small levels being exact data inside $\widehat V^{(\kappa)}$. The
level restriction does couple $\kappa$ mildly to $\Lambda$: since the
falling-factorial weight ratio behaves like
$e^{2(\Lambda+1)/(\Delta+\mathrm{lev})}$, the contraction threshold
grows roughly linearly in $\Lambda$, so higher-$\Lambda$ programs must truncate
at a higher order. Measured at $\Lambda=9$ the assembled per-pair ratio
is $0.047$--$0.068$ across our sectors, tracking the predicted
weight-corrected rate $((\kappa{+}2)/\kappa)^{\Lambda}\rstar^2$ almost
exactly.

\paragraph{Choosing the truncation order $\kappa$.} The geometric
smallness just established is what fixes $\kappa$, by one comparison:
the envelope must fall below the smallest margin any functional needs
to spend on it. The searched functionals come with a sampled relative
floor --- the uniform relative margin $t^\star$ of
Section~\ref{sec:regionpractice}, at worst $1.9\times10^{-4}$ on the
tile abutting the island --- while the envelope shrinks like
$\rstar^{\kappa+2}$, measured at $\sim10^{-13}$ in stripped units
already at $\kappa=20$. So even $\kappa=20$ would have sufficed, with about eight orders of
magnitude to spare --- the envelope against the margin floor, both in
stripped units, the comparison already quoted in
Section~\ref{sec:regionpractice}; this is a budget estimate, and no
certificate at $\kappa=20$ is shipped. Production ships
$\kappa=22$ reference vectors together with the exact level increments
to $28$ ([C6]): the increments supply the levels up to $28$, beyond which the closed-form
tails [C6] take over, and the two extra orders are free headroom. This comparison is
the one place where two parameters of Table~\ref{tab:params} meet
($\kappa$ against $t^\star$), and it fixes the last open value of that
table.

It remains to convert the envelope into the margin bound of
\eqref{eq:ribbon}, in the form the released scripts check.

\begin{proposition}[the margin bound as consumed]\label{prop:margin}
Fix a sector piece with its closed compactified domain and its
coefficient branch (window or tail), and let $\hat\alpha$ be the
exact rational functional of the tile, and let
$\hat c_{\rm prod}:=\tfrac74+\tfrac{B^2}4\ge\hat c$ be the rational
envelope majorant (AM--GM; see the Comment below the proof). With $\hat t_m:=|\hat\alpha_m|$, define the margin
\begin{equation}
  M\;=\;\sum_m\hat t_m\,\hat c_{\rm prod}\;\widehat T^{(\kappa)}_m
  \label{eq:margindef}
\end{equation}
(with $\hat c_{\rm prod}\,\widehat T^{(\kappa)}_m$ replaced by
$16\,\widehat T^{(\kappa),5/2}_m$ on tail pieces, and, on the scalar
windows of the strip tiles below twist one, by the cleared form of
Lemma~\ref{lem:clearedwindow}, which multiplies through by $(2\tau-1)$
and uses $\What$). If the
per-sector Bernstein certificate proves
\begin{equation}
  G\;:=\;W^{(\kappa)}-M
  \;=\;\sum_m\hat\alpha_m\,\widehat V^{(\kappa)}_m
  \;-\;\hat c_{\rm prod}\sum_m|\hat\alpha_m|\,\widehat T^{(\kappa)}_m
  \;\ge\;0
  \label{eq:gdef}
\end{equation}
on the domain, then the functional is nonnegative on the \emph{true}
crossing vectors there: $\sum_m\hat\alpha_m\widehat V^{\rm true}_m\ge0$.
\end{proposition}
\begin{proof}
By Theorem~\ref{thm:envelope},
$|\widehat V^{\rm true}_m-\widehat V^{(\kappa)}_m|$ is bounded by the
branch's envelope for every component $m$, so
\[
  \sum_m\hat\alpha_m\,\widehat V^{\rm true}_m
  \;\ge\;\sum_m\hat\alpha_m\,\widehat V^{(\kappa)}_m
  \;-\;\sum_m|\hat\alpha_m|\,
        \big|\widehat V^{\rm true}_m-\widehat V^{(\kappa)}_m\big|
  \;\ge\;W^{(\kappa)}-M\;=\;G\;\ge\;0,
\]
which is \eqref{eq:ribbon}. This is the bound the production tiles
consume: $G$ is certified nonnegative \emph{directly}, as a single
Bernstein object per sector --- the component-by-component form in
\eqref{eq:gdef}, with one summand for each derivative component $m$,
is what the released scripts assemble, with no intermediate
shorthand.
\end{proof}

Comment:
Both $\widehat V^{(\kappa)}_m$ and the closed-form tails
$\widehat T^{(\kappa)}_m$ are rational functions of $(L,\tau)$, and
for the certificate the irrational $\hat c=\tfrac32+\tfrac B2$ of
\eqref{eq:csupmain} is replaced by the rational majorant
\[
  \hat c_{\rm prod}:=\tfrac74+\tfrac{B^2}4\;\ge\;\tfrac32+\tfrac B2
  \qquad(\text{AM--GM: }2B\le1+B^2),
\]
with $B^2=3\Delta(\tau+1)(\Delta+\ell+2)/8$ polynomial --- this
$\hat c_{\rm prod}$, not $\hat c$, is what the released scripts
assemble, and every margin statement is made with it; the object handed
to the Bernstein
certifier is the \emph{numerator} of $G$ after all terms are placed over
a common denominator, and the polynomial sign claim is equivalent to
$G\ge0$ only because every factor of that denominator has separately
been certified positive on the closed domain (the
factor-by-factor check of Section~\ref{sec:compact}). Stated as a
polynomial-matrix-program fact: the rigorous statement is always the
one \emph{after} clearing a certified-positive denominator.

The dependency chain of the completed proof, and what proves each
statement:
\begin{center}\small
\begin{tabular}{@{}p{0.30\linewidth}p{0.27\linewidth}p{0.37\linewidth}@{}}
\hline
statement & where & proved by\\
\hline
level decomposition \eqref{eq:leveldec} & Lemma~\ref{lem:leveldec} &
exact identity\\
majorant domination $|\omega|\le\hat\omega$ & Lemma~\ref{lem:majorant}
& stepwise triangle ineq.; artifacts [C6]\\
window bounds $0\le c_{\mathrm{lev}}\le\hat c$ ($\tau\ge1$); cleared strip windows
(Lemma~\ref{lem:clearedwindow}) &
Lemmas~\ref{lem:monotone}, \ref{lem:envelopeC}, \ref{lem:strip} &
Lemma~\ref{thm:telescope} + [C1] ([C2] on the strip)\\
tail ratios $\hat q,\hat q_\omega<\tfrac1{10}$ (Lemma~\ref{lem:ratio}, eq.~\eqref{eq:geomdom}) &
Lemma~\ref{lem:ratio} & proof in text; \emph{not consumed}: the tails are exact closed forms [C6], and the ratio bound serves only to size $\kappa$\\
tail bound $c_m\le C_0 3^m$ & Proposition~\ref{prop:geo3} &
[C3] + [C4]\\
envelope & Theorem~\ref{thm:envelope} & all of the above except the tail-ratio row, + [C5]\\
margined certificate $G\ge0$ & Proposition~\ref{prop:margin} &
per-sector Bernstein objects\\
\hline
\end{tabular}
\end{center}

Numerically $M\sim10^{-12}$ in stripped units (the bare size
$\rstar^{\kappa+2}$ of the first omitted level is smaller, ${\sim}10^{-13}$
already at $\kappa=20$; $M$ carries the majorant constants and the fold
factor $16$), against margins
$t^\star\|\widehat V\|_\infty\gtrsim10^{-4}$: the margined certificates
pass with about eight orders of magnitude to spare, the same comparison
that fixed $\kappa$ above. Every link --- the level decomposition, the
majorant maps, the coefficient bounds, the closed-form tails --- is an
exact identity or inequality checked by machine, and the analytic
content has been compressed into the zero-sum form of
Lemma~\ref{thm:telescope} and the base-$3$ induction of
Proposition~\ref{prop:geo3}.

\section{The trusted computing base, exhaustively}
\label{sec:trustedbase}

A certificate run that prints \code{True} establishes the tile exclusion
conditional on the following trusted
computing base of the proof: everything on it is used without being
checked by the certificates, so it is where a skeptical reader should
direct their attention.  As in any computer-assisted proof, the
correctness of the interpreter, the compiler, the operating system and
the hardware is assumed and not listed.
\begin{description}[nosep]
\item[CFT axioms.] The general theory that turns a CFT into the objects
  the certificates handle: conformal symmetry, unitarity (reflection
  positivity), and the existence of a convergent operator product
  expansion \cite{Ferrara:1973yt,Polyakov:1974gs,Mack:1976pa}.  From these
  follow the ingredients used without further proof: crossing symmetry of
  the Euclidean four-point function; its decomposition into conformal
  blocks with coefficients $p_{\mathcal O}=\lambda_{\sigma\sigma\mathcal
  O}^2\ge0$ \cite{Dolan:2000ut,Dolan:2003hv}; the unitarity bounds on
  $(\Delta,\ell)$ that define the sectors of Section~\ref{sec:setup}
  \cite{Mack:1975je,Minwalla:1997ka}; and the blocks as the solutions of
  the conformal Casimir equation \cite{Dolan:2003hv}.  The argument is the
  one of \cite{Rattazzi:2008pe}, and the quantitative inputs T0--T2 below
  are theorems about these objects.  For the status of the axioms
  themselves --- in particular the passage from Euclidean correlators with
  a convergent OPE to the Wightman axioms --- see
  \cite{Kravchuk:2020scc,Kravchuk:2021kwe}.
\item[The unitarity-bound endpoint.] In a unitary CFT a scalar primary
  $\mathcal O$ whose dimension sits exactly at the unitarity bound
  $\tfrac d2-1$ has a level-two descendant $\partial^2\mathcal O$ of zero
  norm, which therefore vanishes in every correlator
  \cite{Minwalla:1997ka}.  This one fact is what Lemma~\ref{lem:corner}
  uses, together with the conformal form of the three-point function, to
  exclude a scalar at $\De=\tfrac12$ that appears in $\sigma\times\sigma$;
  the other corner scenario, a decoupled scalar at $\tfrac12$, is excluded
  by the certificates.  The stronger classical theorem that such a field
  is a free field \cite{Pohlmeyer1969,Weinberg2012minimal} is used in
  App.~\ref{app:endpoint} only to describe what the decoupled scenario
  would be, and no proof depends on it.
\item[Computer algebra.] The exact computations run through the
  functions of a computer-algebra system named in
  App.~\ref{app:qsqrt2sec}; their correctness is trusted.  The
  independent C++/GMP re-checker of the next item covers the Bernstein
  tensors only.
\item[The certifier programs.] The verdicts are produced by programs
  written for this work, and their correctness as implementations of the
  mathematics below is trusted: the Wolfram-language tile scripts and the
  reference Bernstein certifier; the C++/GMP Bernstein engine, which is
  the only certifier of the $275$ large-spin-tail pieces and of the strip
  tiles' corner pieces, where it runs its sign-guided subdivision policy;
  the Python script that splits the double-tail pieces along the diagonal
  (App.~\ref{app:bernstein}); and the log audit \code{audit\_logs.py},
  which is what establishes that every required certificate of every tile
  was actually run and returned \code{True}, since a tile script's final
  verdict is the conjunction of the certificates it recorded.
\item[Exact algebra.] Integer/rational/$\Q(\sqrt2)$ arithmetic; Sturm's
  theorem; the Bernstein implication \eqref{eq:bernpos}. (The implementation,
  not the mathematics, is cross-validated:
  the release ships an independent C++/GMP re-checker for the emitted
  tensors --- no code shared with the producer, cross-validated against
  the reference implementation by verdict-, tree- and counter-identical
  runs under the common longest-edge split policy (its production
  sign-guided policy yields different, smaller trees; soundness is
  policy-independent), and
  used in production for the large-spin-tail layer; App.~\ref{app:bernstein}.)
\item[T0: OPE convergence.] The block decomposition converges absolutely,
  uniformly on compact subsets of a neighbourhood of the symmetric point
  \cite{Pappadopulo:2012jk}, licensing termwise
  application of the functional to \eqref{eq:crossing} and the rearrangement
  into ``a convergent sum of nonnegative terms equals a negative number''.
\item[T1: radial positivity and convergence.] The stripped block's radial
  expansion has nonnegative coefficients and converges for $r<1$
  \cite{Hogervorst:2013sma}.
\item[T2: recursion structure.] The pole recursion converges to the true block and
  its truncation at level $\kappa$ agrees with the block through radial order
  $\kappa$ \cite{Kos:2013tga,Kos:2014bka,Penedones:2015aga}; the true block
  satisfies the diagonal Casimir ODE \cite{Hogervorst:2013kva}, so its
  coefficients obey the radial Casimir recurrence.
\item[L3: coefficient bounds and geometric domination.] The certified
  envelope \eqref{eq:envelope} rests on: (i) the level decomposition and
  majorant maps of Lemma~\ref{lem:leveldec} (exact finite constructions,
  identity machine-verified at production orders); (ii) the coefficient
  bounds $0\le c_{\mathrm{lev}}\le\hat c$ from the zero-sum form
  \eqref{eq:telescoperec}--\eqref{eq:cinfty}, which is \emph{proven}: it is
  derived from the diagonal fourth-order Casimir ODE of
  \cite{Hogervorst:2013kva} by an exact change of variables, and the two
  rational identities that constitute the derivation (the radial Casimir
  recurrence annihilating the diagonal coefficients, and its even-level
  factorization $T_4=16(S-1)\circ T_3$) are machine-checked \emph{symbolically
  in} $(s,\Delta,\ell)$ and shipped as artifacts, with the transcription of
  the published ODE independently validated against our own blocks in four
  parameter regimes; the identities are proven on the open unitary region
  and extend to the conserved-current boundary points $\Delta=\ell+1$ by
  the continuity prescription of \cite{Hogervorst:2013kva} for the
  exceptional Frobenius cases (the apparent singular factors are removable;
  at the certificate level all such cancellations are performed
  algebraically in reduced normal form and covered by the
  factor check, as described in the proof of Lemma~\ref{thm:telescope});
  (iii) the two-level ratio bound \eqref{eq:geomdom}, proved in the text
  and used only to size $\kappa$;
  (iv) on tail pieces, the base-$3$ geometric coefficient bound
  \eqref{eq:geo3}: eighteen per-level window Bernstein certificates
  (levels $30$--$46$, closed compactified domains), the strong induction (window seven)
  of Lemma~\ref{lem:inductstep} from the certified
  base window, and the ten fold-comparison
  certificates ($16\,S^{(5/2)}_r\ge C_0\,S^{(3)}_r$) that transfer
  the corrected bound to the shipped $16$-base-$(5/2)$ tail assemblies
  (\code{cgeo\_window\_cert} artifacts). What remains trusted here beyond the cited ODE is
  finite and inspectable. \footnote{The GFF double-twist majorant
  (Section~\ref{sec:truncationerrorwhat}) stands as an independently
  proven, if less sharp, substitute for the bounds (ii); it does not
  replace the recurrence in the construction of the level-$24$--$28$
  increments (last item of this list).}
\item[Data generation and transcribed formulas.] The block tables are
  generated from scratch in exact arithmetic: the radial levels up to $22$
  by the pole recursion with symbolic spin, and the levels $24$, $26$ and
  $28$ by the zero-sum recurrence \eqref{eq:telescoperec} started from
  \eqref{eq:c2closed}; both give the coefficients of the true block (T2).
  The $\Q(\sqrt2)$-membership
  of all inputs is enforced (not assumed) by the split-and-reconstruct check
  of Section~\ref{sec:qsqrt2}.  Four formulas enter that generation by
  transcription and are trusted as transcribed: the residue coefficients
  and seed of the pole recursion (App.~\ref{app:blocks}, from
  \cite{Kos:2013tga,Kos:2014bka,Penedones:2015aga}); the diagonal Casimir
  operator of \cite{Hogervorst:2013kva} behind the radial recurrence (item
  L3(ii)); the transverse Casimir recursion, which appears verbatim in the
  three released table-building scripts; and the constants of the $z\to\rho$ derivative
  change of variables.  The tile tables and the coefficient functions of
  [C3] are produced by two implementations of the same saturated
  recursion (Wolfram language and C++); their radial coefficients are
  checked to be identical at every level up to $22$
  (\code{tablecheck\_L9.wl}).  What stays trusted is that the shipped
  table files are the output of the shipped build script, which a rerun of
  that script verifies.  Each is validated against independent
  high-precision floating-point evaluations of the blocks in several
  parameter regimes; that validation is a cross-check of the transcription,
  not an input to any certificate.
\end{description}

\section{Mathematically rigorous bounds on the 3d Ising critical exponents}
\label{sec:rigresults}

\label{sec:ring}
\label{sec:paving}

The forty rectangular tiles $T_i$ listed in Table~\ref{tab:tiles} (App.~\ref{app:tiles}) and illustrated in Figure~\ref{fig:ring} satisfy
\begin{equation}
	R_{\rm ext}\setminus R_{\rm inner}\ \subseteq\ \bigcup_i T_i ,
	\label{eq:coverage}
\end{equation}
with \begin{equation}
	R_{\rm ext}=[\tfrac{505}{1000},\tfrac{5405}{10000}]\times
	(\tfrac{50}{100},\tfrac{300}{100}],\qquad
	R_{\rm inner}=(\tfrac{5055}{10000},\tfrac{540}{1000})\times
	(\tfrac{123}{100},\tfrac{155}{100}),
\end{equation}
The inclusion of these rational rectangles is exact and can be checked in the released script \code{check\_coverage.py}.

For each tile we produced one margin-maximizing candidate functional (see Section~\ref{sec:regionpractice}) which was then certified to be strictly positive on the identity vector and nonnegative on every constraint vector on which positivity is demanded (see Section~\ref{sec:certarch}).\footnote{Non-negativity on the constraint vectors together with strict positivity on the identity is what the exclusion argument needs. The certified inequality is the margined one, truncated value minus truncation envelope, so it holds for the true conformal blocks; it is not strict everywhere, since on some compactified faces at infinity the smallest Bernstein coefficient is zero.} The full certificate data is available at \cite{Reehorst:certrelease} and enables any reader to check the certifications for themselves.\footnote{There are $116$ certificates per frame, band or gap tile (Table~\ref{tab:inventory}) and $109$ per strip tile, whose large-spin tail consists of four pieces instead of eleven. In the final certification run recorded in the release, the seventeen of the first layer took about twenty minutes per tile and the per-spin ones four to five minutes per spin. For tiles 16--40 the large-spin tail $\ell>100$ through $\ell=\infty$ is covered by one interpolated trivariate polynomial per tile; its Bernstein certificates take seconds through the C++ checker, but the exact proof that the stored polynomial \emph{is} the certified object takes $14$--$22$ hours per tile. The strip tiles assemble their large-spin tail with symbolic spin instead, in two exact assemblies of $7$--$19$ hours each. Verifying all forty tiles costs about $10^3$ core-hours, roughly $25$ per tile, in $135$ independent tasks of which the longest runs under a day (App.~\ref{app:timings}).}

For a tile with even-scalar window $\Ieps$, the certificate imposes scalar-sector positivity on $\Ieps\cup[3,\infty)$. Thus it excludes any spectrum whose $\Zt$-even scalar exchanges all lie in this positivity domain. In particular, a spectrum with no relevant even scalar at all is excluded automatically by every tile, since its even scalar exchanges lie entirely in $[3,\infty)$. 

To turn the collection of different tiles into bounds on $\Ds$ and $\De$, one
uses the \emph{at-most-one} relevant even scalar assumption of Section~\ref{sec:setup}. Under that assumption every tile excludes the possibility that the unique operator $\epsilon$ lies within it. Thus, together they certify that $\De$ needs to be outside their union. 

Without uniqueness the logic collapses: a theory with two relevant even scalars at, say, $1.21$ and $1.57$ can evade a $[1.55,1.60]$-window tile through the first operator and a $[1.20,1.23]$-window tile through the second. The uniqueness assumptions are thus never consumed by a solver, but they are essential to combining the exclusion regions of the various functionals.

The bottom edge $\De=\tfrac12$ of $R_{\rm ext}$ is covered by no tile: the crossing vectors have a pole at the free-scalar dimension, so positivity of the cleared polynomials on the closed box does not certify positivity of the true block at that point (Lemma~\ref{lem:clearedwindow}). This single edge is excluded by a separate analytic argument, Lemma~\ref{lem:corner} of App.~\ref{app:endpoint}: a scalar exactly at the unitarity bound cannot appear in the $\sigma\times\sigma$ OPE, and a spectrum whose only relevant even scalar is such a decoupled scalar falls under the gap tiles. Writing $\tilde R_{\rm ext}=[\tfrac{505}{1000},\tfrac{5405}{10000}]\times[\tfrac12,3]$ for the closed rectangle, the certificates together with Lemma~\ref{lem:corner} give the following theorem.

\begin{theorem}[certified exclusion region]
\label{thm:ring}
\label{thm:paving}
Let a unitary 3d CFT satisfy assumptions (a), (c), (d) of
Section~\ref{sec:setup} and have at most one relevant $\Zt$-even
scalar, with $\Ds\in[0.505,0.5405]$. Then:
\begin{enumerate}[nosep,label=(\roman*)]
\item the spectrum \emph{must} contain a relevant $\Zt$-even scalar
$\epsilon$: a spectrum with no relevant even scalar is excluded (by the
gap tiles 37--40, whose functionals are positive on every even
scalar $\ge\tfrac85$);
\item every point of $\tilde R_{\rm ext}\setminus R_{\rm inner}$ is
excluded as the location of $(\Ds,\De)$ --- the tiles cover
$R_{\rm ext}\setminus R_{\rm inner}$, open at $\De=\tfrac12$, and the
edge $\De=\tfrac12$ is excluded by Lemma~\ref{lem:corner} --- so that
\begin{equation}
  (\Ds,\De)\ \in\ R_{\rm inner}=(0.5055,\,0.540)\times(1.23,\,1.55);
\end{equation}
\item if, additionally, the critical 3d Ising model has a scaling
limit described by such a CFT, then
$(\Ds,\De)\in R_{\rm inner}$ and its leading critical exponents obey,
via the dictionary \eqref{eq:exponents},
\begin{equation}
  \eta\ \in\ (0.011,\ 0.080),\qquad
  \nu\ \in\ \big(\tfrac{100}{177},\ \tfrac{20}{29}\big)
        \ \approx\ (0.5650,\ 0.6897).
\end{equation}
\end{enumerate}
\end{theorem}
\noindent Parts (i)--(ii) are the certified mathematics. Part (iii) adds the unproven --- but universally assumed --- hypothesis that the scaling limit of the critical 3d Ising model is a conformal (and not merely scale-invariant) field theory; the exponents $\eta,\nu$ are then a change of coordinates on $(\Ds,\De)$, and part (ii) applies to them.

\begin{figure}[t]\centering
\begin{tikzpicture}[x=0.78cm,y=0.78cm]
  \fill[black!18]
    (0,0) rectangle (14.2,9.0);
  \fill[white]
    (0.2,1.035) rectangle (14.0,2.475);  %
  \draw[black!60,thin] (0,0.9) -- (14.2,0.9)  (0,2.7) -- (14.2,2.7);
  \draw[black!60,thin]
    (2.0,0) -- (2.0,0.9) (6.8,0) -- (6.8,0.9) (10.4,0) -- (10.4,0.9)
    (12.4,0) -- (12.4,0.9)
    (2.0,2.7) -- (2.0,9.0) (6.8,2.7) -- (6.8,9.0) (10.4,2.7) -- (10.4,9.0)
    (0.2,0.9) -- (0.2,1.035)  (14.0,0.9) -- (14.0,1.035)
    (0.2,2.475) -- (0.2,2.7)  (14.0,2.475) -- (14.0,2.7);
  \draw[thick] (0,0) rectangle (14.2,9.0);
  \draw[thick,dashed] (0.2,1.035) rectangle (14.0,2.475);
  \fill[black] (5.26,1.857) circle (0.07);
  \node[above right,inner sep=1pt] at (5.26,1.857) {\small Ising};
  \node[below] at (0,-0.1) {\small $0.505$};
  \node[below] at (5.26,-0.1) {\small $0.5181$};
  \node[below] at (14.2,-0.1) {\small $0.5405$};
  \node[left]  at (-0.1,0) {\small $1.00$};
  \node[left]  at (-0.1,0.9) {\small $1.20$};
  \node[left]  at (-0.1,1.857) {\small $1.41$};
  \node[left]  at (-0.1,2.7) {\small $1.60$};
  \node[left]  at (-0.1,9.0) {\small $3.00$};
  \node[below] at (7.1,-0.55) {$\Ds$};
  \node[left]  at (-1.0,4.5) {$\De$};
\end{tikzpicture}
\caption{The certified exclusion region. Shaded: the union of the
forty excluded tiles of Table~\ref{tab:tiles} covers
$R_{\rm ext}\setminus R_{\rm inner}$ (the band $(\tfrac12,1.00]$ continues
below the plotted range; thin lines: a subset of the tile boundaries; the frame between $\De=1.20$ and $1.60$ is tiled
as listed in Table~\ref{tab:tiles}). Dashed: the inner rectangle of
Theorem~\ref{thm:ring}. The 3d Ising point is marked (we use the value determined by conventional pointwise numerical bootstrap bounds \cite{Kos:2016ysd,Chang:2024whx}). Nothing changes qualitatively if we replace this with best estimates obtained using say lattice or experiment: all of these place the Ising model well within these bounds.}
\label{fig:ring}
\end{figure}
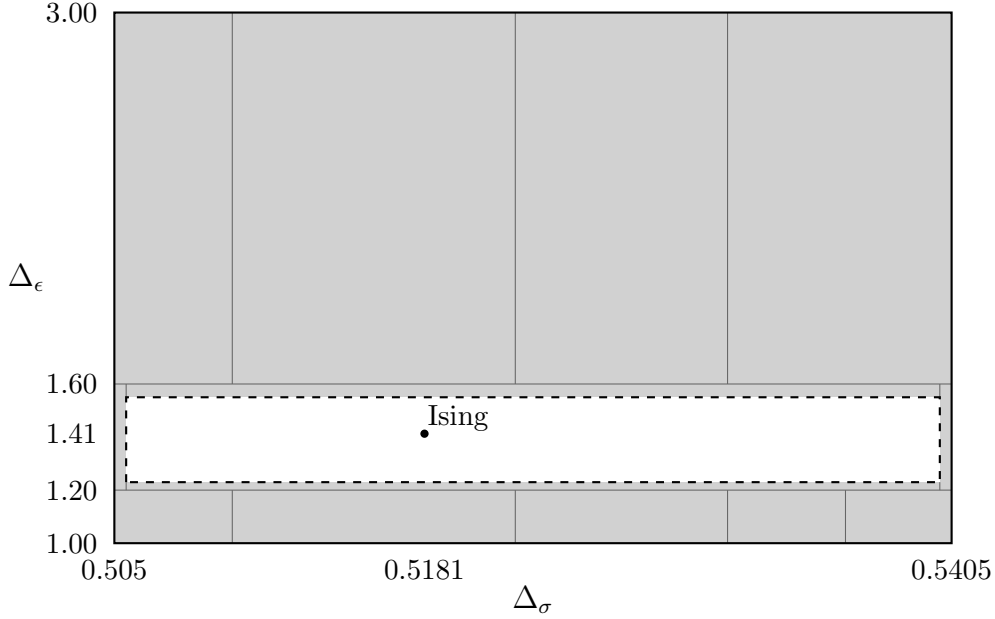

The statement is of course conditional on the spectral assumptions, as
is every island in the literature; the novelty is that \emph{given} the
assumptions the bound is a theorem, not a numerically likely observation. The
strength of these first certified bounds is modest --- they are a proof
of principle. The same architecture scales to higher $\Lambda$, so
stronger bounds are a matter of compute; but the strongest bounds are
not obtained by scaling $\Lambda$ alone. They require mixed-correlator
systems --- the multi-scalar $\sigma$--$\epsilon$ system
\cite{Kos:2014bka,Kos:2016ysd} or the larger systems involving the
spinning stress tensor \cite{Chang:2024whx}. The mixed
extension is also what permits dropping the spin-2 gap assumption (d); its analytic inputs are proved in Section~\ref{sec:mixedoutlook}, where we also list which certificates remain to be produced.

\section{Practical payoff: faster extremal bounds from region functionals}
\label{sec:experiments}

Functionals that are positive on a
\emph{region} of external-dimension space not only help achieve rigor but also change the cost of extracting a bound. As an example this
section shows that region bounds can more quickly and more precisely find the extremum on the island than point bounds (while also providing additional rigor).

We look for the strongest bound on $\Ds$ that can be achieved under the island assumptions of Section~\ref{sec:setup} at $\Lambda=9$.  Both methods below use the identical oracle --- the sampled LP with the audit-and-cut loop of Section~\ref{sec:regionpractice}, the exact
dual simplex, the same identity and stress-tensor audits, the same sampling
densities --- on the same machine.  Only the search strategy differs.  In line with the practical nature of this question for this comparison we only compare the search stage, not the certification. A box or point is called excluded when its functional passes the dense audit (density $1024$ in $\tau$, Section~\ref{sec:regionpractice}); the Bernstein certification of Section~\ref{sec:certarch} is a separate,
later cost and is not part of the timings.

\paragraph{The question.}  With $\De$ restricted to the window $[1,2]$
we are looking for the maximum
\[
  \Ds^{\max}\;=\;\sup\{\Ds:\ (\Ds,\De)\ \text{is not excluded for some}\ \De\in[1,2]\},
\]
i.e.\ the right edge of the island (and for free we localize the $\De$ at which it is
attained).  The answer we obtain is
\[
  0.533\,2797 \;\le\; \Ds^{\max}\;\le\;0.533\,3255 ,
  \qquad \text{with the tip at } \De\in[1.4939,\,1.4949],
\]
where the upper bound is a search-level exclusion of everything to its right
and the lower bound is a value of $\Ds$ at which the witness LP (below) found an allowed spectrum.

\paragraph{Two verdicts: excluded by a functional, allowed by a spectrum.}
Everything in the previous sections concerns one side of the problem, the
\emph{dual} side: a functional $\alpha$ that is positive on every crossing
vector permitted by the assumptions proves that a point or a box is
\emph{excluded}.  To search for an extremal bound it is useful to also have
a test for the opposite verdict: whether a solution of the crossing
equations can be found using a finite sample of operators.  This is also a
linear program, the \emph{primal} side of the one solved in the functional
search.

Fix $\Ds$ and look for a finite spectrum: operators $(\Delta_i,\ell_i)$
permitted by the assumptions, with squared OPE coefficients $p_i\ge0$, such
that
\[
  \vec{\mathrm{Id}}(\Ds)+\sum_i p_i\,\Vvec(\Ds;\Delta_i,\ell_i)=0
\]
in all fifteen components.  We call such a finite spectrum a
\emph{witness}.  If a witness exists, no functional of this derivative
order can exclude the point: applying any admissible $\alpha$ to the
equation above would give a strictly positive number equal to zero.  This
is Farkas' lemma \cite{Farkas:1902,Schrijver:1986}, and it makes the
verdict definitive for the truncated crossing equations as represented in
the LP.  It does not say that
a CFT exists at that point.  It says that the $\Lambda=9$ bootstrap under
these assumptions cannot exclude it, which is exactly the statement that
bounds the island from the inside.

The candidates offered to the witness LP are the operators on a dense grid
in every sector.  We use a wider grid than the initial grid in the dual functional search of
Section~\ref{sec:regionpractice} (the initial dual grid is a subset of the primal grid).  For these fixed dimensions we ask whether nonnegative coefficients $p_i$ exist.  A positive answer makes the dual search for that point or box unnecessary.  This is a linear program with fifteen equality constraints and several thousand unknowns.  That the
unknowns outnumber the equations does not make it easy: the question is
whether $-\vec{\mathrm{Id}}$ lies in the convex cone spanned by the
candidate vectors, and whenever an excluding functional exists all of them
lie on one side of a hyperplane that $-\vec{\mathrm{Id}}$ is not on.  When
a witness exists, one exists with at most fifteen nonzero coefficients, as
many as there are equations, and the simplex returns such a solution. This is Carath\'eodory's theorem for cones \cite{Caratheodory:1911,Schrijver:1986}, and it uses only linearity and the sign condition: if more than fifteen candidates carry weight, their
vectors are linearly dependent, and moving the coefficients along that
dependency until the first one reaches zero gives a solution with one
operator fewer.

Consider the functional search for a problem that is in fact allowed. If the sample set contains the operators of an existing witness, the LP reports that no functional with positive margin exists on the samples; this settles the problem as allowed.  Generically the initial sample set does not contain such operators, because the dimensions in a solution are not all values the grid samples over.  The LP then returns a functional that is positive on the samples; but since we assumed the problem is allowed it must be negative somewhere in between.  The audit finds those places, they are added to the sample set, and the achievable margin shrinks round by round as the samples close in on the operators of a solution (Section~\ref{sec:auditloop}).  The loop ends when the margin has
dropped below a stopping threshold, which we set to $10^{-8}$ in the
relative normalization of Section~\ref{sec:regionpractice}. When this value is reached we conclude the region cannot be excluded.\footnote{This threshold
is a choice of the search, not the requirement of rigor. What certification on the true blocks requires is a margin above the truncation envelope of Section~\ref{sec:envelopestatement}, about $10^{-12}$ in the
same normalization, so the threshold is conservative by four orders of
magnitude, and a functional with a margin between the two values is
discarded although it could be certified.}

Conversely, at a point that can be excluded the witness LP terminates quickly with the answer that no witness exists on the grid.  That answer is limited by the grid in the same way that a functional on a sample set is limited by its samples. It does not prove that no witness lives outside the sampled points.  We use each side for the verdict it can prove.

The witness LP is solved with the same exact simplex as the dual.  It takes
$30$--$40$\,s from scratch and a few seconds when the fifteen operators of
a nearby witness are offered first.  In the region run below it settled
$7$ of the $14$ boxes that could not be excluded in $30$\,s each, against
about $180$\,s for a functional search that ends without a functional.
Witnesses serve the following purpose: they are the allowed end of
every bracket of the point method; they stop the region march with a proof
that the wall cannot advance past that $\Ds$; and they give the lower end
of the final interval that contains $\Ds^{\max}$.

\paragraph{Protocol A: point probes.}  This is the standard approach, made
as strong as the oracle allows, and it uses a point oracle throughout: at
a point $(\Ds,\De)$ the scalar $\epsilon$ sits at exactly $\De$, both for
the functional and for the witness, and nothing is given a window.  At a
fixed $\De$ the edge in $\Ds$ is bracketed between an \emph{allowed} point
and an \emph{excluded} point.  A point is allowed when the witness LP finds
a spectrum with a scalar at $\De$, or when the functional search at the
point ends below the stopping threshold; it is excluded when a functional
at the point passes the audit.  The excluded end starts at $\Ds=0.5405$.
The allowed end is found by stepping down from the allowed value at the
neighbouring value of $\De$, with steps doubling from $2^{-11}$ to $2^{-7}$, until a
witness is found.  The bracket is then bisected --- the driver also allows
a secant step on the margin of the excluded end, but every excluded probe
of this run reported the LP's margin cap, so the secant never fired ---
until it is $2^{-12}$ wide.  The values of $\De$ are searched by golden
section, starting from a coarse scan in steps of $\tfrac18$, to a
tolerance of $\tfrac1{128}$; this assumes that the edge $\Ds(\De)$ has a
single maximum.  The output is an \emph{estimate} of $\Ds^{\max}$: the
midpoint of the last bracket at the best $\De$, here between the allowed
value $0.5330811$ and the excluded value $0.5333252$ at $\De=1.4912$.  The
midpoint itself is never tested.  Nothing but the individual isolated
points is excluded, and the estimate is only as good as the unimodality
assumption and the values of $\De$ that happened to be visited.

\paragraph{Protocol B: a marching wall.}  The region method excludes boxes
$[X-d,X]\times[1,2]$ with a single functional and marches the wall $X$
downward from $X_0=0.5405$, a value known to be excluded.  Before each
solve the witness LP is run at the far edge and at the midpoint of the
box, with scalar candidates in the whole window $\De\in[1,2]$; a witness
proves that the wall cannot pass that value and records it as a lower
limit for the wall.  Otherwise the box is solved.  A success moves the wall
to $X-d$; a failure halves $d$, and so does a success whose margin
falls below $10^{-6}$.  The step starts at $2^{-8}$ and is capped below at the minimal step
$d_{\min}=2^{-16}$.  After a success the next step is set by the
\emph{margin log-model}: the excluding functional's margin decays roughly exponentially
as the wall approaches the edge, so from the last two margins below the
LP's cap a secant on $\log t^\star$ predicts where the margin will reach
$10^{-7}$, ten times the stopping threshold, and the next box ends there.  Until two
such margins are available, i.e.\ while the margins sit at the LP's cap,
the step doubles after each success.  The march ends when a box of width
$d_{\min}$ at the wall cannot be excluded by one functional.

Starting from the final state, a continuation then tests whether
reducing the $\De$ window allows a tighter bound.  Its rule replaces the
end of the march above: whenever a box of width $d_{\min}$ at the wall of
a part of the $\De$ window cannot be excluded by one functional, that part is split at the
value of $\De$ where the last functional's margin is weakest, and the march
continues on each half in the same manner as above.  A part narrower than
$10^{-3}$ in $\De$ is not split further; the continuation ends when every
part whose wall is at the bound is such a part.  This localizes the tip of
the island in $\De$, the white notch at the wall in Figure~\ref{fig:sigmax};
as we discuss below, it does not move the $\Ds$ bound.  Table~\ref{tab:sigmax}
collects the numbers of both protocols.

\begin{table}[t]
\centering
\small
\setlength{\tabcolsep}{4pt}
\begin{tabular}{lrrlll}
\hline
 & time & calls & $\Ds^{\max}$ & witness & tip \\
\hline
A: point probes & 1\,h\,43\,min & 148 & $0.5332031$ (est.) & $0.5330811$ & $\De=1.4912$ \\
B: region wall & 1\,h\,29\,min & 31 & $\le0.5333255$ & $0.5332607$ & \\
\quad + splitting & $+52$\,min & $+28$ & unchanged &  & $\De\in[1.4939,1.4949]$ \\
\hline
\end{tabular}
\caption{The strongest $\Ds$ bound at $\Lambda=9$, $\De\in[1,2]$, by point probes and by the marching region wall.  Calls count functional searches and witness LPs; the witness column is the largest $\Ds$ at which an allowed spectrum was found.  The region run ends with the wall $6.5\times10^{-5}$ above its largest witness; the point run ends with a bracket of $2.4\times10^{-4}$, a gap between witness and wall that the region wall had reached after $76$\,min.}
\label{tab:sigmax}
\end{table}

\begin{figure}[tbp]
\centering
\includegraphics[width=\linewidth]{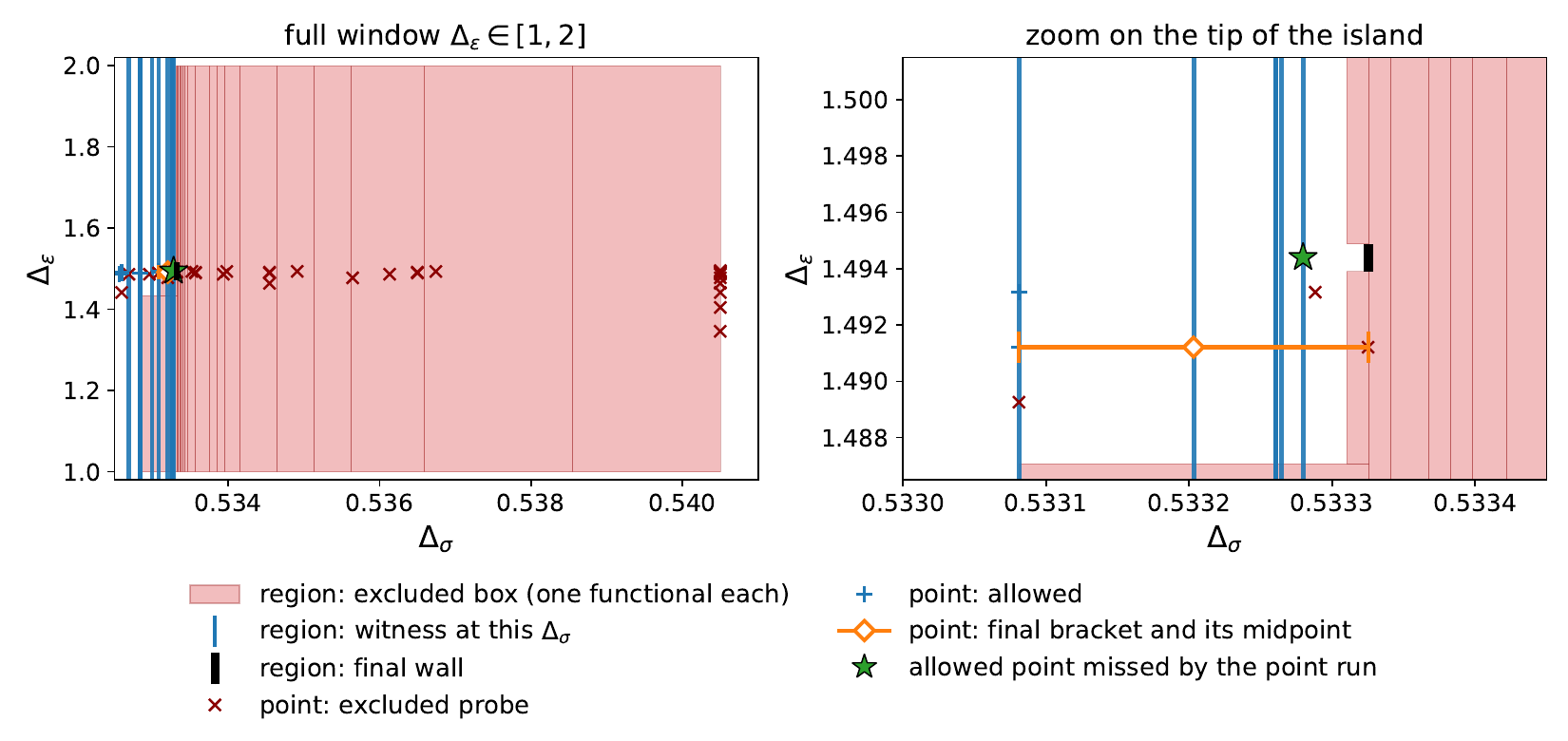}
\caption{End state of the region run (Protocol B) after the continuation, with the point run (Protocol A) overlaid; the right panel zooms on the tip of the island.  Red boxes are excluded by one functional each.  The thick black line is the final wall, at the bound $\Ds=0.5333255$, on the only part of the $\De$ window where it cannot advance: $\De\in[1.4939,1.4949]$, see the white notch at the wall in the right panel. Blue lines mark values of $\Ds$ at which the witness LP found an allowed spectrum. Crosses and plus signs are the excluded and allowed points of the point run.  The orange bar is the final bracket of the point run at its best $\De$. The diamond is the midpoint of that bracket. The star is a manually found allowed point, $(\Ds,\De)=(0.5332797,\,1.4944)$; it lies in the region the point run never probed.}
\label{fig:sigmax}
\end{figure}

\paragraph{What the experiment shows.}
\begin{enumerate}
\item \emph{An excluded region for less than the price of an estimate.}
The region wall finishes in $1$\,h\,$29$\,min, against $1$\,h\,$43$\,min
for the point probes (Table~\ref{tab:sigmax}), and it resolves $\Ds^{\max}$ almost four times as
finely: the point run stops at a bracket of width $2.4\times10^{-4}$,
while the wall ends $6.5\times10^{-5}$ above its witness.  The wall had
closed the gap between its witness and its wall to the point run's width after $76$\,min, three
quarters of the point run's time.

Protocol B delivers an excluded region, $\Ds$ above the bound for all
$\De\in[1,2]$, together with an allowed spectrum $6.5\times10^{-5}$ below
it.  Protocol A delivers a number that excludes nothing beyond $43$
isolated points.

\item \emph{The point search stops short of the tip.}  The tip of the
island lies at $\De\approx1.494$, and the extent of the allowed region
near it is small: the wall's continuation boxes it into a $\De$ interval
of width $10^{-3}$, and the point run found no allowed $\Ds$ at all at
$\De=1.4951$, just above the tip.  The golden section resolves $\De$
only to $\tfrac1{128}$, coarser than that feature, so it settled at
$\De=1.4912$, where the edge is lower.  Its estimate there,
$0.5332031$, is the midpoint of a bracket it did not resolve further, and
its best allowed value, $0.5330811$, lies $2.4\times10^{-4}$ below the
wall's bound.  A witness LP placed directly at the tip, with the scalar
pinned at $\De=1.4944$ exactly as in Protocol A, finds the allowed point
$(\Ds,\De)=(0.5332797,\,1.4944)$ (the star in Figure~\ref{fig:sigmax}),
$2.0\times10^{-4}$ beyond the point run's best allowed value and
$7.7\times10^{-5}$ beyond its estimate.  The structural point: a point
search must fix a value of $\De$ before it can resolve $\Ds$, and it
chooses that value at a resolution set by its tolerance in $\De$, while the wall covers all of
$\De\in[1,2]$ at once and resolves $\Ds^{\max}$ without choosing.

\item \emph{Width in $\Ds$ is what matters, not the extent in $\De$.}  At
the start of the march a single $\Lambda=9$ functional excludes a box of
width $2\times10^{-3}$ in $\Ds$ but not $4\times10^{-3}$, while the full
window $\De\in[1,2]$ costs it almost nothing; near the tip the excludable
width has shrunk to $1.5\times10^{-5}$.  The last two excluded boxes have
widths $2.7\times10^{-5}$ and $1.5\times10^{-5}$ and margins
$5\times10^{-8}$ and $3\times10^{-8}$; the next box of width $d_{\min}$
cannot be excluded by one functional.  The continuation ($52$\,min,
$28$ LP calls, all $21$ functionals of the run pass the dense audit) then
tested splitting at this wall.  The same box of width $d_{\min}$ cannot be
excluded in any of five nested windows, from $\De\in[1,2]$ down to
$\De\in[1.4939,1.4949]$: shrinking the $\De$ window by a factor of $1000$
changes nothing at the tip.  Away from the tip the opposite happens.  The
parts $\De<1.4331$ and $\De\in[1.4331,1.4871]$ are excluded down to
$\Ds=0.5328372$ and $0.5330814$ by one box each, and the two parts
adjacent to the tip advance by one further step of $d_{\min}$.  Splitting
therefore locates the tip, here to $\De\in[1.4939,1.4949]$, but the bound
is set by how narrow a box in $\Ds$ can still be excluded with the
required margin.

\item \emph{Timings.}  In the region run, witness LPs took
$41\%$ of the wall time, the LP solves of the functional search $28\%$,
dense scans $8\%$, and the exact identity audits under $0.1\%$; the
remainder is the construction of the LP rows and the final dense audit.
The $17$ boxes that were excluded took $227$\,s each, the $7$ failed
functional searches $179$\,s each, and the $7$ boxes stopped by a witness
$30$\,s each.  The functional search for a box
proceeds in rounds (Section~\ref{sec:auditloop}): solve the LP on the
current sample set, scan the candidate functional on a much denser set,
and add the locations where it is negative to the sample set.  A box
started from the standard sample set needs of the order of ten such rounds (fifteen for the one region box that started this way, four to fifteen for the point run's probes at $\Ds=0.5405$).  In this process the
grid points added for one box are almost always the ones needed for the
adjacent box, so each box starts from the standard sample set plus the
locations added for its predecessor. Then on average two rounds of further refinement are required ($2.1$ over the sixteen warm-started boxes that were excluded).  Most of the witness time is spent on boxes that turn out to be	
excludable; skipping the witness where the margin log-model predicts a comfortable margin is the obvious next saving.
\end{enumerate}

\paragraph{A third direction.}  From the above experiment we see the first hints that the wall extremalization method scales much better with the dimension of the search space: a wall marches once whatever the number of such directions. The strongest bound is also not significantly dependent on the size of the interval on which positivity is demanded in the directions orthogonal to the search (item~3). The advantage the wall method has over point searches will only become more apparent as we consider higher-dimensional searches. To see this scaling in action we repeated the experiment with a third direction added.

We now look for the extremum of an island also in the dimension of the second $\Zt$-even scalar $\epsilon'$ --- now demanding positivity on the continuum only from $\Delta\ge6$ \cite{Reehorst:2021hmp}.  The functional is then required to be nonnegative on the scalar sector only on $[3.5,4.5]\cup[6,\infty)$ instead of
$[3,\infty)$, a weaker demand, so every functional of the two-dimensional search remains admissible and the extremum over $(\Ds,\De,\Dep)$ lies at or below the two-dimensional one.\footnote{The direction is nonetheless live: e.g. with $\epsilon'$ pinned at $3.5$, $3.83$, $4.5$ or $5.5$ the topmost allowed point of the two-dimensional tip becomes
excluded. The extremum over the window is set by the value $\epsilon'$ where the bound is weakest.} The region search now runs on boxes
$[X-d,X]\times[1,2]\times[3.5,4.5]$; Figure~\ref{fig:sigmax3d} shows its
end state in two slices through the tip.  The point search treats the new direction the way Protocol A treats
$\De$.  It fixes a value of $\Dep$, runs the whole of Protocol A at that
value, which returns the largest $\Ds$ it can find over $\De$, and then
chooses the next value of $\Dep$ by the same golden-section rule that
Protocol A uses to choose $\De$, stopping when successive values of
$\Dep$ are closer than $\tfrac1{16}$.  Each probe remains a point: both
$\De$ and $\Dep$ are pinned exactly, for the functional and for the
witness alike.

\begin{figure}[tbp]
\centering
\includegraphics[width=\linewidth]{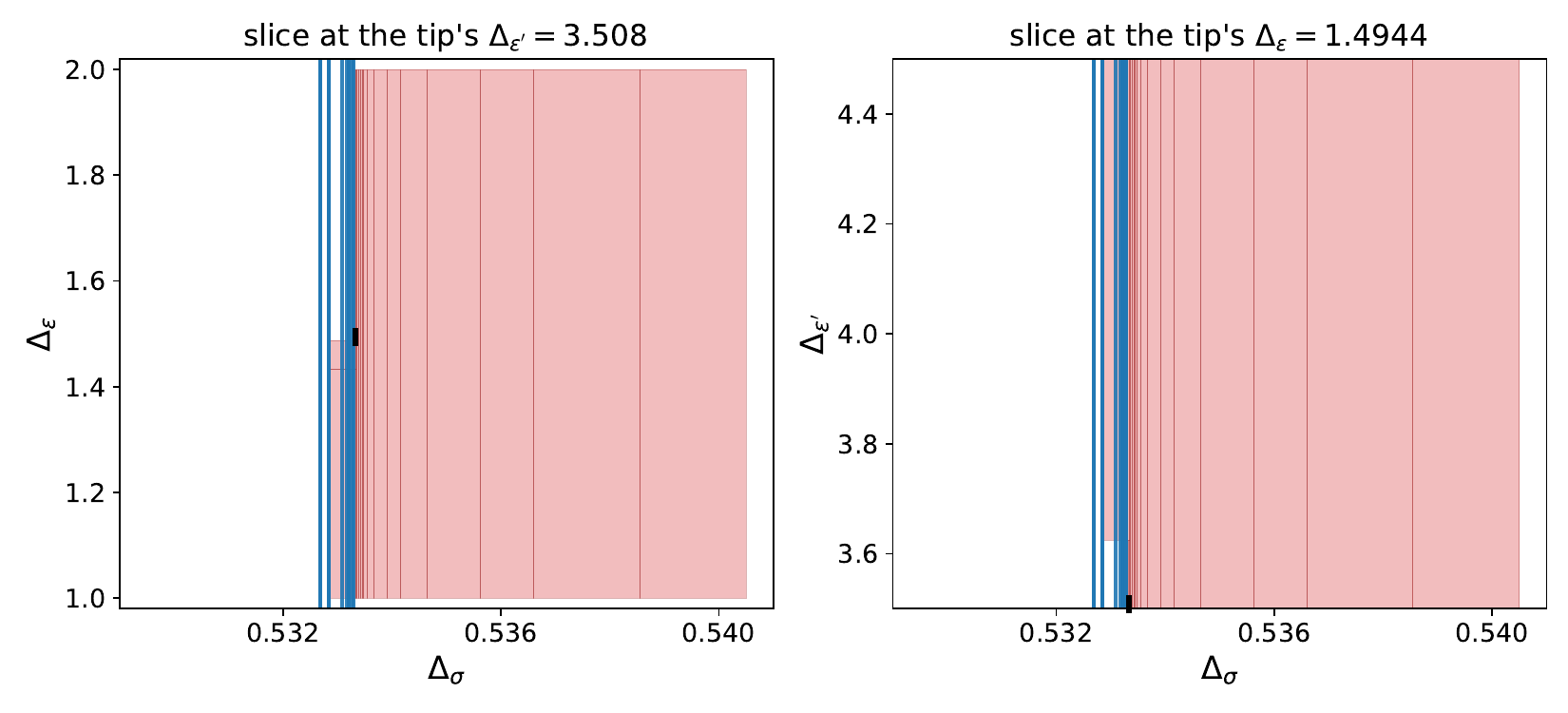}\\[1mm]
\includegraphics[width=\linewidth]{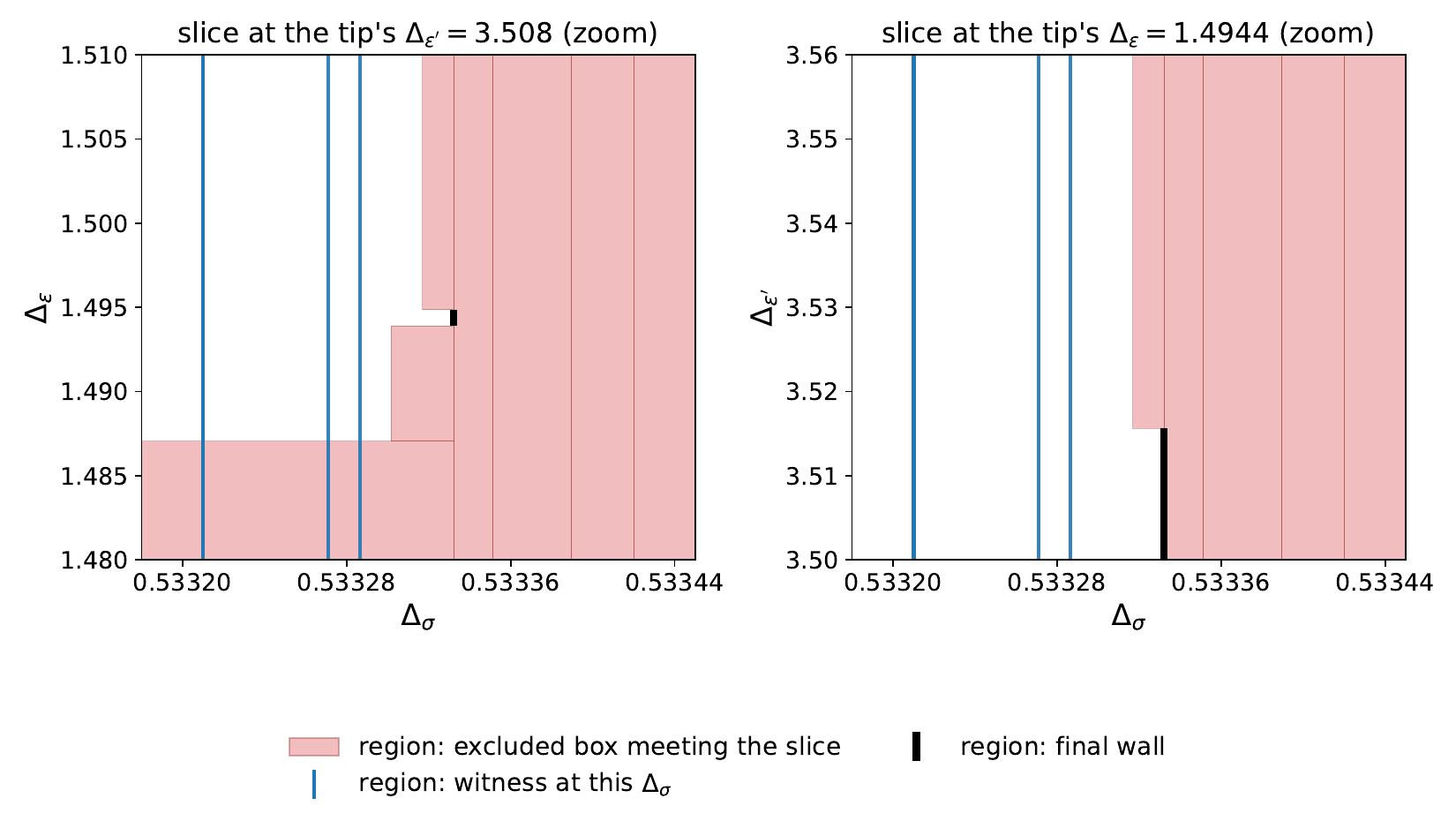}
\caption{The three-direction search, sliced through the tip: the left
panels show the $(\Ds,\De)$ plane at the tip's $\Dep=3.508$, the right
panels the $(\Ds,\Dep)$ plane at the tip's $\De=1.4944$; the lower row
zooms on the tip.  The part of the $\De$ window above the tip is excluded one step of
$d_{\min}$ further than the tip itself and the part below it two steps
further; for $\Dep$ between $3.516$ and $3.625$ the wall stands one step
further left.  Red boxes are
excluded by one functional each.  The thick black line is
the final wall, on the part of the window where it cannot advance.  Blue lines are witnesses.} %
\label{fig:sigmax3d}
\end{figure}

\begin{table}[t]
\centering
\small
\setlength{\tabcolsep}{3pt}
\begin{tabular}{lrrlll}
\hline
 & time & calls & $\Ds^{\max}$ & witness & tip $(\De,\Dep)$ \\
\hline
A: point probes & 8\,h\,34\,min & 1440 & $0.5255127$ (est.) & $0.5253906$ & $(1.443,\,3.520)$ \\
B: region wall & 1\,h\,24\,min & 27 & $\le0.5333322$ & $0.5331758$ & \\
\quad + splitting & $+87$\,min & $+38$ & unchanged &  & $[1.4939,1.4949]\times[3.5,3.516]$ \\
\hline
\end{tabular}
\caption{The three-direction test: the strongest $\Ds$ bound at $\Lambda=9$ over $\De\in[1,2]$, $\Dep\in[3.5,4.5]$, by the same two protocols as in Table~\ref{tab:sigmax}.  Columns as in Table~\ref{tab:sigmax}.  As in Table~\ref{tab:sigmax}, the region-wall line ends when a box of width $d_{\min}$ on the full window cannot be excluded, and everything after the first split is counted in the line below it.  The wall then stands $1.6\times10^{-4}$ above its largest witness; the splitting stage raises the largest witness to $0.5332864$, $4.6\times10^{-5}$ below the wall.}
\label{tab:sigmax3d}
\end{table}

The wall's cost barely changed (Table~\ref{tab:sigmax3d}).  It reached
its bound, $0.5333322$, in $1$\,h\,$24$\,min and $27$ calls, against
$1$\,h\,$29$\,min and $31$ calls in two dimensions; the time to close the gap between witness and wall to the point run's width was
$66$\,min against $76$; an excluded box cost $234$\,s against $227$. The bound agrees with the two-dimensional one to within the resolution: both walls stop when a box
of width $d_{\min}=1.5\times10^{-5}$ cannot be excluded.\footnote{The two runs did
not use the same box grid, and the difference of $7\times10^{-6}$ is half
such a box. The difference is thus a termination issue and unphysical.}

Splitting was tested as before,
now in both directions: the driver chose the direction with the weaker
margin profile, split twice in $\Dep$ and four times in $\De$, located the
tip at $\De\in[1.4939,1.4949]$, $\Dep\in[3.5,3.516]$ in $87$\,min (the
lower panels of Figure~\ref{fig:sigmax3d}), and did
not move the bound.  A part narrower than $\tfrac1{32}$ in $\Dep$ is not split further, the
counterpart of the $10^{-3}$ rule in $\De$; the tip's $\Dep$ interval, of
width $\tfrac1{64}$, is such a part.  The tip is not enclosed from below in $\Dep$: it
sits against the lower edge of the assumed window, $\Dep=3.5$, where the
demand on the functional is closest to the two-dimensional one, and
nothing below that edge is part of the question.
All $20$ functionals pass the dense audit.

The point search fares differently.  With $\Dep$ pinned, the allowed
set near the Ising point is a thin sheet: at fixed $\Dep$ it is about
$2\times10^{-3}$ wide in $\De$, and its position in $\De$ moves with
$\Dep$.  Protocol A begins with a coarse scan: it tests $\De=1,\,1.125,\,1.25,\dots,2$
in turn, and at each of these values it looks for an allowed $\Ds$ before
the golden section starts.  In two dimensions this scan finds allowed
points; here it finds none, so the scan was refined to
$\De=1,\,1.03125,\dots,2$.  Even so, a scan of $132$ witness LPs
($30$\,min) at each of $\Dep=3.883$, $4.117$, $3.734$ and $3.645$ found
no allowed point at any of these $\De$ values, so the golden section had
nothing to compare for its first two hours.  The four scans below
$\Dep=3.6$ each found allowed points at exactly one of the scanned values of $\De$,
always $1.438$ (the bracketing then refines $\De$ within $\pm\tfrac1{32}$ of
it, which is how the table's $\De=1.443$ arises), and gave estimates rising from
$0.5240$ at $\Dep=3.590$ to $0.5255$ at $\Dep=3.520$, where the search
stopped at its tolerance in $\Dep$.  The run took $8$\,h\,$34$\,min and
$1440$ oracle calls, six times the wall's cost, and it determines the bound significantly less accurately.

\paragraph{Code and data.}  The search-stage code behind this section is
released separately from the certificate data \cite{Reehorst:searchrelease}:
the oracle (exact rational simplex, audit-and-cut loop, witness LP, dense
audit), the drivers, the two gates that every production run must pass
(a mock-oracle convergence test and a set of real-oracle regression boxes),
the block data they read, and the complete logs, call tables and stored
functionals of the runs in Tables~\ref{tab:sigmax} and~\ref{tab:sigmax3d},
together with the scripts that produce Figures~\ref{fig:sigmax}
and~\ref{fig:sigmax3d}.  Its README lists the commands behind
each number.  None of it is part of the trusted computing base of
Section~\ref{sec:trustedbase}: the functionals it produces pass the dense audit, and
only the certificate layer turns them into theorems.

\FloatBarrier

\section{The mixed-correlator extension: analytic inputs proved}
\label{sec:mixedoutlook}

The strongest Ising bounds, and the only ones that isolate an island without the need for a spin-2 assumption (d), come from the mixed $\sigma$--$\epsilon$ system \cite{Kos:2014bka,Kos:2016ysd,Simmons-Duffin:2016wlq}\footnote{Or its extension, the $\sigma$--$\epsilon$--$T$ system \cite{Chang:2024whx}.}. This
section establishes the analytic part of that extension; it is a
program, not a second theorem:
\begin{flagshipresult}
The analytic ingredients that certified mixed-correlator tiles consume
are proved here and machine-checked.  These are: a
coefficient majorant for the unequal-external radial Casimir recurrence,
uniform over the whole unitary sector and over every external pair in
$[\tfrac12,3]^2$
(Lemmas~\ref{prop:mixedmajorant}--\ref{lem:mixedcbound}); the exact
polynomial dependence of the truncated mixed tables on the external
dimensions, with an explicit degree, machine-checked by an exact
reconstruction (Lemma~\ref{lem:extpoly}); and the base-window certificates that fix the
majorant constant uniformly in $\delta$ --- the mixed analogue of the
equal-external window certificates [C3].  This has been carried out for
all three channel types over the island-relevant band
($-\tfrac{11}{10}\le\delta\le-\tfrac{2}{5}$ - with the
$\sigma\sigma$--type channels up to $\delta=0$) - giving a
single uniform constant $C^{\rm mix}_0=\tfrac14$ on the whole unitary
sector $\{\ell\ge0,\ \tau\ge\tfrac12\}$, in the cleared normalization near
the two poles; the $280$ Bernstein
certificates ($224$ on $\tau\ge1$ and $56$ on the scalar interval
$\tau\in[\tfrac12,1]$) all pass (released as \code{certificate\_release\_mixed\_extension} \cite{Reehorst:mixedrelease}).  What a certified mixed theorem still requires is the tile certificates --- the functional search.
\end{flagshipresult}
\emph{Nothing in this section enters the proof of the
single-correlator theorem}; its lemmas are consumed only by future
mixed certificates. Proofs that are pure bookkeeping (the degree count
behind Lemma~\ref{lem:extpoly} and the contraction estimates for the
mixed tails) are collected in App.~\ref{app:mixedproofs}.

The starting point is the same object as for the single correlator.
Keeping the external parameters $a=-\tfrac{\Delta_{12}}2$,
$b=\tfrac{\Delta_{34}}2$ \emph{symbolic} through the exact transform
that produced the radial Casimir recurrence \eqref{eq:seventerm} yields
its mixed generalization. In particular,
\begin{equation}
  \sum_{k=0}^{7}\tilde A_k(m;\Delta,\ell,a,b)\,c_{m+k}=0.
\end{equation}
This has
polynomial coefficients --- now for both radial parities, since at
unequal externals the odd levels are populated. The derivation is
machine-checked (checker \code{mixed\_bounds\_checker.wl}, stage
[T$'$]), and the transcription is cross-validated against an independent
object: stage [XCHK] verifies that this recurrence annihilates the
coefficients $c_m$ built from the pole recursion, both at equal externals
and at $\delta=\pm\tfrac{9}{10}$ (levels $m=3,\dots,25$), so the
$a,b$-symbolic transform is not merely trusted at $\delta\ne0$. The mixed tail majorants require a bound $|c_m|\le C\lambda^m$ with $\lambda<1/\rstar=3+2\sqrt2$, since the
level sums carry the factor $\rstar^{\,m}$. We
\emph{certify} this directly on this recurrence, in two steps --- a
machine-certified induction step followed by a strong induction.

\begin{lemma}[certified mixed induction step]
	\label{prop:mixedmajorant}
	For every channel type of the mixed system --- the three sign
	patterns $(\Delta_{12},\Delta_{34})\in
	\{(\delta,\delta),\,(-\delta,\delta),\,(\delta,-\delta)\}$ arising
	from the orderings of the mixed four-point functions, with
	$\delta=\Ds-\De$, certified
	separately; we call $(\delta,\delta)$ and $(\delta,-\delta)$ the
	$\sigma\sigma$--type channels and $(-\delta,\delta)$ the
	$\sigma\epsilon$ channel (\code{sese}, \code{conj} and \code{esse} in
	the release) --- with external difference
	$\delta\in[-\tfrac52,\tfrac52]$ (covering \emph{any} external pair in
	$[\tfrac12,3]^2$),
	\emph{uniformly over the entire unitary sector}
	$\{\ell\ge0,\ \tau\ge\tfrac12\}$ --- including the scalar interval
	$\Delta\in[\tfrac12,1)$, where the $\mathbb Z_2$-odd channel exchanges
	$\sigma$ itself --- and with the band-dependent start
	$m_0=40$ on $|\delta|\le\tfrac34$, $80$ on
	$\tfrac34\le|\delta|\le\tfrac32$, $160$ on
	$\tfrac32\le|\delta|\le\tfrac52$: the strong-induction step of
	window $7$,
	\[
	\sum_{k<7}\big|\tilde A_k(m;\Delta,\ell,\delta)\big|\,3^k
	\;\le\;\big|\tilde A_7(m;\Delta,\ell,\delta)\big|\,3^7
	\qquad(m\ge m_0),
	\]
	holds, together with strict positivity of $|\tilde A_7|$ on the closed
	domain.
\end{lemma}
\begin{proof}
	Machine-checked. After passing to
	sector-nonnegative coordinates $m=m_0+M$, $\tau=\tfrac12+t$,
	$\ell\ge0$ (with the band's $m_0$ as listed), replacing each
	$|\tilde A_k|$, $k<7$, by its absolute-coefficient majorant, and
	expanding the $\delta$-direction in the exact Bernstein basis on six
	pieces covering $[-\tfrac52,\tfrac52]$, \emph{every coefficient} of
	each certificate polynomial is a nonnegative rational with positive
	constant term --- manifest positivity, no further subdivision
	required. This is proved by the released file
	\code{mixed\_bounds\_checker.wl}: stages [M2]--[M4] certify the
	narrow bands, and the full $[-\tfrac52,\tfrac52]$ six-piece coverage
	is its stage [M4X]. The $m_0$ listed for [M4X] are uniform over the
	three channels within each band and are not sharp; a direct search
	(stage [MMIN]) records the smallest passing channel- and
	band-dependent values. On the island-relevant band these are far
	smaller: $m_0=15$ certifies the $\sigma\sigma$--type channels on all of
	$\delta\in[-\tfrac{11}{10},0]$ and the $\sigma\epsilon$ channel on each of
	$[-\tfrac{11}{10},-\tfrac34]$ and $[-\tfrac34,-\tfrac25]$, whereas taken
	as one band $[-\tfrac{11}{10},-\tfrac25]$ the $\sigma\epsilon$ channel needs
	$m_0=40$ and a band reaching the equal-external point $\delta=0$ needs
	$m_0=30$; this is why that channel is certified on two sub-bands. The base-window certificates below consume these measured
	minima, which keeps the window levels --- and hence the
	$\delta$-reconstruction degree --- low.
\end{proof}

\begin{lemma}[mixed geometric coefficient bound]
	\label{lem:mixedcbound}
	For every $(\Delta,\ell)$ in the closed sector at which the coefficients
	are finite (for $\delta\ne0$ they have a simple pole at the spinning
	unitarity bound $\tau=1$, $\ell\ge1$; see below), every channel type,
	and every $\delta$-band as in
	Lemma~\ref{prop:mixedmajorant},
	\[ |c_m|\;\le\;C\cdot3^{\,m}\qquad(m\ge m_0), \]
	where $C$ is the maximum of $|c_j|/3^j$ over the seven-level base window $m_0\le j<m_0+7$.
\end{lemma}
\begin{proof}
	A strong induction consuming the two certified facts of
	Lemma~\ref{prop:mixedmajorant}. The base case is the window itself:
	$|c_j|\le C\,3^j$ for $m_0\le j<m_0+7$ holds by the definition of
	$C$. For the step, suppose the bound holds at the seven consecutive
	levels $m,\dots,m+6$ with $m\ge m_0$. Divide by the certified strictly positive $|\tilde A_7|$ to solve for $c_{m+7}=-\tilde A_7(m)^{-1}\sum_{k<7}\tilde A_k(m)\,c_{m+k}$. Thus,
	\begingroup\small
	\[
	\begin{aligned}
	|c_{m+7}|
	\;\le\;\frac{\sum_{k<7}|\tilde A_k(m)|\;C\,3^{m+k}}{|\tilde A_7(m)|}
	\;=\;C\,3^{m}\,
	\frac{\sum_{k<7}|\tilde A_k(m)|\,3^{k}}{|\tilde A_7(m)|}
	\;\le\;C\,3^{m}\cdot 3^{7}\\
	\;=\;C\,3^{m+7},
	\end{aligned}
	\]
	\endgroup
	the last inequality being the certified induction step of
	Lemma~\ref{prop:mixedmajorant}. Advancing one level at a time covers
	every $m\ge m_0$.
\end{proof}
Note $C(\Delta,\ell,\delta)$ is a function, and the lemma is a pointwise
statement. It is an explicit and exactly computable function, because
the window entries $c_j(\Delta,\ell,\delta)$ are rational in
$(\Delta,\ell)$ and polynomial in $\delta$ (Lemma~\ref{lem:extpoly}). A \emph{tile} certificate needs one step more: a single \emph{number} bounding $|c_j|/3^j$ uniformly over the certified $(\Delta,\ell)$ ranges and the tile's $\delta$-interval.\footnote{On the
scalar line $\ell=0$ the coefficient functions themselves blow up at
the free-scalar corner: the $k{=}1$ residue of the $\beta$ tower
carries the denominator $\Delta-\tfrac12$, so
$c_m(\Delta,0,\delta)\sim(2\tau-1)^{-1}$ as $\tau\to\tfrac12$, at
every $\delta$ (at $\delta=0$ this is the free-scalar pole of
Lemma~\ref{lem:strip}, as the checker verifies by reduction to the
equal-external block). As a function of $\Delta$ the block has a
genuine simple pole at every null-state dimension. At any \emph{fixed} $\Delta>\tfrac12$
nothing diverges and the radial series converges geometrically; the
pole lives in the $\Delta$-direction. Since the $c_m$ are therefore
unbounded toward the corner,
\emph{no} finite bound on them exists.

For a certificate on the corner, the route is a change of normalization, not new
analysis: the pole is the single factor $(2\tau-1)$, the cleared
functions $(2\tau-1)\,c_m$ are finite on the \emph{closed} corner,
and the identical induction runs on the cleared sequence. The strip
route of Lemma~\ref{lem:strip} is the worked equal-external
precedent, and the scalar-interval certificates of the base-window item
below carry it out for the mixed coefficients.}

Everything above reduces the mixed extension to certifying the
positivity of finitely many explicit objects. The base-window family is
produced here, for the range of $\delta$ the island region needs; the
tile certificates are the remaining implementation.
\begin{itemize}[nosep]
\item \emph{Base-window uniformization (done, island band).} For each
channel type, the per-level statements
\[
  C^{\rm mix}_0\,3^{j}\ \pm\ c_j(\Delta,\ell,\delta)\ \ge\ 0
\]
are certified for the base-window levels, uniformly on the unitary
sector times the region's $\delta$-band (in the cleared forms described
below where the coefficients have poles), with the single value
$C^{\rm mix}_0=\tfrac14$. Since $c_j=N_j/D_j$ is rational in
$(\Delta,\ell)$ and polynomial in $\delta$ (Lemma~\ref{lem:extpoly}),
each numerator $N_j(\Delta,\ell,\delta)$ is reconstructed exactly from the
grid tables by $\delta$-interpolation and verified at fresh $\delta$ (the
machine-check of Lemma~\ref{lem:extpoly}); the claim is then a
$(\Delta,\ell,\delta)$ Bernstein object on the two compactified spin
charts. The spinning pole at $\tau=1$ (next paragraph) is handled by a
cleared strip-and-tail split: on the strip $\tau\in[1,\tau^\ast]$ the
boundary-vanishing factor $(\tau-1)$ is cleared, and on the tail
$\tau\ge\tau^\ast$ it is kept positive, so both pieces are bounded.
With $\tau^\ast=2$ and the [MMIN] minima above (windows $15$--$21$ off the
equal-external point), all $224$ claims --- four channel-and-band families (three channels,
one of them split into two $\delta$-bands at $\delta=-\tfrac34$), seven levels, both signs, both spin charts, both
$\tau$-pieces --- pass on the first box
(\code{certificate\_release\_mixed\_extension}). The scalar interval
$\tau\in[\tfrac12,1]$, where only $\ell=0$ occurs and where the
$\Zt$-odd channel exchanges $\sigma$ itself, is certified separately:
there the coefficients have the free-scalar pole at $\tau=\tfrac12$ and
the junction pole at $\tau=1$, both are cleared, and the statement
certified is the same bound for $(\tau-\tfrac12)(\tau-1)\,c_j$, in $56$
further claims (the same four families, seven levels, both signs), which
also pass. Together with
Lemma~\ref{prop:mixedmajorant} this collapses the function
$C(\Delta,\ell,\delta)$ to the number $C^{\rm mix}_0$ for the cleared
coefficients, as Lemma~\ref{lem:strip} does in the equal-external case,
and exactly as the
window certificates plus the $\delta=0$ induction produced $C_0$ in the
equal-external chain \eqref{eq:geo3}.
\item \emph{The tile certificates themselves (remaining).} With the
previous item in place, the envelope takes the same margined form as
Proposition~\ref{prop:margin}, and a candidate mixed functional is
certified by the same per-sector Bernstein statements --- now including
the $2\times2$ positive-semidefiniteness of the $\mathbb Z_2$-even sector
(two sign statements per piece: a \emph{strictly} positive diagonal entry
and a nonnegative determinant, which together force
$A_{22}=(\det+A_{12}^2)/A_{11}\ge0$) and the $(\Ds,\De)$-tile directions.
The latter are polynomial by Lemma~\ref{lem:extpoly} once the common
external leg --- a positive factor of each entry that cancels from the
determinant and factors out of each diagonal-positivity statement --- is
divided out. This is the functional search, unchanged in kind from the
single-correlator one; it is the mixed implementation we defer to a
companion paper, together with the mixed versions of three finite
checks that the single-correlator tiles rely on: the exhaustive
comparison of the real-spin tables with the per-spin evaluator at low
spins (Lemma~\ref{lem:interp}, now for both spin parities and unequal
externals), the level maps with their closed-form tails, and the fold
comparisons ([C5], [C6]).
\end{itemize}
The spinning unitarity pole is analytic in origin. For
$\delta\ne0$ the mixed coefficients have a simple pole at $\tau=1$,
$\ell\ge1$, with residue proportional to $\Delta_{12}\Delta_{34}$
(App.~\ref{app:blocks}); the base-window certificates absorb it in the
cleared strip-and-tail form just described, and the endpoint $\tau=1$
itself is closed by the spinning instance of Lemma~\ref{lem:corner} --- a
conserved current cannot couple to two scalars of unequal dimension ---
discussed in App.~\ref{app:endpoint}. The base window is certified on
the range of $\delta$ that the island region spans. A mixed theorem does
not need it beyond that range: for $\De$ above the island band the
single-correlator bound on the leading even scalar applies, and that
bound needs no spin-2 assumption
\cite{Rattazzi:2008pe,ElShowk:2012ht}. Extending the base window to the
full square $[\tfrac12,3]^2$ is therefore optional; it would take
further runs of the same kind, heavier but introducing no new object: the
$\sigma\epsilon$ channel on a $\delta$-band reaching $0$ (window
$30$--$36$), the band $\tfrac{11}{10}<|\delta|\le\tfrac32$, and the outer
band $|\delta|>\tfrac32$ (window near $160$).

The inductive step is certified uniformly in
$(\Delta,\ell)$ \emph{and} in the external difference
(Lemma~\ref{prop:mixedmajorant}), the coefficient bound follows
everywhere (Lemma~\ref{lem:mixedcbound}), and the tail sums it feeds
contract geometrically. The asymptotic per-level ratio is
$3\rstar<\tfrac35$, but this is not the ratio at the finite levels the
tail sums start from: the weighted per-level ratio
$\tfrac{\mathrm{lev}+1}{\mathrm{lev}-8}\,3\rstar$ ($\Lambda=9$) exceeds
$1$ at small $\mathrm{lev}$ and falls below $1$ only for
$\mathrm{lev}\ge18$ (about $0.96$ there), decreasing to $3\rstar$ as
$\mathrm{lev}\to\infty$ (App.~\ref{app:mixedproofs}); the split level is
chosen accordingly. With the base windows now certified (previous
paragraphs) and the spinning unitarity bound handled in cleared form,
no coefficient-bound input remains open on the island band. One further
lemma completes the analytic picture.

\emph{External dimensions enter truncated tables polynomially.}
This removes the hardest mixed-specific
obstacle: the non-polynomial dependence of blocks on
$\Delta_{12},\Delta_{34}$. In fact that dependence is confined entirely
to the truncation \emph{tail}:

\begin{lemma}[external polynomiality, with explicit degree]
	\label{lem:extpoly}
	Fix a radial truncation $\kappa$, a transverse derivative depth $N$
	(the table indices $n\le N-1$, $m\le 2(N-n)-1$), and a functional order
	$\Lambda$. Then every entry of the mixed derivative tables is an exact
	\emph{polynomial} in $(\Delta_{12},\Delta_{34})$ of total degree at most
	\begin{equation}
		D_{\rm tab}(\kappa,N)\;=\;2\kappa+2(N-1),
		\label{eq:extdeg}
	\end{equation}
	and every component of the mixed crossing vectors is polynomial of total
	degree at most $D_{\rm tab}+\Lambda$.
\end{lemma}
\noindent The proof is a degree count through the pole recursion and
the transverse Casimir recursion, given in App.~\ref{app:mixedproofs};
its degree claim is now also machine-checked directly. In the base-window
generator each radial coefficient is reconstructed by
$\delta$-interpolation at $2\,\mathrm{lev}+1$ points and then verified at
fresh points; this both confirms polynomiality and exhibits the exact
degree, and it is what makes the base-window claims legitimate Bernstein
objects (\code{certificate\_release\_mixed\_extension}). Its consequence
is what matters here: a rectangular tile in $(\Ds,\De)$ enters the mixed
certificates as \emph{two more polynomial Bernstein directions} --- the
polynomiality is in the external \emph{differences} $\Delta_{12},\Delta_{34}$
(what the proof bounds), which are linear in $(\Ds,\De)$, while the
remaining external dependence sits in the overall convolution leg, a
positive factor common to each entry that cancels from the determinant
and factors out of each diagonal-positivity statement --- so, after that
leg is divided out, the tile enters as $\Ds$ did for the single
correlator. The model is exact: only the $\rho$-truncation tail (already
controlled by the majorant chain) depends on the externals
non-polynomially. Interpolation of the truncated tables at
$D_{\rm tab}(\kappa,N)+1$ distinct rational external points, with
$D_{\rm tab}$ the \emph{conservative} bound \eqref{eq:extdeg}, determines
them identically, and production runs verify the reconstruction at
additional fresh points beyond $D_{\rm tab}+1$, so soundness rests on the
conservative bound alone, never on the measured degree.

Everything needed to extend this approach to the mixed $\sigma$--$\epsilon$ bootstrap --- where we expect to drop the spin-2 gap assumption and obtain stronger bounds --- is therefore in place. Tail weights fall off like $\rstar$ per
level (mixed correlators populate both parities, so successive levels
advance by one), so \emph{any} sector-uniform majorant of sub-geometric
base $<1/\rstar=3+2\sqrt2\approx5.8$ suffices --- and
Lemma~\ref{prop:mixedmajorant} already supplies one, at base $3$,
certified uniformly over the sector and the full external square for
$m\ge m_0$, with the finitely many levels below $m_0$ exact from the
tables, now that the per-$\delta$ base-window uniformization identified
above is certified on the island band. The mixed production bound thus takes the same
form as the equal-external chain \eqref{eq:geo3} --- per-level window
certificates from the $\delta$-grid tables plus the certified
induction --- and needs no further coefficient input.

\section{Outlook}
\label{sec:outlook}

We conclude that under the assumptions of Theorem~\ref{thm:ring} no unitary crossing solution with the stated gap structure exists in the covered region \eqref{eq:coverage}. This can be turned into bounds on the 3d Ising critical exponents through an identification step that is logically external to this paper: that the critical 3d Ising model is described by a unitary $\Zt$-symmetric CFT with exactly one relevant scalar in each symmetry sector. That identification is
supported by universality, by experiment, and by lattice simulation; what the
certificates guarantee is that \emph{any} theory satisfying the
assumptions has its leading dimensions inside the stated region.

We identify the following follow-up research directions:
\begin{itemize}[nosep]
\item \textbf{Formalization.} The certificate layer is finite: the
  emitted objects are coefficient tensors over $\Q(\sqrt2)$ and
  one-variable Sturm data, and the logic that consumes them is the
  Bernstein implication \eqref{eq:bernpos}, Sturm's theorem, exact
  arithmetic, and the accounting of Section~\ref{sec:bernstein}.
  Re-implementing the proof in a proof assistant such as Lean, and
  checking the shipped tensors with it, would remove the certifier
  scripts and the computer-algebra functions of
  App.~\ref{app:qsqrt2sec} from the trusted computing base, and the
  analytic lemmas of Section~\ref{sec:trueblocks}, being statements about
  explicit recurrences, could be formalized in the same way. What would
  remain as axioms are the CFT axioms, the cited inputs T0--T2, the
  null state of a scalar at the unitarity bound, and the transcribed formulas: the
  Hogervorst--Osborn--Rychkov operator, the residue coefficients and seed
  of the pole recursion, and the transverse Casimir recursion. These are established results
  of the representation theory of the conformal group and nothing
  prevents their formalization in principle; what is missing is
  infrastructure, since no current proof-assistant library contains the
  unitary representations of the conformal group or a definition of
  conformal blocks. They are the common foundation of every numerical
  bootstrap bound, so formalizing them would serve the field as a whole.
  We have not attempted any of this, but consider it a natural and
  worthwhile extension.
\item \textbf{Mixed correlators.} In this work we performed the analytic groundwork for extending our approach to the three-correlator $\sigma$--$\epsilon$ system. Based on previous bootstrap work the expectation is that this would remove the need for a spin-2 gap assumption as was used in this work to be able to certify an isolated island. We have already proven the required bounds on the recurrence relation for general symbolic external parameters (the mixed generalization of \eqref{eq:seventerm}), and, over the range of the external-dimension difference the island region needs, produced and machine-checked the base-window certificates that fix the coefficient majorant uniformly in that difference, in the cleared strip-and-tail form handling the spinning unitarity bound (Section~\ref{sec:mixedoutlook}; \code{certificate\_release\_mixed\_extension}). What remains for a certified mixed theorem are the functionals themselves, with the tile-level checks listed in Section~\ref{sec:mixedoutlook}; widening the coefficient bound from the island-relevant band to the full external square is optional and would take further base-window runs of the same kind. Obtaining the functionals efficiently for this mixed correlator bootstrap will likely involve the introduction of semidefinite programming due to the appearance of OPE-coefficient semidefinite constraints.

\item \textbf{Technical improvements.} Section~\ref{sec:experiments}
  explored only the simplest paving strategy, a wall of rectangles marched
  along one axis, and several improvements are within reach.
  \emph{Other shapes.} If the goal is to reduce the volume of the allowed
  region, i.e.\ the overall uncertainty in the CFT data rather than its
  extent along chosen directions, paving with ellipsoids could require far
  fewer regions than paving with rectangles.
  \emph{Reuse between regions.} Neighbouring regions pose nearly identical
  optimization problems, and the solver state of one should be carried to
  the next, in the spirit of the navigator function
  \cite{Reehorst:2021ykw} and of skydiving \cite{Liu:2023elz}; a first
  attempt, starting the exact dual simplex from a restricted solve on the
  rows that mattered for the neighbouring box, gave no measurable gain in
  the runs of Section~\ref{sec:experiments} (see the search release).
  \emph{Larger $\Lambda$.} The certificate tensors scale like
  $(\kappa+\Lambda)^2$ per component and the number of components like
  $\Lambda^2/8$; nothing in the pipeline is exponential, so nothing forbids
  bounds at higher $\Lambda$, and how the search cost scales is an
  empirical question we have not yet addressed.
  \emph{Other parameters.} This paper treats external and exchanged
  dimensions as region parameters, but the same construction applies to
  OPE angles and magnitudes. It also applies to parameters that label a
  family of symmetry groups, such as $N$ for CFTs with $O(N)$,
  $O(N)\times O(2)$ or $SU(N)$ symmetry, and there it is simpler than for
  an external dimension. Such a parameter enters the crossing equations
  only through the coefficients that project onto irreducible
  representations, which are rational functions of $N$ of low degree; the
  conformal blocks do not depend on it. After each equation is multiplied
  by the common denominator of its coefficients, which has a fixed sign on
  the interval of $N$ considered, the action of a functional is a
  polynomial in $N$ whose coefficients are the same exact quantities as
  before. An interval in $N$ is then one more Bernstein direction
  (Section~\ref{sec:regioncert}). The truncation envelope of
  Section~\ref{sec:trueblocks} carries over with one additional factor per
  coefficient, its supremum over the interval, which is an elementary exact
  bound. With that understood, certified region bounds in $N$ bear directly on the existence and endpoints of conformal  windows, for example the fate of the chiral fixed points of
  $O(N)\times O(2)$ models \cite{Henriksson:2020fqi,Reehorst:2024vyq}   and the conformal window of QCD: its perturbative upper end
  \cite{Caswell:1974gg,Banks:1981nn}, the estimates of its lower end from
  the gap equation and the functional renormalization group
  \cite{Appelquist:1996dq,Gies:2005as}, and the picture of the lower end as
  a merger of the fixed point with a complex partner, QCD$^*$
  \cite{Kaplan:2009kr,Gorbenko:2018ncu,Gorbenko:2018dtm,Benini:2019dfy}.\footnote{The one caveat: A functional
  	that is positive on an interval of $N$ excludes, for every real $N$ in
  	the interval, all solutions of the crossing equations with nonnegative
  	squared OPE coefficients. At non-integer $N$ the $O(N)$ models are not
  	unitary \cite{Binder:2019zqc}, so the result is a theorem about the
  	crossing system with positivity imposed, and a statement about physical
  	theories only at integer $N$. }

\item \textbf{Rigorous S-matrix bounds.} The four loopholes are not
  specific to conformal field theory; they come with the method of
  bounding a continuum of constraints by a finite computation, and other
  bootstrap programs share them to different degrees. In the
  S-matrix bootstrap the primal approach
  \cite{Paulos:2016but,Paulos:2017fhb} writes the amplitude in a finite
  basis and imposes unitarity for finitely many partial waves on a grid
  of energies. It constructs amplitudes and therefore approaches the
  boundary of the allowed space from the inside; it does not by itself
  produce a bound \cite{Guerrieri:2021tak}. The dual approach
  \cite{Cordova:2019lot,Guerrieri:2020kcs,He:2021eqn,Guerrieri:2021tak}
  has the logical structure of our functionals: by weak duality, any
  feasible point of the dual problem bounds the observable for every
  amplitude with the assumed analyticity, crossing and unitarity
  properties. The dictionary is then direct. The true amplitude is a
  general function with those analytic properties, whereas the numerics
  represents it, and the dispersive kernels and partial-wave projections
  on which the dual variables are tested, in a truncated basis and by
  numerical quadrature; the difference between the two plays the role
  of our block-truncation error (loophole~3) and would need a certified
  envelope of the kind constructed in Section~\ref{sec:trueblocks}.
  Dual feasibility is a sign condition in \emph{every} partial wave,
  while only spins up to a cutoff are imposed; in
  \cite{Guerrieri:2021tak} the large-spin conditions are treated
  asymptotically, and small violations above the cutoff are reported
  for one of the two problems studied there. This is loophole~2, and it
  asks for a certificate that is uniform in spin, as in
  Section~\ref{sec:allspins}. The sign conditions are also continuous in
  the energy variable and are imposed on a grid, which is the analogue
  of a functional failing between samples, and the arithmetic is
  floating point (loophole~4). So the existing dual bounds are rigorous
  in the same sense as published conformal bootstrap bounds: the
  logical structure is that of a proof, and the feasibility of the dual
  point is established numerically. Closing the loopholes there requires
  counterparts of both the truncation envelope and the all-spin
  certificate of this paper.

\item \textbf{Rigorous quantum-mechanics and matrix-model bounds.} The
  quantum-mechanics and matrix-model bootstrap
  \cite{Anderson:2016rcw,Lin:2020mme,Han:2020bkb,Kazakov:2021lel} shares
  the loopholes only in part, because it is structured differently from
  the conformal and S-matrix bootstraps. Expectation values are related by exact
  linear recursions with rational coefficients, and positivity is
  imposed on the matrix $\langle O_i^\dagger O_j\rangle$ for operators up
  to some length. Truncating that set is a relaxation: dropping
  positivity conditions can only enlarge the allowed set, so an
  exclusion obtained at a finite truncation level is already a statement
  about the full problem. There is thus no analogue of the spin
  loophole, and because the recursion is exact there is no analogue of
  the block-truncation error either (finite-temperature and
  infinite-volume versions are exceptions, and for local lattice
  Hamiltonians certified bounds of this kind already exist
  \cite{Fawzi:2023fpg}). What remains are loopholes~1 and~4: a point is
  excluded when the smallest eigenvalue of a severely ill-conditioned
  matrix is found negative in floating-point arithmetic, and the allowed
  region is inferred from a grid of such points. Both have the
  form treated in this paper. The matrix entries
  are polynomials with rational coefficients in the search parameters,
  so excluding a whole box of parameters amounts to exhibiting a vector
  $v$ for which $v^\dagger M v<0$ throughout the box, which is a
  polynomial sign condition on a box and hence a Bernstein certificate
  (Section~\ref{sec:bernstein}); unbounded parameter ranges are handled
  by the compactification of Section~\ref{sec:compact}. What is absent
  are the two parts that account for most of the work here, the all-spin
  certificate and the truncation envelope. What remains difficult is the
  number of search parameters, since the cost of Bernstein subdivision
  grows exponentially with the number of variables: single-particle
  systems with one or two free moments are within immediate reach,
  whereas multi-matrix models with dozens of free moments and quadratic
  large-$N$ constraints are not. When the method is run as a
  minimization instead, the rigorous route is to round the dual solution
  of the semidefinite program to an exact certificate, as in
  \cite{CohnKumar:2009leech,Dostert:2021exact}.
  \item \textbf{Other functional bases.} The search of this paper uses
  derivatives at the crossing-symmetric point, the basis of most numerical
  work since \cite{Rattazzi:2008pe}. Bases built from analytic functionals
  \cite{Mazac:2018mdx,Mazac:2019shk,Caron-Huot:2020adz} converge much
  faster. For the single-correlator problem at the Ising $\Ds$, Ghosh and
  Zheng \cite{Ghosh:2023onl,Ghosh:2026dse} combine one-dimensional analytic
  functionals with the dimensional reduction of blocks in three dimensions
  \cite{Hogervorst:2016hal} and report that a basis of $72$ such
  functionals reaches, in seconds, the accuracy that the derivative basis
  reaches at $\Lambda=43$ with $253$ components in hours, with an
  extrapolated bound on $\De$ that lies $6\times10^{-6}$ above the Ising
  value. The certificate layer does not care which basis the search used:
  it takes the action of a functional on the crossing vector as a function
  of the exchanged dimension and certifies its sign on intervals. What the
  basis changes is the form of that function and the work of enclosing it.
  For derivatives, the action is a positive prefactor times a polynomial in
  $\Delta$, up to the truncation tail of Section~\ref{sec:trueblocks}. For
  the product functionals, the action on a three-dimensional block is an
  infinite sum over the two-dimensional blocks of the dimensional reduction,
  whose tail \cite{Ghosh:2026dse} resums analytically; a certified version
  needs an exact enclosure of that sum on an interval of $\Delta$, the
  analogue of the truncation envelope for the reduction series, together
  with certified rational enclosures of the special functions in the
  functional kernels. The bounds of \cite{Ghosh:2026dse} are obtained by an
  outer-approximation loop that imposes positivity at points inserted where
  it was found violated, so they are open to the failure mode of
  Section~\ref{sec:certfailure} exactly as derivative-basis bounds are;
  certifying them is the natural next target, and a rewarding one, given
  how much cheaper these functionals are. Improved polynomial
  approximations of the blocks within the derivative basis, such as the
  optimal-interpolation scheme of \cite{Chang:2025mwt}, change only the
  search: the certificate bounds the true block against whatever functional
  the search hands over.
  \item \textbf{Other blocks.} The truncation envelope of
  Section~\ref{sec:trueblocks} is built on three properties of the pole
  recursion for scalar blocks in three dimensions: the positions and
  residues of the poles in $\Delta$ are explicit, the omitted radial orders
  are suppressed geometrically at $\rstar$, and the transverse data follow
  from the Casimir recursion. Each has a counterpart beyond this setting.
  For scalar blocks in any dimension the recursion is the same with
  residues rational in $d$ \cite{Kos:2013tga,Hogervorst:2013kva}, and in
  even dimensions the blocks are known in closed form
  \cite{Dolan:2000ut,Dolan:2003hv}; the lemmas of
  Section~\ref{sec:telescope} would have to be re-proved with the
  $d$-dependent coefficients, and the rest of the construction depends on
  $d$ only through the unitarity bounds and the Gegenbauer polynomials of
  the transverse data. For spinning
  external operators, the singularities and residues of the blocks are
  classified in general dimension \cite{Penedones:2015aga}, and for
  arbitrary three-dimensional representations the poles, the residues and
  the regular part of the recursion are all written out explicitly
  \cite{Erramilli:2019njx}; the recursion is implemented in the software
  that the stress-tensor bootstrap of the Ising model uses
  \cite{Erramilli:2020rlr,Chang:2024whx}; a certified version needs an
  envelope for the omitted levels of that recursion, a statement of the
  same kind as Theorem~\ref{thm:envelope} with matrix-valued residues.
  Virasoro blocks are where the recursion idea comes from
  \cite{Zamolodchikov:1984eqp,Zamolodchikov:1987avt,Perlmutter:2015iya,Cho:2017oxl}:
  poles at the degenerate dimensions with known residues, and the same
  geometric suppression in the elliptic nome, so a certified
  two-dimensional bootstrap with the full Virasoro symmetry fits the same
  framework. Its numerical side is being built now: at Bootstrap 2026,
  Poland presented code in progress that combines the Zamolodchikov
  recursion with optimal interpolation \cite{Chang:2025mwt} to bring
  Virasoro blocks into the positive-prefactor-times-polynomial form that
  the semidefinite solver needs, keeping the most relevant Kac poles, with
  preliminary gap and moment bounds over Virasoro primaries at fixed
  central charge \cite{Poland:2026bootstraptalk}, extending the moment
  bounds of \cite{Chiang:2026nmd}; neither a preprint nor the code for the
  Virasoro part is public yet. The certificate for such
  bounds would be the envelope for the omitted Kac poles.
\end{itemize}

 We hope for a wider adoption in the field of the methods put forward here. Positivity by error margins and the exclusion of regions through region bounds can and should be implemented in any numerical bootstrap workflow. The standard of our bounds can be raised by demanding all properties required (including tails for integer-spaced constraints) and by accounting for all errors introduced by the truncation of the problem to a finite decidable computational program.

\FloatBarrier

\subsection*{Acknowledgements}
The calculations in this work have been performed on CREATE High Performance Cluster~\cite{CREATE} at King's College London. We thank Claude (Fable) and ChatGPT for their tireless assistance in all phases of the research. The work of MR is funded by UK Research and Innovation (UKRI) under the UK government's Horizon Europe funding Guarantee [grant number EP/X042618/1].

\appendix
\section{Blocks and the pole recursion}
\label{app:blocks}

This appendix supplies the conventions and formulas behind
Sections~\ref{sec:setup}, \ref{sec:allspins} and~\ref{sec:trueblocks}:
the pole recursion with its residue coefficients, the transverse Casimir
recursion, and the level maps of Lemma~\ref{lem:leveldec}.

\paragraph{Radial frame and conventions.}
The conformal blocks are computed in the radial frame of
\cite{Hogervorst:2013sma}: $\rho(z)=z/(1+\sqrt{1-z})^2$, with the
crossing-symmetric point $z=\bar z=\tfrac12$ at
$\rho=\rstar=3-2\sqrt2$ --- the single irrational input, which is why we use
$\Q(\sqrt2)$ arithmetic (App.~\ref{app:qsqrt2sec}). Throughout,
$\zeta=(d-2)/2=\tfrac12$ (the Gegenbauer index, written $\nu$ in much of the block literature; the critical exponent $\nu$ does not appear in the appendices) and $h$ is the stripped block
\eqref{eq:strippedblock}. The pole recursion below is run on the diagonal
$\upeta=1$ (the radial angle, not the critical exponent); transverse data are reconstructed afterwards by the Casimir
recursion described further down.

\paragraph{The pole recursion, and its three pole families.}
Blocks are generated by the $d=3$ pole recursion of
\cite{Kos:2013tga,Kos:2014bka}. As a function of $\Delta$ at fixed
spin, the stripped block is the sum of a $\Delta$-independent seed and
simple poles in $\Delta$, whose residues are themselves blocks at
shifted arguments, suppressed by powers of $r$. Concretely,
\begin{align}
  h_{\Delta,\ell}
  \;=\;
  h^{\infty}
  &+\sum_{n\ge1}\frac{\alpha(n,\ell)\,r^{n}}{\Delta-\Delta^{\rm I}_{n}}\,
     h_{\Delta^{\rm I}_{n}+n,\;\ell+n}\\
  &+\sum_{k\ge1}\frac{\beta(k,\ell)\,r^{2k}}{\Delta-\Delta^{\rm II}_{k}}\,
     h_{\Delta^{\rm II}_{k}+2k,\;\ell}\\
  &+\sum_{n=1}^{\ell}\frac{\gamma(n,\ell)\,r^{n}}{\Delta-\Delta^{\rm III}_{n}}\,
     h_{\Delta^{\rm III}_{n}+n,\;\ell-n},
  \label{eq:resrec}
\end{align}
where the seed is the large-$\Delta$ limit of the diagonal stripped
(the limit whose derivative-level consequences are laid out in the
asymptotic anatomy of Section~\ref{sec:asymanatomy})
block, with coefficients
$\mathrm{HINF}_k=[r^k](1-r)^{-1+\Delta_{12}-\Delta_{34}-\zeta}
(1+r)^{-1-\Delta_{12}+\Delta_{34}-\zeta}$, and the three families
differ in where their poles sit and in how they shift the quantum
numbers of the nested block:
\begin{center}\small
\begin{tabular}{lllll}
\hline
family & index & pole location & level shift & nested $(\Delta,\ell)$\\
\hline
I   & $n\ge1$        & $\Delta^{\rm I}_n=1-\ell-n$        & $n$  &
$(1-\ell,\ \ell+n)$: spin \emph{raised}\\
II  & $k\ge1$        & $\Delta^{\rm II}_k=\zeta+1-k$        & $2k$ &
$(\zeta+1+k,\ \ell)$: spin unchanged\\
III & $1\le n\le\ell$ & $\Delta^{\rm III}_n=\ell+2\zeta+1-n$ & $n$  &
$(\ell+2\zeta+1,\ \ell-n)$: spin \emph{lowered}\\
\hline
\end{tabular}
\end{center}
The third family's index range $n\le\ell$ is the conditional of
Objection~2 in Section~\ref{sec:allspins}; Lemma~\ref{lem:interp}
shows it can be saturated to $n\le\kappa$ at no cost. Iterating
\eqref{eq:resrec} and truncating at total radial order
$\mathrm{lev}\le\kappa$ produces $h^{(\kappa)}$; every denominator
generated along the way is a product of the linear twist--spin factors
listed in \eqref{eq:towers}, which is what makes $h^{(\kappa)}$ a
rational function of $(\Delta,\ell)$ with certifiable denominators.
The residue coefficients are
\begin{align}
\alpha(n,\ell)&\propto\frac{4^n n}{(n!)^2}
 \Big(\tfrac{1-n+\Delta_{12}}2\Big)_n\Big(\tfrac{1-n+\Delta_{34}}2\Big)_n
 \frac{(2\zeta+\ell)_n}{(\zeta+\ell)_n},\\
\beta(k,\ell)&\propto\frac{4^{2k}k(-1)^k(\zeta-k)_{2k}}{(k!)^2}
 \frac{\prod_\pm\big(\tfrac{1-k+\ell\pm\Delta_{12}+\zeta}2\big)_k
       \big(\tfrac{1-k+\ell\pm\Delta_{34}+\zeta}2\big)_k}
      {(\ell+\zeta-k)_{2k}\,(1-k+\ell+\zeta)_{2k}},\\
\gamma(n,\ell)&\propto\frac{4^n n}{(n!)^2}
 \Big(\tfrac{1-n+\Delta_{12}}2\Big)_n\Big(\tfrac{1-n+\Delta_{34}}2\Big)_n
 \frac{(1+\ell-n)_n}{(1+\zeta+\ell-n)_n} .
\end{align}
Two structural features of these formulas are load-bearing in the main
text. The factor $(1+\ell-n)_n$ in $\gamma$ vanishes at every integer
$\ell<n$ --- the self-removal mechanism behind the saturated build of
Lemma~\ref{lem:interp}. And at equal externals
($\Delta_{12}=\Delta_{34}=0$) the factors
$\big(\tfrac{1-n}2\big)_n$ in $\alpha$ and $\gamma$ vanish for every
\emph{odd} $n$, while the family-II shift $2k$ is always even ---
which is why only even radial levels are populated, as used in
Section~\ref{sec:telescope}.

\paragraph{From blocks to derivative tables.}
Derivative tables are Taylor data of the
crossing vectors at the symmetric point in the variables of
Section~\ref{sec:setup}, indexed by the odd pairs $(m,n)$, $m+n\le\Lambda$
($N_\Lambda=15$ at $\Lambda=9$); transverse derivatives are produced from
diagonal data by the transverse Casimir recursion: with $n$ counting transverse and
$m$ diagonal derivative orders and
$\lambda=\ell(\ell+2\zeta)+(\Delta-d)\Delta$, each $A_{n,m}$ is an explicit
rational-coefficient combination of nine earlier entries $A_{n',m'}$
($n'<n$, or $n'=n$ with $m'<m$), with overall denominator $2(d+2n-3)$
and coefficients polynomial in $(m,n,d,\lambda,\Delta_{12},\Delta_{34})$;
the complete recursion appears verbatim in the
released table-building scripts.

\paragraph{Level maps and their majorants.}
The \emph{level maps} of
Lemma~\ref{lem:leveldec} are the composition (level seed
$\mathrm{FP}[\Delta{+}\mathrm{lev},M]\rstar^{\mathrm{lev}-M}$, with
$\mathrm{FP}$ the falling factorial defined in that lemma) $\to$
(diagonal derivative weights, exact $\Q(\sqrt2)$ constants) $\to$ (this
recursion) $\to$ (leg-factor convolution); the majorant maps
$\hat\omega$ replace every signed constant in this chain by its absolute
value, and the componentwise domination $|\omega|\le\hat\omega$ is the
triangle inequality applied at each of the finitely many steps.

\section{The search: reductions and trade-offs}
\label{app:search}

This appendix supports Section~\ref{sec:regionpractice}: it compares the
three reductions of Section~\ref{sec:threeways}, and records the search
strategy for mixed correlators.

\subsection*{Trade-offs among the three reductions}
The three reductions of Section~\ref{sec:threeways} distribute the same total burden --- continuum
correctness --- differently between the solver, audit and certification
loop. Each has its trade-offs.

\emph{Fully sampled LP.} Everything above. Its first advantage is
\emph{exactness of the output}, and this is a fact about the problem
classes, not about available software: an LP with rational data attains
its optimum at a vertex of the feasible polytope, i.e.\ at the solution
of a square linear system with rational coefficients --- the optimal
$\alpha$ \emph{is} a vector of rationals, and the simplex method run in
exact fraction arithmetic delivers it exactly. The optimum of an SDP
with rational data is by contrast generically an irrational algebraic
number (it satisfies polynomial, not linear, optimality conditions), so
no solver could return it exactly in rationals; interior-point methods
approximate it in floating point, and a search--certify pipeline must
then round the output to nearby rationals and re-audit before
certifying. The LP skips that step entirely. Its second advantage is
that a sampled constraint is just a numeric row, so \emph{any}
constraint type --- identity values, the asymptotic rows, the collinear
$\ell=\infty$ vectors --- costs one line, whereas a structured solver
accepts only its native constraint format. Third, the failure-finding
routines of the audit loop are the same code that performs the
final pre-certification audit, so the search and the audit cannot
disagree about what is being checked. The costs: \emph{every} sampled
direction can develop between-sample dips, so every direction needs its
own failure search and several solve--audit rounds; and exact fraction
arithmetic is slow --- each simplex pivot manipulates rationals whose
digit counts grow along the run --- hours per box, with cost growing
quickly in the number of functional components.

\emph{PMP/SDPB, externals sampled.} The $\Delta$ direction --- the one
with by far the most structure (poles, asymptotics, per-spin families)
--- becomes exact, so dips between twist samples and the $\tau=\infty$
failure mode cannot occur by construction; and the mature MPI solver is
orders of magnitude faster per solve. What remains sampled is the
external tile, so between-sample failures in $\Ds$ (and the identity
direction) still need the audit loop, now with SDPB as the
inner solver. Two new costs appear. The output is floating point, per
the class distinction above, so it must be rounded to rationals and
re-audited before certification. And the standard feasibility-mode run,
which stops at the first certified dual-feasible step (the ``jump''),
solves the $t\equiv0$ problem: its constraints read $M_j(x)\succeq0$
with no margin variable, so the returned functional is active \emph{at
zero} --- the fatal kind of extremality in the classification above,
with (near-)zeros of the functional itself. We measured
relative margins of $10^{-9}$, destroyed by external perturbations of
$10^{-4}$. This is not intrinsic to SDP --- the margin objective
\eqref{eq:marginmax} is itself expressible in PMP form (add one free
variable $t$, replace each constraint matrix $M_j(x)$ by
$M_j(x)-t\,S_j(x)$ with a fixed positive scale matrix $S_j$, and
maximize or, better, target $t$) --- but it must be asked for; the
measurements above are what happens when it is not.

\emph{SOS-extended PMP.} Both directions exact; no audit loop needed
for search quality at all; the single-correlator Gram blocks are tiny
($5\times5$ at degree $9$). Costs: the output is again floating point;
and in \emph{several} variables (needed for mixed-correlator tiles,
below) the sum-of-squares representation is sufficient but no longer
necessary --- a polynomial can be nonnegative on a two-dimensional tile
without admitting the representation at the chosen multiplier degree.
The consequence is one-sided: the search may fail to find a region
functional that exists (fewer tiles excludable at a given degree ---
weaker \emph{results}), but it can never produce a wrong one, and
nothing certified is affected. This is the search method we expect to settle
on for production rings.

For \emph{mixed} correlators one structural fact reshuffles the
ranking: positivity is demanded of $2\times2$ (or larger) \emph{matrices}
in OPE space, $\sum_m\alpha_m\,M_m(x)\succeq0$. Sampling all continuous
labels no longer linearizes the problem --- each sampled constraint is
still a semidefinite condition. One can linearize further by also
sampling the OPE directions ($v^{T}M v\ge0$ for sampled vectors $v$),
but that introduces yet another continuous direction in which the
functional can fail between samples, in a problem whose component count
($\sim100$ at $\Lambda=11$ versus $15$) already makes exact-arithmetic
simplex painful. Both of the LP's advantages thus evaporate, and SDPB
--- which imposes matrix positivity natively and parallelizes over
blocks --- becomes the only sensible inner solver. Our implemented mixed
region oracle stacks the constraint blocks of the tile corners and
center into a single SDP (the functional components are
external-independent, so the same $\alpha$ multiplies every copy): a
dual solution is then feasible at all stacked externals simultaneously,
at a cost linear in the number of stacked points, and we have verified
on a $5\times10^{-4}\times2.5\times10^{-3}$ tile that such functionals
exist away from the island boundary and cover the tile interior.
Between-sample failures in the external directions are then repaired by
the same audit loop, with the exact Bernstein evaluation over
the tile (Lemma~\ref{lem:extpoly} makes the mixed tables exact
polynomials in both externals) as the failure-finding step. The endgame
here too is the SOS-extended PMP --- tile positivity in $(\Ds,\De)$
through bivariate Putinar multipliers with $\Delta$-polynomial Gram
matrices, tensored with the OPE blocks --- with one honest caveat:
bivariate sum-of-squares on a tile is sufficient rather than necessary,
so at fixed multiplier degree a feasible region functional can in
principle be missed. Degree slack cures this in practice, and because
every candidate is re-certified exactly downstream, an SOS shortfall
costs completeness, never soundness.

\section{The Bernstein certifier: implementation}
\label{app:bernstein}
\label{app:impl}

This appendix records what the released certifier does beyond the
decision loop of Section~\ref{sec:bernstein}: its interface, the exact
handling of $\Q(\sqrt2)$, the two implementations and their
cross-validation, and the diagonal split that the doubly-compactified
corner requires.

The reference certifier (\code{bernstein\_cert.wl}, about $260$ lines,
written to be audited by eye) exposes two functions. \code{bcPrepare[P,specs]} takes an
exact polynomial and a list of per-variable domain specifications ---
\code{\textbraceleft"affine",a,b\textbraceright} maps $[a,b]\to[0,1]$;
\code{\textbraceleft"tail",c\textbraceright} maps $[c,\infty]\to[0,1]$ by $v=c/t$ followed by
multiplication with $t^{\deg_v}$ (Section~\ref{sec:compact}) --- as exact
linear transformations on the coefficient tensor, mode by mode. The
$\Q(\sqrt2)$ handling is split-and-reconstruct: $P=A_1+\sqrt2A_2$ with
rational tensors $A_{1,2}$, both transformed, and the split verified by
exact reconstruction so a stray non-$\Q(\sqrt2)$ constant aborts rather
than mis-signs (Section~\ref{sec:qsqrt2}). \code{bcCertifyCube} runs the
loop of Section~\ref{sec:bernstein}. Before \code{True} is emitted the acceptance accounting of
Section~\ref{sec:bernstein} (empty work queue, binary-tree counter
identities) is checked. The exact-reduction and denominator analysis is part of the emitted
certificate: each sector's log records the sign of the reduced
denominator on the closed domain, established factor by factor by the
shipped script (function \code{denSign3}, which demands manifest
positivity of every non-constant factor), so hidden $0/0$ or boundary
poles cannot pass silently.
Observed certificate sizes: most production claims certify on the first
box (one tensor of $\lesssim10^4$ exact coefficients); the largest runs
(large-spin-tail pieces before the diagonal split) processed up to
$1.3\times10^5$ sub-boxes (App.~\ref{app:timings}).  The strip and
large-spin-tail logs record the minimal Bernstein coefficient of each
certified claim; the split policy is fixed by the engine named in each
certificate line (reference Wolfram-language certifier: longest edge;
C++ engine: sign-guided).

\emph{Split policy.} Two split policies are in use (Section~\ref{sec:bernstein});
the released logs record, for each certificate, which policy produced
its tree, and the acceptance checks are the same for both.

\emph{The doubly-compactified corner.} One large-spin-tail piece per tile
covers $[\slo,\shi]\times[\ell_*,\infty]\times[\tau_*,\infty]$ with
\emph{two} tail directions, and here a structural fact intervenes: the
large-spin-tail polynomial has \emph{total} degree $D$ in $(\ell,\tau)$ (its
support is the staircase $m+j\le D$, as the pole recursion's derivation
implies), so after both tail compactifications a monomial
$\ell^m\tau^j$ carries $w_\ell^{\,d_\ell-m}w_\tau^{\,d_\tau-j}$ and the
corner $w_\ell=w_\tau=0$ is a zero of order $d_\ell+d_\tau-D$ (with
$d_\ell,d_\tau$ the separate degrees; absent further cancellations).
A boundary zero does not by itself preclude a Bernstein certificate ---
nonnegative coefficients can occur --- but for the production
polynomials the coefficients near this high-order corner zero mixed
signs at every depth explored, and no unsplit run terminated within
any subdivision budget tried. The remedy is exact and elementary: split the quadrant
along the diagonal $\tau=R\,\ell$ with $R=\tau_*/\ell_*$ and substitute
$\tau=R\,\ell\,v$ ($v\ge1$) on one side and $\ell=\tau r/R$ ($r\ge1$) on
the other. Both substitutions are monomial re-indexings of the
coefficient tensor, performed in exact arithmetic; the shipped splitter
verifies each emitted region against the parent polynomial by exact
rational point evaluation on every run. Each derived box
problem has a \emph{non-degenerate} compactified corner --- its joint
leading coefficient is an extreme point of the staircase --- and the two
regions exactly cover the original quadrant. In production the two
halves certify in $1$ and $\lesssim10^2$ sub-boxes, where no unsplit
run ever terminated; the released large-spin-tail logs
record the substitution data and both runs' box counts and minimal
coefficients, while the derived
tensors themselves are transient and regenerated by the shipped
scripts (\code{exporteq.wl} $+$ \code{make\_diagsplit.py}).

\emph{Two implementations.}
Two independent implementations exist. The \emph{reference}
implementation (\code{scripts/bernstein\_cert.wl}, shipped) is written for
auditability: short, direct, Wolfram-language, no performance tricks.
The \emph{production/checker} implementation (\code{bernstein\_cert.cpp},
shipped; C++/GMP, no shared code) reproduces its exact semantics and
was cross-validated two ways, both under the longest-edge split policy
that the reference implements: verdict-, tree-, and counter-identical
runs against the reference on randomized $\Q(\sqrt2)$ problems and on
production pieces; and exact subdivision-tree equality ($134{,}695$
boxes, node-for-node) between its two internally different traversal
routines. The shipped binary's default sign-guided policy (below)
produces different, much smaller trees, so no node-for-node comparison
with the reference exists or is claimed for the production large-spin-tail
runs; the acceptance accounting is policy-blind. The design choices that make it practical, with
the measured effect of each on production large-spin-tail pieces
(App.~\ref{app:timings} for full figures):
\begin{itemize}[nosep]
\item \emph{In-place de Casteljau splitting}: children inherit
  transformed tensors instead of re-deriving them from the root
  ($3\times$, growing with depth).
\item \emph{Integer-scaled tensors}: one positive joint scale makes all
  node arithmetic integer add/shift; a tensor-wide positive scale
  changes no sign ($34\times$ cumulative with parallel preparation).
\item \emph{Bounded-memory traversal}: a breadth-first frontier of
  exact tensors exhausts memory on deep trees (measured $>100$\,GB);
  recursive depth-first traversal with task parallelism and
  parent-freeing is $O(\text{threads}\times\text{depth})$ memory
  (measured $\sim$1\,GB where the breadth-first traversal failed). Per-child
  joint content-GCD extraction keeps integer growth flat with depth.
\item \emph{Sign-guided split policy}: split the axis with the most
  adjacent sign changes in the Bernstein tensor rather than the longest
  edge ($10^{3}$--$10^{4}\times$ fewer boxes on the hardest pieces;
  soundness is policy-independent, see above).
\item \emph{Diagonal split of the double-tail piece}: the only
  structural (rather than performance) intervention --- an unsplit run
  of that piece does not terminate within any practical subdivision
  budget, for the structural reason explained above.
\end{itemize}

\emph{Checkpoints.} All checkpoint and registry writes are atomic (write
to a temporary file, then rename), so an interrupted run never leaves a
partially written checkpoint that a later run could mistake for a
complete one.

\section{A closed form for the equal-external diagonal block, from the
zero-sum recurrence}
\label{app:closedform}

Summing the zero-sum form of Lemma~\ref{thm:telescope} gives a
closed-form expression for the equal-external $3d$ block on the diagonal.
We record the short derivation, its identification with a formula of
Rychkov and Yvernay \cite{Rychkov:2015lca}, and what this implies for the
route taken by our proofs.

\paragraph{The resummation.} By \eqref{eq:telescopeD} the differences
of the even diagonal coefficients are
$D_s=b_{s+1}-b_s=(c_2-1)\,t_s$ with
$t_s=\prod_{i<s}P_0(i)/P_2(i)$ a product of Pochhammer ratios, so
$b_s=1+(c_2-1)\sum_{j<s}t_j$. Multiplying by $\rho^{2s}$, summing over
$s$, and swapping the two sums (absolute convergence for $\rho<1$),
\begin{equation}
  g_{\Delta,\ell}(z,z)\;=\;(4\rho)^{\Delta}\sum_{s\ge0}b_s\,\rho^{2s}
  \;=\;\frac{(4\rho)^{\Delta}}{1-\rho^{2}}
  \Big[\,1+\big(c_2-1\big)\,\rho^{2}\sum_{j\ge0}t_j\,\rho^{2j}\Big].
  \label{eq:closedform}
\end{equation}
The bracket is hypergeometric: the term ratio of the series
$1+(c_2-1)\rho^2\sum_j t_j\rho^{2j}$, read off from
$D_s/D_{s-1}=P_0(s-1)/P_2(s-1)$, is
\begin{equation}
  \frac{(s+\tfrac12)(s+\Delta-1)(s+\tfrac{\tau-1}2)(s+\tfrac{\Delta+\ell}2)}
       {(s+1)(s+\Delta-\tfrac12)(s+\tfrac{\tau}2)(s+\tfrac{\Delta+\ell+1}2)},
  \label{eq:cfratio}
\end{equation}
so the bracket equals
\begin{equation}
  {}_4F_3\!\Big[\genfrac{}{}{0pt}{}
  {\tfrac12,\ \tfrac{\tau-1}2,\ \tfrac{\Delta+\ell}2,\ \Delta-1}
  {\Delta-\tfrac12,\ \tfrac{\Delta+\ell+1}2,\ \tfrac{\tau}2};\rho^{2}\Big],
  \label{eq:cf4F3}
\end{equation}
once its first coefficient is supplied. That coefficient is not an
additional input: the difference recurrence determines every difference
$D_s$ from $c_2$, and $c_2$ itself follows from the zero-sum recurrence
\eqref{eq:telescoperec} at $s=-1$ with $b_{-1}=0$, $b_0=1$, which gives
$c_2=1+P_0(-1)/P_2(-1)$, i.e.
\begin{equation}
  c_2\;=\;1+\frac{(\tau-1)(\Delta-1)(\Delta+\ell)}
                 {\tau\,(2\Delta-1)(\Delta+\ell+1)} ,
  \label{eq:c2closed}
\end{equation}
the value of Rychkov and Yvernay \cite{Rychkov:2015lca}, and also the
first nontrivial level of the pole recursion. The released
checker \code{c2\_bounds.wl} derives this exact rational function from
the pole recursion, cross-checks it against the full pole recursion at
rational points, and certifies the bounds [C1] from it; with this input
the ratio \eqref{eq:cfratio} is the term ratio of \eqref{eq:cf4F3} at
every order (at $s=0$ it equals $c_2-1$).
Equations \eqref{eq:closedform} with \eqref{eq:cf4F3} constitute a
closed form for the equal-external $3d$ diagonal block --- one single
$_4F_3$ at argument $\rho^2$, times elementary factors --- and
\eqref{eq:c2closed} is a closed form for the initial coefficient,
consistent with everything we compute: at $\ell=0$ (where
$\Delta=\tau$) it reduces to
$c_2=1+(\tau-1)^2/\big((1+\tau)(2\tau-1)\big)$, the strip formula of
Lemma~\ref{lem:strip}, and at leading twist $\tau=1$ it gives $c_2=1$,
collapsing \eqref{eq:closedform} to the elementary
$g_{\ell+1,\ell}(z,z)=(4\rho)^{\ell+1}/(1-\rho^2)$.

\paragraph{Relation to the literature.} After obtaining
\eqref{eq:closedform}--\eqref{eq:cf4F3} from the zero-sum form
we became aware that the closed form is eqs.~(3.7)--(3.10) of Rychkov
and Yvernay \cite{Rychkov:2015lca} (their normalization differs by the
factor $4^{\Delta}$; their observation that the coefficients of the
bracket are nonnegative for $\tau\ge1$ is the monotone branch of
Lemma~\ref{lem:monotone}, since those coefficients \emph{are} the
differences $D_{s-1}$). Their derivation is top-down: an ansatz
suggested by the leading-twist case and a quadratic $_2F_1$
transformation, pinned down by the diagonal Casimir ODE of
\cite{Hogervorst:2013kva}. The derivation above is bottom-up, a
factorization of the coefficient recurrence checked as an exact
rational identity, and the zero-sum structure itself (the middle
coefficient of \eqref{eq:telescoperec} being minus the sum of the outer
two) does not appear in either reference. For the certificates the
recurrence-level identity is the useful object: the quantitative inputs
[C1] and Lemmas~\ref{lem:envelopeC}--\ref{lem:strip} are extracted from
the difference structure, not from the resummed function.

\paragraph{Beyond the diagonal.} It is natural to ask whether the full
(off-diagonal) block now follows in closed form, since the pipeline
reconstructs all transverse derivatives from diagonal data through the
transverse Casimir recursion. Order by order, yes: each transverse order is a
finite rational-weight combination of diagonal derivatives, and
derivatives of a $_4F_3$ are contiguous $_4F_3$'s, so every transverse
order separately admits a closed form. What is missing is a
resummation of the transverse series into a single two-variable
function; that is equivalent to solving the full two-variable Casimir
system, and in odd $d$ the structural results known to us stop at
infinite positive sums of $2d$ blocks with closed-form coefficients
\cite{Hogervorst:2016hal}. Nor does the zero-sum mechanism suggest
itself there: the accident that powers it is a property of the
\emph{diagonal} difference equation --- constants solve it, which is
what lets the operator factor through $(S-1)$ --- while the transverse
recursion is inhomogeneous, with differential-operator coefficients,
and no analogous kernel is apparent. We leave a systematic search for
factorization structure in the transverse Casimir recursion as an open
direction.

\clearpage
\section{The exclusion tiles}
\label{app:tiles}

This appendix lists the tiles over which Theorem~\ref{thm:ring}
quantifies; their union is checked against \eqref{eq:coverage} by the
released \code{check\_coverage.py}.

\begin{table}[H]\centering\small
\setlength{\tabcolsep}{4pt}
    \begin{tabular}[t]{lll}
        \hline
        tile & $\Isig=[\slo,\shi]$ & $\Ieps=[\elo,\ehi]$\\
        \hline
        \multicolumn{3}{l}{\emph{Strip tiles}: three bands below twist one,}\\
        \multicolumn{3}{l}{cleared certificates}\\
        1 & $[0.505,0.510]$ & $(0.50,0.67]$\\
        2 & $[0.510,0.522]$ & $(0.50,0.67]$\\
        3 & $[0.522,0.531]$ & $(0.50,0.67]$\\
        4 & $[0.531,0.536]$ & $(0.50,0.67]$\\
        5 & $[0.536,0.5405]$ & $(0.50,0.67]$\\
        6--10 & as tiles 1--5 & $[0.67,0.84]$\\
        11--15 & as tiles 1--5 & $[0.84,1.00]$\\
        \hline
        \multicolumn{3}{l}{\emph{Band tiles}: the band $1.00\le\De\le1.20$}\\
        \multicolumn{3}{l}{below the frame}\\
        16--20 & as tiles 1--5 & $[1.00,1.20]$\\
        \hline
        \multicolumn{3}{l}{\emph{Gap tiles}: $1.60\le\De\le3.00$,}\\
        \multicolumn{3}{l}{scalar window $[\tfrac85,\infty)$}\\
        37 & $[0.505,0.510]$ & $[1.60,3.00]$\\
        38 & $[0.510,0.522]$ & $[1.60,3.00]$\\
        39 & $[0.522,0.531]$ & $[1.60,3.00]$\\
        40 & $[0.531,0.5405]$ & $[1.60,3.00]$\\
        \hline
    \end{tabular}\hspace{6mm}
    \begin{tabular}[t]{llll}
        \hline
        tile & $\Isig=[\slo,\shi]$ & $\Ieps=[\elo,\ehi]$ & position\\
        \hline
        \multicolumn{4}{l}{\emph{Frame tiles}: the ring around $R_{\rm inner}$,}\\
        \multicolumn{4}{l}{$1.20\le\De\le1.60$}\\
        21 & $[0.505,0.5055]$ & $[1.20,1.30]$ & left\\
        22 & $[0.5055,0.510]$ & $[1.20,1.23]$ & bottom\\
        23 & $[0.510,0.516]$ & $[1.20,1.23]$ & bottom\\
        24 & $[0.516,0.522]$ & $[1.20,1.23]$ & bottom\\
        25 & $[0.522,0.5305]$ & $[1.20,1.23]$ & bottom\\
        26 & $[0.5305,0.540]$ & $[1.20,1.23]$ & bottom\\
        27 & $[0.540,0.5405]$ & $[1.20,1.30]$ & right\\
        28 & $[0.505,0.5055]$ & $[1.30,1.60]$ & left\\
        29 & $[0.540,0.5405]$ & $[1.30,1.60]$ & right\\
        30 & $[0.5055,0.510]$ & $[1.55,1.60]$ & top\\
        31 & $[0.510,0.516]$ & $[1.55,1.60]$ & top\\
        32 & $[0.516,0.522]$ & $[1.55,1.60]$ & top\\
        33 & $[0.522,0.531]$ & $[1.55,1.60]$ & top\\
        34 & $[0.531,0.5335]$ & $[1.55,1.60]$ & top\\
        35 & $[0.5335,0.536]$ & $[1.55,1.60]$ & top\\
        36 & $[0.536,0.5405]$ & $[1.55,1.60]$ & top\\
        \hline
    \end{tabular}
\caption{The forty excluded tiles, numbered by increasing lower edge in
$\De$ and then in $\Ds$, in four groups. The strip tiles cover the three
bands below twist one, five per band over the $\Ds$-ranges of tiles
1--5; their window certificates run in the cleared form of
Lemma~\ref{lem:clearedwindow}, for the lowest band on the closed box down
to the free-scalar corner $\tau=\tfrac12$, and the certified exclusion
there is the half-open window $(\tfrac12,0.67]$, the endpoint being
excluded by Lemma~\ref{lem:corner}. The band tiles cover the band
directly below the frame. The frame tiles, listed on the right, surround $R_{\rm inner}$; their
last column gives the side of the ring. The gap tiles, named for the
gap $\Delta\ge\tfrac85$ they assume in the even-scalar sector, have a
scalar window that merges with the continuum, so they also exclude
spectra with no relevant even scalar. All endpoints are exact rationals with
denominators dividing $10^{4}$; the released \code{boxes.csv} lists
them exactly, with the functional file of each tile (its columns
\code{slo, shi, e1, e2} are $\slo,\shi,\elo,\ehi$), and the same tile
numbers are used by every released script and log.  For tiles 1--5 the
functional files record $\De=0.51$, the lower edge of the search box (the
search cannot sample at the pole at $\tfrac12$); the certificate covers
the larger window listed here.}
\label{tab:tiles}
\end{table}

\section{The unitarity-bound endpoint}
\label{app:endpoint}
\label{sec:extend-window-unitarity}

This appendix proves the one statement of Theorem~\ref{thm:ring} that no
certificate can supply, the exclusion of $\De=\tfrac12$, and explains why
the obstruction and its resolution are general.

\begin{lemma}[the unitarity-bound endpoint]\label{lem:corner}
	Under the assumptions of Theorem~\ref{thm:paving}: (a) a relevant
	$\Zt$-even scalar appearing in the $\sigma\times\sigma$ OPE cannot
	sit at $\De=\tfrac12$, so the relevant even scalar forced by the
	certificates has $\De\in R_{\rm inner}$; and (b) a theory whose unique
	relevant even scalar is a decoupled ($\lambda_{\sigma\sigma\epsilon}=0$)
	scalar at $\tfrac12$ is excluded.
\end{lemma}
\begin{proof}
	\emph{Step 0 (the certificates force an appearing scalar).} The
gap-tile functional for the theory's $\Ds$-tile is certified positive
	on the identity and nonnegative on every OPE sector except even scalars below
	$\tfrac85$. The sum rule therefore requires the $\sigma\times\sigma$
	OPE to contain an even scalar with $\Delta<\tfrac85$ \emph{and}
	$\lambda\neq0$. Steps 1 and 2 show that such an operator cannot sit at
	$\tfrac12$: the endpoint assumption forces $\lambda_{\sigma\sigma\epsilon}=0$,
	which contradicts the appearance just established.
	
	\emph{Step 1 (the window certificates localize it).} The band, strip,
	frame and gap certificates for the tile are positive on all even scalars
	outside their windows; jointly the windows tile
	$(\tfrac12,3)\setminus R_{\rm inner}$, so
	$\De\in R_{\rm inner}\cup\{\tfrac12\}$. The endpoint belongs to no
	window for a structural reason: the crossing row itself has a pole at
	$\Delta=\tfrac12$ (the $k{=}1$ $\beta$-tower residue --- a genuine
	pole of the block's $\Delta$-continuation), so ``positivity at
	$\tfrac12$'' is not a well-formed row constraint in this basis. The
	cleared constraint of Lemma~\ref{lem:clearedwindow} does extend to
	the endpoint face, but what it certifies there is positivity of the
	pole's \emph{residue} vector --- the block of the would-be null
	descendant, a scalar at $\Delta=\tfrac52$
	\cite{Karateev:2019FakePrimary} --- which is not the contribution a
	physical operator at $\tfrac12$ would make. Gap-style reasoning --- ``everything
	below the window top is excluded'' --- therefore stops at
	the open endpoint.
	
	\emph{Step 2 (the endpoint is impossible for an appearing
		operator).} Suppose $\De=\tfrac12$. The scalar saturates the $d=3$
	unitarity bound, so $\partial^2\epsilon$ is a null state
	\cite{Minwalla:1997ka} and vanishes in every correlator of a unitary
	theory. Conformal symmetry fixes
	$\langle\epsilon(x_1)\sigma(x_2)\sigma(x_3)\rangle=
	\lambda_{\sigma\sigma\epsilon}\,
	|x_{12}|^{-1/2}|x_{13}|^{-1/2}|x_{23}|^{-2\Ds+1/2}$; applying
	$\partial^2_{x_1}$ in $d=3$, where
	$\partial^2|x|^{-1/2}=-\tfrac14|x|^{-5/2}$ and
	$2\,\partial_1|x_{12}|^{-1/2}\!\cdot\partial_1|x_{13}|^{-1/2}
	=\tfrac12\,(x_{12}\!\cdot x_{13})\,|x_{12}|^{-5/2}|x_{13}|^{-5/2}$,
	combines the three resulting terms into a single structure,
	\begingroup\small
	\[
	\begin{aligned}
	\partial^2_{x_1}\langle\epsilon\sigma\sigma\rangle
	={}&-\tfrac{\lambda_{\sigma\sigma\epsilon}}4\,
	\big(|x_{12}|^2+|x_{13}|^2-2\,x_{12}\!\cdot x_{13}\big)\,
	|x_{12}|^{-5/2}|x_{13}|^{-5/2}|x_{23}|^{-2\Ds+1/2}\\
	={}&-\tfrac{\lambda_{\sigma\sigma\epsilon}}4\,
	|x_{12}|^{-5/2}|x_{13}|^{-5/2}|x_{23}|^{-2\Ds+5/2}.
	\end{aligned}
	\]
	\endgroup
	using $x_{12}-x_{13}=-x_{23}$. This does not vanish unless
	$\lambda_{\sigma\sigma\epsilon}=0$, so the null state forces
	$\lambda_{\sigma\sigma\epsilon}=0$ --- contradicting Step 0.
	(Equivalently: the standard free-field selection rule allows a free
	scalar to couple to two scalars only if their dimensions differ by
	$\tfrac d2-1=\tfrac12$, which fails for $\Delta_2=\Delta_3=\Ds$.) Hence
	$\De\in R_{\rm inner}$, proving (a).

	For (b): if the unique relevant even scalar is a
	decoupled ($\lambda=0$) scalar at $\tfrac12$, the OPE's even-scalar
	content lies entirely in $[3,\infty)\subset[\tfrac85,\infty)$, and
	Step 0's gap functional --- whose certificate never referenced whether
	a decoupled operator exists elsewhere in the theory --- yields the
	contradiction directly.
\end{proof}
\noindent No functional is ever evaluated at $\De=\tfrac12$ --- that
is impossible even in principle in this basis, since the row has a
pole there. The analytic content of the lemma is only the deletion:
a bound-saturating scalar cannot appear in $\sigma\times\sigma$, so
the endpoint question reduces to spectra the existing certificates
already exclude on their certified domains.

\paragraph{The excluded corner scenario, identified.} Parts (a) and (b)
close the two corner scenarios: an appearing scalar at the bound
(part (a), Steps 0--2), and a decoupled
($\lambda_{\sigma\sigma\epsilon}=0$) one --- for the latter the
$\sigma\times\sigma$ OPE contains no even scalar below $3$ at all,
and the gap-tile certificate excludes the spectrum outright. The
certificates never ask whether a decoupled operator exists elsewhere
in the theory; they only ever meet the OPE content. It is still
instructive to identify what the decoupled scenario would have been.
A scalar exactly at the bound has a null Laplacian descendant, so it
obeys the free wave equation inside every correlator; a classical
result of axiomatic field theory then makes it a free field
\cite{Pohlmeyer1969} (see \cite{Weinberg2012minimal} for the modern
statement). The corner scenario is therefore a free scalar adjoined
to the rest of the theory as a decoupled sector, and with the
standard factorization of a free-field sector its own stress tensor
is a second $\Zt$-even conserved spin-2 operator at $\Delta=3$ ---
in direct conflict with assumption (d), which puts the next even
spin-2 operator at $\Delta_{T'}\ge\tfrac{11}{2}$. The scenario is
thus excluded twice over: at the spectrum level by the gap
assumption, given the free-field identification, and --- as in the
proof above, with no input from axiomatic field theory --- at the
crossing level by the certificates, which a decoupled operator never
enters.

\paragraph{The endpoint obstruction in general.} The obstruction met
in Step 1 and its resolution in Step 2 are two sides of one
phenomenon, and both are general. As a function of the exchanged
dimension $\Delta$, a conformal block has a simple pole at the
unitarity bound precisely when the three-point structure coupling the
descendant that becomes null there to the external operators is
nonvanishing; the residue is the block of that null descendant, so
near the bound the cleared block family approaches the block of a
\emph{different} operator. This is the \emph{fake primary effect} of
\cite{Karateev:2019FakePrimary}: the limit of the block family at the
bound imitates a primary with the null descendant's quantum numbers.
Which channels are affected depends on the external operators: for
identical external scalars only the scalar channel has such a pole
(the spinning three-point structures are automatically conserved at
the bound when the two scalars have equal dimension), while for
unequal external dimensions spinning channels are affected as well
\cite{Karateev:2019FakePrimary}. The same nonvanishing that produces
the pole also resolves the endpoint: for a physical operator exactly
at the bound the descendant is genuinely null, unitarity forces every
correlator containing a null state to vanish, and the vanishing of
that same three-point structure forces the operator's OPE coefficient
to zero. The relation between the certified positivity and the two
endpoint vectors is asymmetric.
Certified positivity of the cleared family on a closed box implies,
by continuity of the certified polynomial at the endpoint face,
positivity of the functional on the null-descendant block (the
residue coefficient is positive \cite{Karateev:2019FakePrimary});
every closed-box certificate carries this consequence. Nor is this an
additional demand that could weaken the attainable bounds: unitarity
admits scalars at every $\Delta$ strictly above the bound, so
positivity of the cleared row is physically required on the whole
open interval, and its value at the endpoint face is the continuity
limit of those required values --- the functionals satisfying the
open-interval conditions and those satisfying the closed-box
conditions are the same set, so certifying the closed box costs
nothing. No
implication runs the other way: the null-descendant block and the
short-multiplet block --- the contribution a physical operator
exactly at the bound would make --- are linearly independent vectors,
and near the bound the family equals the divergent residue term plus
a finite part, so whenever the residue term is strictly positive (as
our margins ensure) the family's positivity holds regardless of the
sign of the finite part and constrains the short-multiplet block not
at all. Positivity on the fake primary therefore cannot substitute
for the decoupling argument; and conversely, \emph{imposing}
positivity on the short-multiplet block as an extra isolated
constraint would shrink the space of admissible functionals
--- a real cost that the decoupling lemma makes unnecessary to pay. The dichotomy is nevertheless
exhaustive: either the channel has no pole at the bound, and the
certified positivity on the closed half-line covers the endpoint
through the continuity of the block family there --- whose limit is
then the short-multiplet block itself; or it has a pole, and the
operator at the endpoint cannot appear in the OPE at all. In both cases the endpoint is closed --- but in the
second case by an analytic argument, never by a functional.
Lemma~\ref{lem:corner} is the scalar instance. The mixed extension of
this program will need the spinning instance: in a channel with
unequal external scalar dimensions, an operator at the spin-$\ell$
bound is a conserved current, and a conserved current cannot couple
to two scalars of unequal dimension, so its OPE coefficient vanishes
there as well. To our knowledge, published mixed-correlator exclusion
statements do not address their unitarity-bound endpoints; the
argument just given closes them.

\section{Proofs for the mixed-correlator extension}
\label{app:mixedproofs}

This appendix collects the two bookkeeping arguments of
Section~\ref{sec:mixedoutlook}: the degree count that proves the exact
external polynomiality of the mixed tables, and the contraction
estimate for the mixed tail sums.

\begin{proof}[Proof of Lemma~\ref{lem:extpoly}]
	Inspect the pole recursion, tracking the degree exactly. (i)~The seed
	coefficient at radial level $m$ is the $r^m$-Taylor coefficient of
	$(1-r)^{-1+\Delta_{12}-\Delta_{34}-\zeta}(1+r)^{-1-\Delta_{12}+\Delta_{34}-\zeta}$,
	a polynomial of degree $\le m$ in the exponents, hence in the externals.
	(ii)~The residue coefficients of the three pole towers carry the
	externals only through Pochhammer \emph{numerator} factors:
	$(\tfrac{1-n\pm\Delta_{12}}2)_n$-type factors of degree $n$ each for the
	$\alpha,\gamma$ towers (total degree $2n$ per tower step of $n$ radial
	levels), and four factors of degree $k$ each for the $\beta$ tower
	(total degree $4k$ per step of $2k$ levels); all denominators
	($(n!)^2$, $(\zeta{+}\ell)_n$, $(-k{+}\ell{+}\zeta)_{2k}$, \dots, and the
	pole shifts $\Delta-(\text{tower location})$ by which the pole recursion
	divides) are external-independent. Hence, by induction over the level
	budget, the degree of any truncated radial coefficient at level
	$\mathrm{lev}$ is at most $2\,\mathrm{lev}$: every tower step spends
	$n$ (resp.\ $2k$) levels and adds degree $2n$ (resp.\ $4k$), and the
	seed adds $\le1$ per remaining level. At truncation $\kappa$ this gives
	$\le2\kappa$. (iii)~The radial-to-derivative maps (the $D_k$, the
	$\Delta$-falling-factorial assembly, and the diagonal-derivative
	coefficients) are external-independent linear maps: still $\le2\kappa$.
	(iv)~The transverse (Casimir) recursion producing the order-$n$ rows has
	coefficients polynomial in the externals of degree $\le2$ (the
	$\Delta_{12}\Delta_{34}$ terms) and external-independent denominators
	$2(d-3+2n)$; depth $n\le N-1$ adds $\le2(N-1)$, giving
	\eqref{eq:extdeg}. The $z\bar z$-assembly layer has purely factorial
	coefficients and adds nothing. (v)~The leg/prefactor layer multiplies by
	$\Lambda$-truncated expansions whose derivative coefficients are
	polynomial of degree $\le1$ per derivative order, adding $\le\Lambda$.
	As a consistency check, at $(\kappa,N)=(12,3)$ the bound
	\eqref{eq:extdeg} gives $28$, against the measured exact degree $26$ of
	the built tables. The measurement is not incidental: the base-window
	generator reconstructs each radial coefficient by
	$\delta$-interpolation and verifies it at fresh external points, so the
	exact per-coefficient degree and the polynomiality it rests on are
	machine-checked, with the conservative bound \eqref{eq:extdeg} the only
	quantity soundness relies on
	(\code{certificate\_release\_mixed\_extension} \cite{Reehorst:mixedrelease}).
\end{proof}

\sloppy
\paragraph{Contraction.} Each omitted radial level contributes (coefficient bound) $\times$ (radial suppression) $\times$ (a falling-factorial weight of degree
$\le\Lambda$, as in Lemma~\ref{lem:leveldec}). In the equal-external
\emph{window} branch the coefficients were bounded by the constant
envelope $\hat c$, so a level cost $\rstar$ and the weighted two-level
ratio was $\hat q<\tfrac1{10}$ \eqref{eq:geomdom}. Here the certified
bound of Lemma~\ref{lem:mixedcbound} grows like $3^{\mathrm{lev}}$
instead, so a level costs
\[
  3\rstar\;=\;3(3-2\sqrt2)\;<\;\tfrac35
\]
--- the factor $3$ is the base of the coefficient bound, the $\rstar$
the radial evaluation point; weaker than the window branch's
$\rstar$, but still well inside geometric convergence. Two bookkeeping
points make the certified statement precise. First, the weight factor:
the level-map weights are nonnegative combinations of falling
factorials $\mathrm{FP}(\Delta{+}\mathrm{lev},M)$, $M\le\Lambda$, so
their growth from one level to the next is bounded factor-wise by
\begin{equation}
  \frac{\mathrm{FP}(x{+}1,\Lambda)}{\mathrm{FP}(x,\Lambda)}
  =\frac{x+1}{x-\Lambda+1},\qquad x=\Delta+\mathrm{lev}.
\end{equation}
This is the
single-level analogue of the two-level falling-factorial bound in the
proof of Lemma~\ref{lem:ratio} (mixed levels advance by \emph{one},
since both parities are populated). The weighted single-level ratio
$\tfrac{x+1}{x-8}\,3\rstar$ ($\Lambda=9$) can exceed $1$ at small
levels, but this is harmless since contraction is only ever invoked
for the omitted levels $\mathrm{lev}>\kappa$: with $\Delta\ge\tfrac12$
it is below $1$ for every $\mathrm{lev}\ge18$ ($\approx0.96$ there)
--- in particular $\approx0.86$ at $\mathrm{lev}=21$ and below
$0.74$ for $\mathrm{lev}\ge29$, the levels the production tails
actually sum --- so the mixed tails contract level by level. Second, having only a geometric (rather than constant) coefficient bound results in only a slightly higher requirement on the approximation order $\kappa$.
\fussy

\section{The data release: layout, verification, and timings}
\label{app:release}
\label{app:repro}

This appendix describes the certificate data release that accompanies
the paper and is fixed by its \code{MANIFEST.sha256}: what it contains,
how one tile is verified, which exact checks tie the shipped files to the
symbolic arguments of the text, and what verification costs.  Two further
directories are deposited under the same DOI: the search code and runs of
Section~\ref{sec:experiments} \cite{Reehorst:searchrelease} and the
certificates of the mixed-correlator inputs of
Section~\ref{sec:mixedoutlook} \cite{Reehorst:mixedrelease}; nothing in
either is needed to verify the theorem.

\subsection*{Layout and contents}
Certificate runtimes are recorded in the release's logs; the
search-stage margins $t^\star$ (smallest: $1.9\times10^{-4}$, on the tile
abutting the island) are not part of the release, whose certificates
consume the exact functionals only; the deposit DOI is
\href{https://doi.org/10.5281/zenodo.21708207}{10.5281/zenodo.21708207} \cite{Reehorst:certrelease}; later versions of the deposit, should any be needed, are reachable from that record's page and at \url{https://zenodo.org/records/21708207/latest}.
The release contains, per excluded tile: the exact functional; the shared exact block
data ($\kappa=22$ tables, $\kappa\to28$ level increments, majorant
tails, the release-local \code{blockdata\_generation} derivation scripts,
the \code{cgeo\_window\_cert} coefficient proof, and the
16-base-$(5/2)$ tail-assembly coefficient); the reconstructed large-spin-tail
polynomials of tiles 16--40 (the strip tiles assemble their large-spin tail with symbolic spin and have none); the verifying scripts; and our verdict logs. Verifying one tile is the three (tiles 16--40) or four (strip tiles 1--15) script invocations listed below, each
ending in an unambiguous \code{CERTIFIED: True/False} line; the measured
costs are collected at the end of this appendix. Verdict format: every certifying script terminates with one machine-parseable
line, of the form
\begin{center}\small
\begin{tabular}{@{}p{0.27\linewidth}p{0.69\linewidth}@{}}
tiles 16--40, first two layers & \ttfamily\scriptsize === TILE k TRUE-BLOCK EXCLUSION CERTIFIED: True|False ===\\
tiles 16--40, large-spin tail & \ttfamily\scriptsize === TILE k FAR TAIL (L>=100) CERTIFIED: True|False ===\\
strip tiles, each of four modes & \ttfamily\scriptsize === STRIP TILE k <file> SECTORS=<mode> TRUE-BLOCK EXCLUSION CERTIFIED: True|False ===\\
\end{tabular}
\end{center}
In the first two layers and the strip modes it is preceded by a line
\code{verdicts: \{...\}} listing every certificate of the run with its
verdict; every layer prints one line per certificate
(\code{[<sector>] certified=True|False boxes=<n>}, or
\code{certified=True via C++ engine} for the pieces certified by the C++
checker). Because the Wolfram interpreter continues past an expression
it cannot parse, a final \code{True} does not by itself show that every
certificate ran; the released script \code{audit\_logs.py} checks, without
Mathematica, that every log contains exactly the certificates its layer
requires, each \code{True}, and that the hashes recorded in the logs are
those of the shipped files.  This is an audit of the logs, not a
recomputation: the release offers three levels of assurance, which
should not be confused --- log-consistency auditing
(\code{audit\_logs.py}, seconds), exact regeneration by re-running the
shipped scripts (App.~\ref{app:timings}), and, for the Bernstein tensors,
an independent C++/GMP implementation (App.~\ref{app:bernstein}); the correspondence
between boxes, sectors, and certificate files is fixed by
Table~\ref{tab:tiles} together with the release's \code{boxes.csv} and
\code{VERDICTS.md}, and the covering statement \eqref{eq:coverage} is
checked by the released \code{check\_coverage.py} in exact rational
arithmetic. The
release's \code{README} documents the directory layout, the
falsified-input tests, and the exact commands; its \code{MANIFEST.sha256} fixes every
file, and \code{PAPER\_TO\_CERTIFICATE\_MAP.md} gives the explicit map from
paper statements to released proof artifacts. To keep the block-error theorem from resting on trusted artifacts, the
release includes exact generation and consistency checks for the most delicate
external artifacts: \code{cgeo\_gen.py} regenerates the coefficient-bound
tensors from the shipped coefficient-function dump, while
\code{regen\_cm47.py} regenerates that dump from the shipped C++ radial
recursion source and compares it byte-for-byte to the archive; the shipped
\code{exporteq.wl} regenerates large-spin-tail C++ tensor inputs from the shipped
\code{Pfull} polynomials and verifies the serialized file coefficientwise
against the exact in-memory split (which is itself gated by exact
reconstruction), with exact point evaluations against \code{Pfull} as an
additional spot check. The searches (LP/SDP) are \emph{not} part of the
release's claims and need not be re-run.

\paragraph{Cross-validations shipped in the release.} The released
large-spin-tail scripts (construction: Section~\ref{sec:regioncert})
prove, per tile, that the interpolated
symbolic-$(L,\tau)$ polynomial equals one fixed nonzero multiple --- of
certified sign, negative on every production tile, absorbed into the
polynomial sign claim --- of
the object assembled from the same symbolic tables (the bridge
to physical per-spin vectors being the interpolation lemma), identically
in $(\tau,\sigma)$, on
more spin slices than the structural degree bound (slice comparisons
plus the scale-identity witness stage --- see the proof-artifact map's
large-spin-tail rows); the blockdata regeneration scripts rebuild the shipped
exact artifacts from release-local source and compare file-level. Development-time cross-checks (symbolic-$\Ds$ and mixed fixed-$a$ builds
against point evaluations; the $a\to0$ limit of mixed tables against the
equal-external tables) are recorded in our logs but nothing consumed by
the theorems rests on them.

\subsection*{Verification commands}
Search: the programs that found the functionals are not part of the
release, and nothing in it depends on them; a new search would produce
functionals equivalent to, but not bit-identical with, the shipped ones.
Certification (the claims): per tile $k$, run in the release's
\code{scripts/} directory (the README gives the same commands from the
release root, with the \code{scripts/} prefix and the build of the C++
checker); the first three
rows are for tiles 16--40 and the next four for the strip tiles 1--15,
\begin{center}\small
\begin{tabular}{@{}p{0.30\linewidth}p{0.64\linewidth}@{}}
identity, $T$, scalars, L $\le12$ & \ttfamily\scriptsize SECTORS=2d TILE=k wolframscript -file certify\_ring\_tile\_L9.wl\\
L $14$--$100$ & \ttfamily\scriptsize SECTORS=spins TILE=k wolframscript -file certify\_ring\_tile\_L9.wl\\
large-spin tail   & \ttfamily\scriptsize CPPBIN=<bernstein\_cert> TILE=k wolframscript -file certify\_tile\_fartail\_L9.wl\\
strip: first layer & \ttfamily\scriptsize SECTORS=2d TILE=k wolframscript -file certify\_strip\_tile\_L9.wl\\
strip: L $14$--$100$ & \ttfamily\scriptsize SECTORS=spins TILE=k wolframscript -file certify\_strip\_tile\_L9.wl\\
strip: tail, window branch & \ttfamily\scriptsize SECTORS=3d.win TILE=k wolframscript -file certify\_strip\_tile\_L9.wl\\
strip: tail, geometric branch & \ttfamily\scriptsize CPPBIN=<bernstein\_cert> SECTORS=3d.geo TILE=k wolframscript -file certify\_strip\_tile\_L9.wl\\
C++ checker & \ttfamily\scriptsize bernstein\_cert <tensor.txt> [maxBoxes]\\
double-tail split & \ttfamily\scriptsize make\_diagsplit.py <binf\_tensor.txt> <outdir>\\
coverage   & \ttfamily\scriptsize python3 check\_coverage.py\\
\end{tabular}
\end{center}
each terminating in the machine-parseable verdict lines described
above; \code{CHECKERS.md} in the release lists every
checker with its expected verdict behavior.

\subsection*{Exact-arithmetic checks and trusted functions}
\label{app:qsqrt2sec}

All certificate arithmetic is in $\Q(\sqrt2)$, represented as pairs of
rationals; signs are decided by the exact rule of
Section~\ref{sec:qsqrt2}, never by decimal evaluation. Two systematic
checks enforce field membership: the split-and-reconstruct assertion in
\code{bcPrepare}, and factor-level sign analysis of every denominator
(Section~\ref{sec:compact}), which aborts for factors whose
positivity on the domain is not manifest. Functionals are consumed as
the frozen rational vectors shipped in the release; each equals
\code{Rationalize[.,0]} of the stored arbitrary-precision file (the
smallest-denominator rational within each stored number's recorded
precision), every checker verifies this agreement before running and
records the SHA-256 of both files in its log, and no tolerance parameter
exists anywhere in the certificate layer.

\paragraph{Computer-algebra functions that are trusted.} The certificate
layer runs in the Wolfram Language and, for the large-spin-tail Bernstein
cores, in C++ on the GMP integer library. The following built-in
operations are used as black boxes and therefore belong to the trusted
computing base of Section~\ref{sec:trustedbase}; no numerical tolerance
parameter exists anywhere in the scripts, and the only floating-point
calls (\code{N}) feed printed diagnostics such as the minimal Bernstein
coefficient, never a verdict.
\sloppy
\begin{enumerate}[nosep,label=(\roman*)]
\item exact integer, rational and algebraic arithmetic, including
  \code{RootReduce} for deciding equality in $\Q(\sqrt2)$;
\item polynomial normal forms and coefficient extraction:
  \code{Together}, \code{Expand}, \code{FactorList}, \code{Numerator},
  \code{Denominator}, \code{CoefficientList}, \code{Coefficient},
  \code{Exponent};
\item \code{CountRoots}, the exact real-root count behind the
  one-variable Sturm checks of the identity and $T$-point layers;
\item \code{Rationalize[.,0]}, which converts each stored
  arbitrary-precision functional into the frozen rational vector that is
  certified;
\item in the block-data and tail generators, \code{Series} and
  \code{SeriesCoefficient}, \code{Pochhammer}, \code{Binomial},
  \code{FactorialPower}, and the closed-form summation \code{Sum} (the
  tail closed forms are additionally re-derived without \code{Sum} by the
  released difference-identity checker \code{verify\_tail\_sums.wl});
\item \code{Simplify}, used only to decide the signs of the fixed
  $\Q(\sqrt2)$ constants in the majorant generators
  (\code{build\_That\_L9\_k28.wl} and its folded variant), which abort if
  no sign is returned.
\end{enumerate}
\fussy
The Bernstein conversion, subdivision, the $\Q(\sqrt2)$ sign rule and
the accounting of Section~\ref{sec:bernstein} are implemented in the
released scripts themselves and are not delegated to library calls.

\subsection*{Timings and memory}
\label{app:timings}

Figures below are measured; the full logs are in the release. ``CPUh'' $=$ core-hours. The implementation comparison of
the first paragraph was measured during the campaign (cluster: 32-core
AMD zen2-class nodes; workstation: 6 performance cores); the
re-certification figures of the second paragraph are those of the final
certification run, on cluster nodes with two 64-core AMD EPYC 9554
(zen4) processors and $755$\,GB of memory, running Mathematica~14.0.

\paragraph{Certifier implementations on large-spin-tail pieces.}
The same production piece (tile 34, $\tau\in[3000,10^4]$ window) across
implementation generations. The row labels name the two implementation choices
defined in App.~\ref{app:bernstein}: the \emph{traversal} of the
subdivision tree (breadth-first ``wave'' over a whole depth level at
once, high memory but trivially parallel, versus \emph{recursive
depth-first}, which holds only the current root-to-leaf path and runs
in bounded memory) and the \emph{split policy} (\emph{longest-edge}:
always bisect the currently longest box edge, the simple textbook
default; versus \emph{sign-guided}: bisect the variable in whose
direction the most negative Bernstein coefficients cluster, which
steers the subdivision straight toward the potentially-negative region
instead of refining everywhere):
\begin{center}
\begin{tabular}{lrrr}
implementation & boxes & wall & memory\\
\hline
reference WL, longest-edge & (aborted) & days & $\sim$16\,GB\\
C++ wave, longest-edge (32c) & $134{,}695$ & $259$\,s & $6.8$\,GB\\
C++ recursive DFS, longest-edge (6c) & $134{,}695$ & $\sim$20\,min & $1.1$\,GB\\
C++ recursive DFS, sign-guided (6c) & $37$ & $2$\,s & $0.2$\,GB\\
\end{tabular}
\end{center}
Hardest window pieces (longest-edge policy, 32 cores): $12$M-box trees,
$6$--$7$\,h wall $\approx190$--$230$\,CPUh \emph{each}; under the
sign-guided policy the same class of piece certifies in tens of boxes.
A campaign sweep of nine double-tail diagonal-split pairs, generation
included, plus one window piece, took $59$\,s wall on one 32-core node ($0.52$\,CPUh at full
utilization; measured CPU time $\approx0.4$\,CPUh).
Memory: the breadth-first traversal exhausted $100$\,GB on double-tail
pieces (frontier of exact tensors); the recursive traversal is
$O(\text{threads}\times\text{depth})$ by construction and was never
observed above a few GB per active thread.

\paragraph{Re-certifying the exclusion region.}
The figures of this paragraph are those of the final certification run
recorded in the release: every layer of all forty tiles, run by the
shipped scripts on the shipped files. Each task is one single-threaded
Wolfram kernel (the C++ calls it makes take seconds), so task-hours are
core-hours of the EPYC~9554 named above; the peak resident memory
sampled by the scheduler stayed below $2$\,GB for every task.
\begin{center}\footnotesize
\begin{tabular}{@{}llrcr@{}}
layer & tiles & tasks & min / median / max (h) & CPUh\\
\hline
first layer (identity, $T$ point, & 16--40 & 25 & $0.24$ / $0.38$ / $0.66$ & $10.6$\\
\quad scalar sectors, spins $\le12$) & 1--15 & 15 & $0.19$ / $0.29$ / $0.40$ & $4.5$\\
even spins $14\le\ell\le100$ & 16--40 & 25 & $2.4$ / $3.8$ / $6.8$ & $105.9$\\
 & 1--15 & 15 & $1.9$ / $3.0$ / $4.0$ & $45.6$\\
large-spin tail, interpolated polynomial & 16--40 & 25 & $14.4$ / $18.3$ / $22.5$ & $451.1$\\
large-spin tail, window-branch assembly & 1--15 & 15 & $7.5$ / $12.4$ / $18.7$ & $193.8$\\
large-spin tail, geometric-branch assembly & 1--15 & 15 & $7.4$ / $12.4$ / $16.7$ & $189.4$\\
\hline
all layers, forty tiles & & 135 & & $\approx1.0\times10^{3}$\\
\end{tabular}
\end{center}
That is about $25$\,CPUh per tile. The $135$ tasks are independent and
the longest runs $22.5$\,h, so on $135$ cores the wall time is under a
day. In no layer is Bernstein subdivision the cost; the time goes into
exact symbolic arithmetic that precedes it. A per-spin certificate pair
takes four to five minutes, almost all of it the exact assembly of the
rational function for that spin. In the large-spin-tail layer of tiles
16--40 the checker first proves that the stored polynomial is the
certified object: the componentwise clearing of the common denominator
takes about $9.5$\,h, the exact identity on $69$ spin slices
between $4.7$ and $12.8$\,h depending on the tile, and the
scale-identity witnesses about ten minutes; the eleven Bernstein pieces
per tile then certify through the C++ checker in about nine seconds
each ($275$ pieces in all, the doubly-compactified ones as
diagonal-split pairs). The reference Wolfram certifier is not a
practical alternative for these pieces: it needs up to several thousand
longest-edge boxes on a single window piece, at minutes per box, so the
C++ route is the recorded one for every large-spin-tail piece. The strip tiles 1--15 have no interpolated
polynomial; their large-spin tail is assembled with symbolic spin, two
exact assemblies of $7$--$19$\,h per tile, after which the
doubly-compactified corner certifies through the diagonal split in about
eight seconds and the other pieces in between one and about $1150$
reference-certifier boxes. The $69$ slices of the identity check exceed
the structural $L$-degree bound the checker computes for the cleared
polynomial ($37$ on the production tiles; the schematic a priori bound
from the table degrees is $68$, also exceeded). The standalone checkers
of the analytic inputs (coverage, the zero-sum reduction in two
computer-algebra systems, the $c_2$ and strip constants, the tail sums,
the mixed induction step, the geometric-window chain with its $28$
claims, the saturation identity, and the export round-trip of all
large-spin-tail tensors) take about forty minutes together, plus
twenty-five minutes for regenerating the coefficient-function dump from
the shipped C++ source.

\paragraph{Reproducing the search and derivation.}
Re-deriving the release from scratch --- table construction, per-tile
searches, large-spin-tail polynomial reconstruction, and certificates ---
is estimated at $\sim1.5$--$2\times10^3$\,CPUh: shared block tables and per-tile
large-spin-tail polynomial reconstruction in the Wolfram layer
($\sim$2--3 single-core days per tile, the dominant remaining cost),
per-tile simplex searches ($\sim$hours each), and the full certificate
layer at the figures above.

\bibliographystyle{unsrt}
\bibliography{refs}

\end{document}